\documentclass{amsart}
\usepackage{graphicx}
\usepackage{amssymb}
\usepackage{amsfonts}
\usepackage{float}
\usepackage{placeins}

\newtheorem{theorem}{Theorem}

\newtheorem{lemma}{Lemma}
\newtheorem{notation}{Notation}
\newtheorem{proposition}{Proposition}
\newtheorem{remark}{Remark}
\numberwithin{equation}{section}
\numberwithin{theorem}{section}
\numberwithin{lemma}{section}
\numberwithin{proposition}{section}
\numberwithin{corollary}{section}
\numberwithin{remark}{section}

\begin{document}
\title[Continuous-Time Quantum Markov Chains]{Continuous-Time Quantum Markov Chains And Discretizations Of $p$-Adic Schrödinger Equations: Comparisons And Simulations}

\author[Z\'{u}\~{n}iga-Galindo]{W. A. Z\'{u}\~{n}iga-Galindo}
\address{University of Texas Rio Grande Valley. School of Mathematical \&
	Statistical Sciences. One West University Blvd. Brownsville, TX 78520,
	United States}
\email{wilson.zunigagalindo@utrgv.edu}

\author[Chac\'{o}n-Cort\'{e}s]{L. F. Chac\'{o}n-Cort\'{e}s}
\address{Pontificia Universidad Javeriana, Departamento de Matem\'{a}ticas,
	Cra. 7 N. 40-62, Bogot\'{a} D.C., Colombia}
\email{leonardo.chacon@javeriana.edu.co}

\begin{abstract}%
		The continuous-time quantum walks (CTQWs) are a fundamental tool in the development of quantum algorithms.  Recently, it was shown that discretizations of p-adic Schrödinger equations give rise to continuous-time quantum Markov chains (CTQMCs); this type of Markov chain includes the CTQWs constructed using adjacency matrices of graphs as a particular case. Unlike classical continuous-time Markov chains, the CTQMCs studied here do not satisfy the classical Chapman-Kolmogorov equation; this is a genuine quantum effect, due to the fact that the underlying unitary evolution composes at the level of probability amplitudes rather than probabilities, and we make precise the sense in which the term ``quantum Markov chain'' is used. In this paper, we study a large class of $p$-adic Schrödinger equations and the associated CTQMCs by comparing them with $p$-adic heat equations and the associated continuous-time Markov chains (CTMCs). The comparison is done by a mathematical study of the mentioned equations, which requires, for instance, solving the initial value problems attached to the mentioned equations, and through numerical simulations. We conducted multiple simulations, including numerical approximations of the limiting distribution. Our simulations show that the limiting distribution of quantum Markov chains is greater than the stationary probability of their classical counterparts, for a large class of CTQMCs.
\end{abstract}
	

	\maketitle

	\section{Introduction}
	Quantum computing is a computing paradigm that uses quantum physics effects
	like superposition, interference, and entanglement to achieve computational
	power beyond classical computing; see, e.g., the recent survey \cite{Qian et
		al}, and the references therein. The continuous-time quantum walks (CTQWs) are
	the quantum counterparts of the classical random walks. CTQWs are a fundamental
	tool for developing efficient quantum algorithms, simulating complex quantum
	systems, and potentially achieving quantum advantages in various computational
	tasks. CTQWs offer potential advantages over classical algorithms,
	particularly in graph-based problems like spatial search, \cite{Farhi-Gutman}%
	-\cite{Zimboras et al}. Currently, there is a large variety of quantum
	physical systems that can be used for implementing quantum walks and their
	applications. For instance, trapped-ion and trapped-atom systems, bulk-optics systems, and integrated photonics systems. These considerable advances in  physical implementations of quantum walks are allowing the implementation of specialized quantum walk computing systems for applications of practical
 interest, \cite{Qian et al}, \cite{Manouchehri et al}.

	In the 1930s, Bronstein showed that general relativity and
	quantum mechanics imply that the uncertainty $\Delta x$ of any length
	measurement satisfies $\Delta x\geq L_{\text{Planck}}:=\sqrt{\frac{\hbar
			G}{c^{3}}}$, where $L_{\text{Planck}}\approx10^{-33}$ $cm$\ is the Planck
	length. Bronstein's argument is widely regarded as an early motivation for the
	idea that space-time may not be an infinitely divisible continuum,
	\cite{Bronstein}, although this is a physical interpretation rather than a
	mathematical theorem, and it does not by itself establish the topological
	structure of space-time. One possible mathematical rendering of this idea is
	to model space (or space-time) as a completely disconnected topological
	space. The ultrametric spaces are naturally completely disconnected. The field of $p$-adic numbers $\mathbb{Q}_{p}$ is a
	paramount example of an ultrametric space, with a very rich mathematical
	structure. There are several possible interpretations of the Bronstein
	inequality. One of them drives the loop quantum gravity. Volovich gave another interpretation of Bronstein's inequality in the 80s. The
	inequality mentioned implies that real numbers cannot be used in models at the
	level of Planck's length, because the Archimedean axiom, which appears
	naturally if we use real numbers, implies that lengths can be measured with
	arbitrary precision. Volovich proposed using $p$-adic numbers in physical
	models at the Planck scale, \cite{Volovich}, see also \cite{Dragovich et al}-\cite[Chapter 6]{Varadarajan}.
	
	In the Dirac-von Neumann formulation of quantum mechanics (QM), the states of
	a closed quantum system are vectors of an abstract complex Hilbert space
	$\mathcal{H}$, and the observables correspond to linear self-adjoint operators
	in $\mathcal{H}$. In $p$-adic QM, $\mathcal{H}=L^{2}(\mathbb{Q}_{p}^{N})$ or
	$\mathcal{H}=L^{2}(U)$, where $U\subset\mathbb{Q}_{p}^{N}$. This choice
	implies that the position vectors are $\mathbb{Q}_{p}^{N}$, while the time is
	a real variable. Using $\mathbb{R}\times\mathbb{Q}_{p}^{3}$ as a model of
	space-time, implies assuming that the space ($\mathbb{Q}_{p}^{3}$)\ has a
	discrete nature: there are no continuous word lines joining two different
	points in the space. The assumption of the discreteness of space, which
	requires passing from $\mathbb{R}\times\mathbb{R}^{3}$ to $\mathbb{R}%
	\times\mathbb{Q}_{p}^{3}$, is incompatible
	with special and general relativity. We warn the reader that there are several
	different types of $p$-adic QM; for instance, if the time is assumed to be a
	$p$-adic variable, the QM obtained radically differs from the one considered
	here. To the best of our knowledge, $p$-adic QM started in the 1980s under the
	influence of Vladimirov and Volovich, \cite{V-V-QM3}. The literature about
	$p$-adic QM is pervasive, and here we cited just a few works,
	\cite{V-V-QM3}-\cite{Zuniga-ultimo}.
	
	In \cite{Zuniga-QM-2}, the first author showed that a large class of $p$-adic
	Schr\"{o}dinger equations are the scaling limit of certain continuous-time
	quantum Markov chains (CTQMCs). Practically, a discretization of such an
	equation gives a CTQMC. As a practical result, new types of
	continuous-time quantum walks (CTQWs) on graphs using two symmetric matrices were introduced.
	This construction includes, as a particular case, the CTQWs constructed using
	adjacency matrices. The final goal of the mentioned work was to contribute in
	the understanding of the role of the hypothesis of the discreteness of space in
	the foundations of quantum mechanics. The connection between $p$-adic QM and
	CTQWs show that $p$-adic QM has a physical meaning. $p$-Adic QM is a
	non-local theory because the Hamiltonians used are non-local operators, and
	consequently, spooky actions at a distance are allowed. The paradigm asserting that the universe is not
	locally real implies that $p$-adic QM allows realism.
	
	In this work, we study a large class of $p$-adic Schr\"{o}dinger equations and
	the associated CTQMCs by comparing them with $p$-adic heat equations and the
	associated CTMCs. This work is motivated by \cite{Childs et al}, where a
	particular class of CTQWs was compared with their classical counterparts. Here,
	the comparison is done by a mathematical study of the mentioned equations,
	which requires, for instance, solving the initial value problems attached to
	the mentioned equations, and through numerical simulations. A $p$-adic heat
	equation is a continuous (or scaling) limit of a CTMC. This assertion has been
	verified for some classes of equations, see \cite{Zuniga-QM-2}, which include
	the one considered here. But, for general $p$-adic heat equations, it is an
	open conjecture. By performing a Wick rotation on a heat equation, we obtain
	a Schr\"{o}dinger equation, which is a continuous limit of a CTQMC. For
	general equations, this assertion is an open problem, but for the ones
	considered here, it was established in \cite{Zuniga-QM-2}.
	
	The equations mentioned can be well-approximated by systems of ordinary linear
	differential equations. With the matrices defining these systems, we construct
	random walks. We present a new approach to study CTQWs, which does not use
	adjacency matrices of graphs to construct random walks, but includes this
	construction as a particular case. Typically, the construction of CTQWs is
	based on graphs; we prefer using CTQMCs instead of CTQWs to emphasize that we
	do not use graphs.
	
	We denote by $\mathbb{Z}_{p}$, where $p$ is a prime number, the unit ball of
	$\mathbb{Q}_{p}$, \ which can be visualized as an infinite rooted tree. The $p$-adic heat
	equations on the unit ball considered here have the form%
	\[
	\left\{
	\begin{array}
		[c]{l}%
		\frac{\partial}{\partial t}u\left(  x,t\right)  =J\left(  \left\vert
		x\right\vert _{p}\right)  \ast u\left(  x,t\right)  -u\left(  x,t\right)  =%
		{\displaystyle\int\limits_{\mathbb{Z}_{p}}}
		J\left(  \left\vert x-y\right\vert _{p}\right)  \left\{  u\left(  y,t\right)
		-u\left(  x,t\right)  \right\}  dy\text{, }\\
		\\
		u\left(  x,0\right)  =u_{0}\left(  x\right)  ,
	\end{array}
	\right.
	\]
	where $x\in\mathbb{Z}_{p},t\geq0$, and $J\left(  \left\vert x\right\vert
	_{p}\right)  \geq0$ with $\int_{\mathbb{Z}_{p}}J\left(  \left\vert
	y\right\vert _{p}\right)  dy=1$. Here $dx$ is the Haar measure of $\mathbb{Q}%
	_{p}$, and $\left\vert \cdot\right\vert _{p}$ is the $p$-adic norm. Here, for
	the sake of simplicity, we omit the definition of the function spaces. The
	solution of the above initial problem is $u\left(  x,t\right)  =Z_{0}(x,t)\ast
	u_{0}\left(  x\right)  $, where $Z_{0}(x,t)$ is a heat kernel, which the
	transition density function for a Markov process in $\mathbb{Z}_{p}$, see
	Proposition \ref{Prop_1A}\ and Theorem \ref{Theorem_1}, In Appendix B.
	
	The above equation is a particular example of a $p$-adic heat equation.
	In the general framework, the $p$-adic heat equations are deeply connected with models of complex hierarchical systems; they have been studied
	extensively in the last forty years; see, e.g. \cite{Zuniga-Cambriedge}-\cite{Zuniga-Textbook}, and the references therein.

	Now, we construct a CTMC attached to the above $p$-adic heat equation. We
	decompose the unit ball as a disjoint union of ball of smaller radius:%
	\[
	\mathbb{Z}_{p}=%
	{\displaystyle\bigsqcup\limits_{I\in G_{l}}}
	\left(  I+p^{l}\mathbb{Z}_{p}\right),
	\]
	where  $G_{l}%
	=\mathbb{Z}_{p}/p^{l}\mathbb{Z}_{p}$ is a finite rooted tree with $l+1$ levels. We denote by $\Omega\left(  p^{l}\left\vert x-I\right\vert _{p}\right)  $ the
	characteristic function of the ball $I+p^{l}\mathbb{Z}_{p}$. We now construct
	a CTMC, with states $I\in G_{l}$, by defining the transition probability
	between states as
	\[
	p_{r,v}(t)=p^{l}%
	{\displaystyle\int\limits_{v+p^{l}\mathbb{Z}_{p}}}
	Z_{0}\left(  x,t\right)  \ast\Omega\left(  p^{l}\left\vert x-r\right\vert
	_{p}\right)  dx.
	\]
	Furthermore, $\lim_{t\rightarrow\infty}p_{r,v}(t)=p^{-l}$, which means that
	the Markov chain admits a stationary probability distribution, cf. Theorem
	\ref{Theorem_2}, in Appendix C.
	
	By performing a Wick rotation in the heat equation, we obtain a
	Schr\"{o}dinger equation:%
	\[
	\left\{
	\begin{array}
		[c]{l}%
		\mathrm{i}\frac{\partial}{\partial t}\Psi\left(  x,t\right)  =-J\left(  \left\vert
		x\right\vert _{p}\right)  \ast\Psi\left(  x,t\right)  +\Psi\left(  x,t\right)
		\text{, }x\in\mathbb{Z}_{p},t\geq0\\
		\\
		\Psi\left(  x,0\right)  =\psi_{0}\left(  x\right)  ,
	\end{array}
	\right.
	\]
	where $\mathrm{i}=\sqrt{-1}$. The solution of this equation is given by $\Psi\left(  x,t\right)
	=Z_{0}\left(  x,it\right)  \ast\psi_{0}\left(  x\right)  $, cf. Theorem
	\ref{Theorem_3}, in Appendix C. Now, we construct a quantum Markov chain in
	$G_{l}$ using the following transition probabilities:%
	
	\[
	\pi_{r,v}\left(  t\right)  =\left\vert \left\langle p^{\frac{l}{2}}%
	\Omega\left(  p^{l}\left\vert x-r\right\vert _{p}\right)  ,Z_{0}(x,it)\ast
	p^{\frac{l}{2}}\Omega\left(  p^{l}\left\vert x-v\right\vert _{p}\right)
	\right\rangle \right\vert ^{2},
	\]
	cf. Proposition \ref{Poposition_4}\ and Theorem \ref{Theorem_5}, in Appendix D.
	
	Set $\boldsymbol{J}\varphi=J\ast\varphi-\varphi$, then the operators
	$\boldsymbol{J}$, $\mathrm{e}^{t\boldsymbol{J}}$, $\mathrm{e}^{it\boldsymbol{J}}$ are limits of
	finite dimensional operators $\boldsymbol{J}^{\left(  l\right)  }$,
	$\mathrm{e}^{t\boldsymbol{J}^{\left(  l\right)  }}$, $\mathrm{e}^{it\boldsymbol{J}^{\left(
			l\right)  }}$, and consequently in the computation of $p_{r,v}(t)$,
	respectively of $\pi_{r,v}\left(  t\right)  $, we can use $Z_{0}\left(
	x,t\right)  \ast\Omega\left(  p^{l}\left\vert x-r\right\vert _{p}\right)
	=\mathrm{e}^{t\boldsymbol{J}^{\left(  l\right)  }}\Omega\left(  p^{l}\left\vert
	x-r\right\vert _{p}\right)  $, respectively $Z_{0}(x,it)\ast p^{\frac{l}{2}%
	}\Omega\left(  p^{l}\left\vert x-v\right\vert _{p}\right)
	=\mathrm{e}^{it\boldsymbol{J}^{\left(  l\right)  }}p^{\frac{l}{2}}\Omega\left(
	p^{l}\left\vert x-v\right\vert _{p}\right)  $, and thus the construction of
	the above-mentioned Markov chains depends on a matrix $\boldsymbol{J}^{\left(
		l\right)  }$. By suitably selecting this matrix, we construct the
	CTQWs based on graph adjacency matrices; see Appendix E.
	
	Following  \cite{Childs et al}, \cite{Mulkne-Blumen}, the comparison between CTQMCs and CTMCs is based on the comparison between
	\[
	p_{I}^{\text{sta}}=\lim_{t\rightarrow\infty}p_{I,J}\left(  t\right)  =p^{-l},
	\]
	and the limiting distribution (or longtime average), which is defined as
	\[
	\chi_{I,J}=\lim_{T\rightarrow\infty}\frac{1}{T}%
	{\displaystyle\int\limits_{0}^{T}}
	\pi_{I,J}\left(  t\right)  dt.
	\]

	We obtained a formula for $\pi_{I,J}\left(  t\right)  $ as a $p$-adic
	oscillatory integral; see Theorem \ref{Theorem_5} in Appendix D. But this
	result is not sufficient to compute an explicit formula for $\chi_{I,J}$. \ We
	select a family of kernels $J_{\alpha}\left(  \left\vert x\right\vert
	_{p}\right)  ,\alpha\in\left(  0,\infty\right),$ (the Bessel potentials),
	and study the properties of the CTQMCs and CTMCs via numerical simulations.
	Our simulations show that there exist values $I$ and $J$  such that $\chi_{I,J}>p_{I}^{\text{sta}}$. This fact is
	interpreted as the computational power of the CTQMC is greater than that of
	the corresponding CTMC. All the matrices in our simulations have an
	ultrametric structure, which is natural since we are using $p$-adic numbers.
	The simulations given in \cite{Mulkne-Blumen} show that the matrices $\left[
	\chi_{I,J}\right]  $ have an ultrametric structure supporting the generality
	of our approach.

	Currently, we cannot give specific insights on how our techniques can be used in the development of quantum algorithms. However, as we mentioned in \cite{Zuniga-QM-2}, the connection between $p$-adic QM and CTQMCs opens several new research problems. Here, we pose the question: what type of quantum networks are produced by discretizing non-linear $p$-adic Schr\"{o}dinger equations? In our opinion, this question is deeply connected with the construction of quantum neural networks, see \cite{Schuld et al}-\cite{Zuniga-Galindo-JNMP}.
	
	\section{The road map}

The paper was written for a large audience, including physicists,
mathematicians, and computer scientists. For this reason, the paper was
written so that most of the technical (mathematical) results can be avoided
on a first reading. This work is organized into two parts. The first
discusses the results and simulations; only some mathematical results are
discussed due to their relevance in the construction of CTQWs. But most of
the technical results were placed in appendices at the end of the paper. The
appendices contain a detailed (as much as possible) presentation of the
mathematical results, intended for mathematically inclined readers.

This paper is part of a research effort led by the first author to understand the connection between the foundations of QM from the perspective of the hypothesis that the physical space has a discrete nature
at short distances, \cite{Zuniga-AP}-\cite{Zuniga-ultimo}.  Remarkably, this
investigation is connected with quantum computing.

In \cite{Zuniga-QM-2}, the first author established that the $p$-adic QM is
a non-local theory that admits realism, by showing that a large class of $p$%
-adic Schr\"{o}dinger equations are continuous versions of CTQWs. A very
detailed discussion of the foundations of QM and its connections with $p$%
-adic QM is presented in \cite{Zuniga-QM-2}, and \cite{Zuniga-ultimo}. Here,
we just mention that the connection between the hypothesis that the physical
space has a discrete nature and CTQWs was already discussed in \cite[p. 2]%
{Zuniga-QM-2}: \textquotedblleft Since $p$-adic QM is nonlocal, this theory
admits a reality independent of the observer: quantum particles exist, and
they move randomly in a configuration space of type $\mathbb{R}\times 
\mathcal{X}^{N}$, this means that the quantum motion can be modeled as a
Markov process in $\mathbb{R}\times \mathcal{X}^{N}$. Intuitively, the CTQWs introduced here are approximations of the quantum motion
mentioned.\textquotedblright\ Here, $\mathcal{X}$ is a mathematical model of
a discrete space.

After \cite{Zuniga-QM-2}, a natural question is whether a new theory for
CTQWs based on $p$-adic analysis is possible. This paper constitutes the
first step in this direction. Here, we present a new theory for CTQWs,
including simulations and comparisons. The results are presented in a
mathematically rigorous manner; we aim to provide a short, almost
self-contained paper. We warn the reader that this task is difficult; the
theory of $p$-adic heat equations is already very extensive and technical;
see, for instance, \cite{Zuniga-LNM-2016}. The heat equations used in this
paper were studied in \cite{Zuniga-networks} and \cite{Zuniga-Anselmo}. 

The first author has conjectured that every heat equation is a continuous limit
of a continuous-time random walk, and that the corresponding Schr\"{o}dinger equation is a
continuous limit of a CTQW. The study of each CTQW goes through the study a
particular $p$-adic Schr\"{o}dinger equation and its discretization. Here we
consider only a particular class of heat and Schr\"{o}dinger equations
whose discretizations are easily computed.

There are several different constructions of CTQMCs based on $p$-adic Schr%
\"{o}\-dinger equations. In Sections \ref{Section-CTQMC-I}, \ref%
{Section-CTQMC-II}, we review these constructions and introduce new ones. In
Section \ref{Sec_Num_simulations}, we present several numerical simulations
comparing classical random walks against their quantum counterparts. This
requires selecting a kernel $J$; we use the Bessel potentials $J_{\alpha
}\left( \left\vert x\right\vert _{p}\right) $ because they have been
extensively studied, and many useful properties and formulas are known. The
simulations show the behavior of the transition probability $p_{I,J}(t)$
between state $J$ and $I$ at time $t$, and the corresponding quantum version 
$\pi _{I,J}(t)$. The simulations follow the approach described in \cite%
{Mulkne-Blumen}, \cite{Childs et al}, in which the authors studied a
specific network and compared $p_{I}^{\text{sta}}$ and $\chi _{I,J}$. The
inequality $p_{I}^{\text{sta}}$ $<$ $\chi _{I,J}$ (for some values $I$ and $J$) is interpreted as the fact
that quantum random walks have greater computational power than classical
ones. The rationale behind this assertion is that the entries in the matrix $%
\chi _{I,J}$  are spread out and genuinely depend on both $I$ and $J$, showing
that the walker can perform random walks around each state, while the
classical limiting distribution $p_{I}^{\text{sta}}$ is independent of the
starting state $J$, which means that, in the long-time limit, the classical
walker forgets its initial position, while the quantum walker does not. This phenomenon also occurs in our
simulations.
The last section (Discussion)  provides a summary of the results and the contributions of the present paper in our research program on the foundations of QM, space discreteness, and  CTQWs.

	\section{Preliminaries}
	

	\subsection{The field of $p$-adic numbers}
	
	From now on, we use $p$ to denote a fixed prime number. Any non-zero $p$-adic
	number $x$ has a unique expansion of the form%
	\begin{equation}
		x=x_{-k}p^{-k}+x_{-k+1}p^{-k+1}+\ldots+x_{0}+x_{1}p+\ldots,\text{ }
		\label{p-adic-number}%
	\end{equation}
	with $x_{-k}\neq0$, where $k$ is an integer, and the $x_{j}$s\ are numbers
	from the set $\left\{  0,1,\ldots,p-1\right\}  $. The set of all possible
	sequences of the form (\ref{p-adic-number}) constitutes the field of $p$-adic
	numbers $\mathbb{Q}_{p}$. There are natural field operations, sum and multiplication, on series of form (\ref{p-adic-number}). There is also a norm
	in $\mathbb{Q}_{p}$ defined as $\left\vert x\right\vert _{p}=p^{-ord(x)}$,
	where $ord_{p}(x)=ord(x)=k$, for a nonzero $p$-adic number $x$. By definition
	$ord(0)=\infty$. The field of $p$-adic numbers with the distance induced by
	$\left\vert \cdot\right\vert _{p}$ is a complete ultrametric space. The
	ultrametric property refers to the fact that $\left\vert x-y\right\vert
	_{p}\leq\max\left\{  \left\vert x-z\right\vert _{p},\left\vert z-y\right\vert
	_{p}\right\}  $ for any $x$, $y$, $z$ in $\mathbb{Q}_{p}$. The $p$-adic
	integers are the $p$-adic numbers which are sequences of the form
	(\ref{p-adic-number}) with $-k\geq0$. All
	these sequences constitute the unit ball $\mathbb{Z}_{p}$. The unit ball is an
	infinite rooted tree with fractal structure. As a topological space
	$\mathbb{Q}_{p}$\ is homeomorphic to a Cantor-like subset of the real line,
	see, e.g., \cite{V-V-Z}, \cite{Alberio et al}.
	
	A function $\varphi:\mathbb{Q}_{p}\rightarrow\mathbb{C}$ is called locally
	constant, if for any $a\in\mathbb{Q}_{p}$, there is an integer $l=l(a)$, such
	that
	\[
	\varphi\left(  a+x\right)  =\varphi\left(  a\right)  \text{ for any }%
	|x|_{p}\leq p^{l}.
	\]
	The set of functions for which $l=l\left(  \varphi\right)  $ depends only on
	$\varphi$ form a $\mathbb{C}$-vector space denoted as $\mathcal{U}%
	_{loc}\left(  \mathbb{Q}_{p}\right)  $. We call $l\left(  \varphi\right)  $
	the exponent of local constancy. If $\varphi\in\mathcal{U}_{loc}\left(
	\mathbb{Q}_{p}\right)  $ has compact support, we say that $\varphi$ is a test
	function. We denote by $\mathcal{D}(\mathbb{Q}_{p})$ the complex vector space
	of test functions. There is a natural integration theory so that
	$\int_{\mathbb{Q}_{p}}\varphi\left(  x\right)  dx$ gives a well-defined
	complex number. The measure $dx$ is the Haar measure of $\mathbb{Q}_{p}$. In Appendix A, we give a quick review of the basic aspects of the $p$-adic
	analysis required here.
	
	If $U$ is an open subset of $\mathbb{Q}_{p}$, $\mathcal{D}(U)$ denotes the
	$\mathbb{C}$-vector space of test functions with supports contained in $U$,
	then $\mathcal{D}(U)$ is dense in
	\[
	L^{\rho}\left(  U\right)  =\left\{  \varphi:U\rightarrow\mathbb{C};\left\Vert
	\varphi\right\Vert _{\rho}=\left\{
	{\displaystyle\int\limits_{U}}
	\left\vert \varphi\left(  x\right)  \right\vert ^{\rho}dx\right\}
	^{\frac{1}{\rho}}<\infty\right\}  ,
	\]
	for $1\leq\rho<\infty$, see, e.g., \cite[Section 4.3]{Alberio et al}. In this
	paper, we consider QM in the sense of the Dirac-von Neumann formulation\ on the
	Hilbert space $L^{2}\left(  U\right)  $. Given $f,g\in L^{2}\left(  U\right)
	$, we set
	\[
	\left\langle f,g\right\rangle =%
	{\displaystyle\int\limits_{U}}
	f\left(  x\right)  \overline{g\left(  x\right)  }dx,
	\]
	where the bar denotes the complex conjugate.

	\subsection{Additive characters}
	
	Using that any non-zero $p$-adic number $x$ admits an expansion of the form
	$x=p^{ord(x)}\sum_{j=0}^{\infty}x_{j}p^{j}$, $x_{0}\neq0$, the fractional part
	of\textit{ }$x\in\mathbb{Q}_{p}$, denoted as $\{x\}_{p}$, is the rational number
	defined as
	\[
	\left\{  x\right\}  _{p}=\left\{
	\begin{array}
		[c]{lll}%
		0 & \text{if} & x=0\text{ or }ord(x)\geq0\\
		&  & \\
		p^{ord(x)}\sum_{j=0}^{-ord(x)-1}x_{j}p^{j} & \text{if} & ord(x)<0.
	\end{array}
	\right.
	\]
	Set $\chi_{p}(y)=\exp(2\pi i\{y\}_{p})$ for $y\in\mathbb{Q}_{p}$. The map
	$\chi_{p}(\cdot)$ is an additive character on $\mathbb{Q}_{p}$, i.e., a
	continuous map from $\left(  \mathbb{Q}_{p},+\right)  $ into $S$ (the unit
	circle considered as a multiplicative group) satisfying $\chi_{p}(x_{0}%
	+x_{1})=\chi_{p}(x_{0})\chi_{p}(x_{1})$, $x_{0},x_{1}\in\mathbb{Q}_{p}$. \ The
	additive characters of $\mathbb{Q}_{p}$ form an Abelian group which is
	isomorphic to $\left(  \mathbb{Q}_{p},+\right)  $. The isomorphism is given by
	$\xi\rightarrow\chi_{p}(\xi x)$, see, e.g., \cite[Section 2.3]{Alberio et al}.
	
	\subsection{$p$-Adic QM}
	
	In the Dirac-Von Neumann formulation of QM, to every isolated quantum system
	there is associated a separable complex Hilbert space $\mathcal{H}$ called the
	space of states. The Hilbert space of a composite system is the Hilbert space
	tensor product of the state spaces associated with the component systems. Non-zero vectors describe the states of a quantum system from
	$\mathcal{H}$. Each observable corresponds to a unique linear self-adjoint
	operator in $\mathcal{H}$.  The most important observable of a quantum system
	is its energy. We denote the corresponding operator by $\boldsymbol{H}$. Let
	$\Psi_{0}\in\mathcal{H}$ be the state at time $t=0$ of a certain quantum
	system. Then at time $t$ the system is represented by the vector $\Psi\left(
	t\right)  =\boldsymbol{U}_{t}\Psi_{0}$, where $\boldsymbol{U}_{t}%
	=\mathrm{e}^{-it\boldsymbol{H}}$, $t\geq0$,\ is a unitary operator called the evolution
	operator. The vector function $\Psi\left(  t\right)  $ is differentiable if
	$\Psi\left(  t\right)  $ is contained in the domain $Dom(\boldsymbol{H})$ of
	$\boldsymbol{H}$, which happens\ if at $t=0$, $\Psi_{0}\in Dom(\boldsymbol{H}%
	)$, and in this case the time evolution of $\Psi\left(  t\right)  $ is
	controlled by the Schr\"{o}dinger equation $\mathrm{i}\frac{\partial}{\partial t}%
	\Psi\left(  t\right)  =\boldsymbol{H}\Psi\left(  t\right)  $, where
	$\mathrm{i}=\sqrt{-1}$ and the Planck constant is assumed to be one. For an in-depth
	discussion of QM, the reader may consult \cite{Berezin et al}-\cite{Takhtajan}%
	. In $p$-adic QM, $\mathcal{H=}L^{2}\left(  U\right)  $, where $U$ is an open
	compact subset of $\mathbb{Q}_{p}$. For the sake of simplicity, we use
	$\mathcal{H=}L^{2}\left(  \mathbb{Z}_{p}\right)  $, where $\mathbb{Z}%
	_{p}=\left\{  x\in\mathbb{Q}_{p};\left\vert x\right\vert _{p}\leq1\right\}  $
	is the unit ball of $\mathbb{Q}_{p}$;  $\mathbb{Z}
	_{p}$ is a compact additive group.
	
	\section{CTQMCs and $p$-adic Schr\"{o}dinger equations}
	
	In this section, we discuss two distinct techniques for constructing CTQMCs based on $p$-adic Schr\"{o}dinger equations. Most constructions of CTQWs are based on the use of adjacency matrices of graphs. This approach cannot be used to construct the CTQMCs introduced in this section. On the contrary, most CTQWs based on graphs can be constructed as a discretization of suitable $p$-adic Schr\"{o}dinger equations; see \cite{Zuniga-QM-2}, and
	Appendix E.
	
	\subsection{Initial value problems for $p$-adic
		Schr\"{o}dinger equations on the unit ball}
	
	We set $\ \boldsymbol{H}:\mathcal{D}(\mathbb{Z}_{p})\rightarrow L^{2}%
	(\mathbb{Z}_{p})$ to be a self-adjoint operator defined on the dense subspace
	$\mathcal{D}(\mathbb{Z}_{p})\hookrightarrow L^{2}(\mathbb{Z}_{p})$. Consider the Cauchy problem
	\begin{equation}
		\left\{
		\begin{array}
			[c]{l}%
			\Psi\left(  \cdot,t\right)  \in L^{2}(\mathbb{Z}_{p})\text{, }t>0;\text{ }%
			\Psi\left(  x,\cdot\right)  \in C^{1}(0,\infty)\text{, }x\in\mathbb{Z}_{p};\\
			\\
			\mathrm{i}\frac{\partial}{\partial t}\Psi\left(  x,t\right)  =\boldsymbol{H}\Psi\left(
			x,t\right)  \text{, }x\in\mathbb{Z}_{p},t>0;\\
			\\
			\Psi\left(  x,0\right)  =\Psi_{0}\left(  x\right)  \in L^{2}(\mathbb{Z}%
			_{p})\text{;}\\
			\\
			\left\Vert \Psi\left(  \cdot,t\right)  \right\Vert _{2}=\left\Vert \Psi
			_{0}\right\Vert _{2}=1\text{, }t\geq0.
		\end{array}
		\right.  \label{Eq_3}%
	\end{equation}

	We assume the existence of a unique solution for (\ref{Eq_3}). We denote the
	unitary group attached to $\boldsymbol{H}$ as $\mathrm{e}^{-i\boldsymbol{H}t}$,
	$t\in\mathbb{R}$.
	
	\subsection{An orthonormal basis for $L^{2}(\mathbb{Z}_{p})$}
	
	We set%
	\[
	\mathbb{Q}_{p}/\mathbb{Z}_{p}=\left\{  \sum_{j=-m}^{-1}x_{j}p^{j};\text{for
		some }m>0\right\}  .
	\]
	For $b\in\mathbb{Q}_{p}/\mathbb{Z}_{p}$, $r\in\mathbb{Z}$, we denote by
	$\Omega\left(  \left\vert p^{r}x-b\right\vert _{p}\right)  $ the
	characteristic function of the ball $bp^{-r}+p^{-r}\mathbb{Z}_{p}$.
	
	We now define
	\begin{equation}
		\Psi_{rbk}\left(  x\right)  =p^{\frac{-r}{2}}\chi_{p}(p^{-1}k\left(
		p^{r}x-b\right)  )\Omega\left(  \left\vert p^{r}x-b\right\vert _{p}\right)  ,
		\label{Basis_0}%
	\end{equation}
	where $r\in\mathbb{Z}$, $k\in\{1,\dots,p-1\}$, and $b\in\mathbb{Q}%
	_{p}/\mathbb{Z}_{p}$. Then,
	\begin{equation}%
		{\displaystyle\int\limits_{\mathbb{Q}_{p}}}
		\Psi_{rbk}\left(  x\right)  dx=0, \label{Average}%
	\end{equation}
	and $\left\{  \Psi_{rbk}\left(  x\right)  \right\}  _{rbk}$ forms a complete
	orthonormal basis of $L^{2}(\mathbb{Q}_{p})$; see, e.g., \cite[Theorems 9.4.5
	and 8.9.3]{Alberio et al}, or \cite[Theorem 3.3]{Zuniga-Cambriedge}.
	
	The set of functions
	\[%
	{\displaystyle\bigcup}
	\left\{  \Omega\left(  \left\vert x\right\vert _{p}\right)  \right\}  \text{
		\ \ }%
	{\displaystyle\bigcup}
	\text{ \ }%
	{\displaystyle\bigcup\limits_{k\in\{1,\dots,p-1\}}}
	\text{ \ \ }%
	{\displaystyle\bigcup\limits_{r\leq0}}
	\text{ \ \ }%
	{\displaystyle\bigcup\limits_{\substack{bp^{-r}\in\mathbb{Z}_{p}%
				\\b\in\mathbb{Q}_{p}/\mathbb{Z}_{p}}}}
	\left\{  \Psi_{rbk}\left(  x\right)  \right\}  ,
	\]
	is an orthonormal basis in $L^{2}(\mathbb{Z}_{p})$, \cite[Proposition
	2]{Zuniga-PhysicaA}.
	
	We assume that $\Omega\left(  \left\vert x\right\vert _{p}\right)  $, and the
	$\Psi_{rbk}\left(  x\right)  $ are eigenfunctions of $\boldsymbol{H}$:%
	\[
	\boldsymbol{H}\Omega\left(  \left\vert x\right\vert _{p}\right)
	=E_{\text{gnd}}\Omega\left(  \left\vert x\right\vert _{p}\right)  \text{,
	}\boldsymbol{H}\Psi_{rbk}\left(  x\right)  =E_{rbk}\Psi_{rbk}\left(  x\right)
	\text{ for }r,b,k,
	\]
	where $E_{\text{gnd}}$, $E_{rbk}$ are the corresponding eigenvalues. Then,
	\begin{equation}
		\Psi\left(  x,t\right)  =A_{0}\Omega\left(  \left\vert x\right\vert
		_{p}\right)  \mathrm{e}^{-iE_{\text{gnd}}t}+%
		{\displaystyle\sum\nolimits_{r,b,k}}
		A_{rbk}\Psi_{rbk}\left(  x\right)  \mathrm{e}^{-iE_{rbk}t}, \label{Eq_4A}%
	\end{equation}
	where the complex constants $A_{0}$, $A_{rbk}$ are determined by the condition%
	\begin{equation}
		\Psi\left(  x,0\right)  =\Psi_{0}\left(  x\right)  =A_{0}\Omega\left(
		\left\vert x\right\vert _{p}\right)  +%
		{\displaystyle\sum\nolimits_{r,b,k}}
		A_{rbk}\Psi_{rbk}\left(  x\right)  . \label{Eq_4B}%
	\end{equation}

	\subsection{\label{Section-CTQMC-I}Construction of CTQMCs I}
	
	Let $\mathcal{N}\subset\mathbb{Z}_{p}$ be a subset of measure zero, and let
	$\mathcal{K}_{j}$ $\subset\mathbb{Z}_{p}$ be disjoint, open, compact subsets
	such that
	\begin{equation}
		\mathbb{Z}_{p}\smallsetminus\mathcal{N}=%
		{\displaystyle\bigsqcup\limits_{j\in\mathbb{J}}}
		\mathcal{K}_{j}, \label{partition_K}%
	\end{equation}
	where $\mathbb{J}$ is countable (i.e., $\mathbb{J}$ is either finite or
	countably infinite).
	
	We set $e_{v}=c_{v}1_{\mathcal{K}_{v}}$, where $1_{\mathcal{K}_{v}}$ is the
	characteristic function of $\mathcal{K}_{v}$, and the constant $c_{v}$ is
	chosen so that $\left\Vert e_{v}\right\Vert _{2}=1$. Then, $\left\{
	e_{j}\right\}  _{j\in \mathbb{J}}$ is an orthonormal subset of $L^{2}(\mathbb{Z}_{p})$.
	We assume that
	\begin{equation}
		\mathcal{H}_{\mathbb{J}}:=\mathit{Span}\left\{  e_{j};j\in\mathbb{J}\right\}  \text{ is
			a Hilbert subspace of }L^{2}(\mathbb{Z}_{p}), \label{H1}%
	\end{equation}

	and \ that%
	\begin{equation}
		\mathrm{e}^{-\mathrm{i}\boldsymbol{H}t}\mathcal{H}_{\mathbb{J}}\subset\mathcal{H}_{\mathbb{J}%
		}\text{.} \label{H2}%
	\end{equation}

	The function
	\begin{equation}
		\Psi_{v}\left(  x,t\right)  :=\mathrm{e}^{-\mathrm{i}\boldsymbol{H}t}e_{v}(x)\in L^{2}\left(
		\mathbb{Z}_{p}\right)  , \label{Function_Psi}%
	\end{equation}
	is the solution of the Cauchy problem (\ref{Eq_3}), with $\Psi_{0}\left(
	x\right)  =e_{v}(x)$.
	
	We construct a network with nodes $j\in\mathbb{J}$, where $j$ is identified
	\ with $1_{\mathcal{K}_{j}}$, and define the transition probability between
	nodes $v$ and $r$ as%
	\begin{equation}
		\pi_{r,v}\left(  t\right)  =\left\vert \left\langle e_{r},\mathrm{e}^{-\mathrm{i}\boldsymbol{H}%
			t}e_{v}\right\rangle \right\vert ^{2}=\left\vert \left\langle e_{r},\Psi
		_{v}\right\rangle \right\vert ^{2}, \label{TranProb}%
	\end{equation}
	where $\left\langle \cdot,\cdot\right\rangle $\ is the inner product of
	$L^{2}\left(  \mathbb{Z}_{p}\right)  $. Now, the Cauchy-Schwarz inequality,%
	\[
	0\leq\sqrt{\pi_{r,v}\left(  t\right)  }\leq\left\Vert e_{r}\right\Vert
	_{2}\left\Vert \Psi_{v}\right\Vert _{2}=\left\Vert e_{r}\right\Vert
	_{2}\left\Vert e_{v}\right\Vert _{2}=1,
	\]
	implies that $\pi_{r,v}\left(  t\right)  \in\left[  0,1\right]  $.
	
	\begin{proposition}
		\label{Proposition_1}Under the hypotheses (\ref{H1})\ and (\ref{H2}), $%
		{\textstyle\sum\nolimits_{r\in\mathbb{J}}}
		\pi_{r,v}\left(  t\right)  =1$, for $t\geq0$, and any $v\in\mathbb{J}$.
	\end{proposition}
	
	\begin{proof}
		First,%
		\[
		L^{2}(\mathbb{Z}_{p})=\mathcal{H}_{\mathbb{J}}%
		{\textstyle\bigoplus}
		\mathcal{H}_{\mathbb{J}}^{\perp}\text{,}%
		\]
		where $\mathcal{H}_{\mathbb{J}}^{\perp}$ is the orthogonal complement of
		$\mathcal{H}_{\mathbb{J}}$. Then, for each $t\geq0$,
		\[
		\Psi_{v}\left(  x,t\right)  =%
		{\displaystyle\sum\limits_{r\in\mathbb{J}}}
		c_{r}\left(  t\right)  e_{r}\left(  x\right)  +z_{v}(x,t),
		\]
		where $z_{v}(x,t)\in\mathcal{H}_{\mathbb{J}}^{\perp}$, and
		\begin{equation}
			\left\Vert \Psi_{v}\left(  \cdot,t\right)  \right\Vert _{2}^{2}=%
			{\displaystyle\sum\limits_{r\in\mathbb{J}}}
			\left\vert c_{r}\left(  t\right)  \right\vert ^{2}+\left\Vert z_{v}%
			(\cdot,t)\right\Vert _{2}^{2}=1. \label{Eq_5A}%
		\end{equation}
		On the other hand, by using that $\pi_{r,v}\left(  t\right)  =\left\vert
		\left\langle e_{r}\left(  x\right)  ,\Psi_{v}\left(  x,t\right)  \right\rangle
		\right\vert ^{2}=\left\vert c_{r}\left(  t\right)  \right\vert ^{2}$, and
		(\ref{Eq_5A}),
		\[%
		{\textstyle\sum\nolimits_{r\in\mathbb{J}}}
		\pi_{r,v}\left(  t\right)  =%
		{\textstyle\sum\nolimits_{r\in\mathbb{J}}}
		\left\vert c_{r}\left(  t\right)  \right\vert ^{2}=1\Leftrightarrow\text{
		}z_{v}(x,t)=0\text{ }\Leftrightarrow\Psi_{v}\left(  x,t\right)  \in
		\mathcal{H}_{\mathbb{J}},
		\]
		for any $t\geq0$.
	\end{proof}
	
	In the framework of continuous-time Markov chains, $\pi_{s,r}\left(  t\right)
	$ represents the probability that the chain is in state $s$ at time $t$, given
	it started in state $r$ at time $0$.
	
	\begin{theorem}
		\label{Theorem_0A}Under the hypotheses \ref{H1}\ and \ref{H2}, $\left[
		\pi_{s,r}\left(  t\right)  \right]  _{s,r\in\mathbb{J}}$ is the transition
		matrix of a  continuous-time quantum Markov chain  with state space
		$\mathbb{J}$.
	\end{theorem}
	
\begin{remark}
	\label{Remark_Definitions}We take this opportunity to make precise the
	terminology used throughout the paper. By a continuous-time Markov chain
	(CTMC) with state space $\mathbb{J}$ we mean, as usual, a family of
	transition probabilities $\left[  p_{r,v}(t)\right]  _{r,v\in\mathbb{J}}$,
	$t\geq0$, satisfying the classical Chapman-Kolmogorov equation
	\[
	p_{r,v}(t+s)=\sum_{u\in\mathbb{J}}p_{r,u}(t)\,p_{u,v}(s)\text{, }t,s\geq0.
	\]
	By a continuous-time quantum Markov chain (CTQMC) we mean, as in
	\cite{Zuniga-QM-2} and in Theorem \ref{Theorem_0A} above, the family of
	squared transition amplitudes
    \[\pi_{r,v}(t)=\left\vert \left\langle
	e_{r},\mathrm{e}^{-\mathrm{i}\boldsymbol{H}t}e_{v}\right\rangle \right\vert
	^{2}\]
    attached to a unitary evolution $\mathrm{e}^{-\mathrm{i}%
		\boldsymbol{H}t}$. We emphasize that this usage should be distinguished
	from at least two other established meanings of \textquotedblleft quantum
	Markov chain\textquotedblright\ in the literature, both of which do satisfy
	a genuine composition law: the operator-algebraic quantum Markov chains of
	Accardi, defined via quasi-conditional expectations, and the quantum Markov
	chains obtained from open quantum walks, i.e., trace-preserving completely
	positive (Lindblad/GKLS-type) semigroups $\left(  T_{t}\right)  _{t\geq0}$
	satisfying $T_{t+s}=T_{t}\circ T_{s}$. Neither of these is used in this
	paper. Rather, $\left[  \pi_{r,v}(t)\right]  $ is used in the sense of the
	continuous-time quantum walk (CTQW) literature, \cite{Farhi-Gutman},
	\cite{Childs et al}, \cite{Mulkne-Blumen}, in which the walk is
	\textquotedblleft quantum\textquotedblright\ precisely because it composes
	at the level of amplitudes rather than probabilities; the absence of a
	classical Chapman-Kolmogorov equation for $\left[  \pi_{r,v}(t)\right]  $,
	discussed below, is the expected signature of this convention rather than a
	defect. A continuous-time (quantum) random walk is simply the process on the
	state space $\mathbb{J}$, or on the vertex set of a graph, whose transition
	probabilities are given by a CTMC (respectively CTQMC) of this type,
	typically constructed, as here, from a $p$-adic heat (respectively
	Schr\"{o}dinger) equation.
	
	It is important to note that, in contrast with the classical case, the
	family $\left[  \pi_{r,v}(t)\right]  $ does \emph{not} satisfy the classical
	Chapman-Kolmogorov equation in general. Indeed, writing $\mathrm{e}%
	^{-\mathrm{i}\boldsymbol{H}s}e_{v}=\sum_{u\in\mathbb{J}}c_{u}(s)e_{u}$
	(possible under hypotheses \ref{H1}-\ref{H2}), the unitary group property
	$\mathrm{e}^{-\mathrm{i}\boldsymbol{H}(t+s)}=\mathrm{e}^{-\mathrm{i}%
		\boldsymbol{H}t}\mathrm{e}^{-\mathrm{i}\boldsymbol{H}s}$ gives
	\[
	\pi_{r,v}(t+s)=\left\vert \sum_{u\in\mathbb{J}}c_{u}(s)\left\langle
	e_{r},\mathrm{e}^{-\mathrm{i}\boldsymbol{H}t}e_{u}\right\rangle \right\vert
	^{2},
	\]
	which involves interference between the terms of the sum and is, in
	general, different from $\sum_{u\in\mathbb{J}}\pi_{r,u}(t)\pi_{u,v}(s)$; the
	two coincide only in the absence of interference, e.g. when a single term of
	the sum dominates. This lack of a classical Chapman-Kolmogorov equation for
	$\left[  \pi_{r,v}(t)\right]  $ is not a defect of the construction: it is
	the precise mathematical manifestation of the fact that quantum evolution
	composes at the level of probability amplitudes rather than at the level of
	probabilities, and it is what allows quantum random walks to outperform
	their classical counterparts (cf. the discussion of $p_{I}^{\text{sta}}$
	versus $\chi_{I,J}$ in Section \ref{Sec_Num_simulations}). What Theorem
	\ref{Theorem_0A} does guarantee, and what justifies calling $\left[
	\pi_{r,v}(t)\right]  $ the transition matrix of a Markov chain, is only that
	it is a genuine stochastic matrix for every fixed $t\geq0$ (nonnegative
	entries summing to $1$ along each column), which is what is needed to
	interpret $\pi_{r,v}(t)$ as a transition probability at time $t$ and to
	define, e.g., limiting distributions as in Section
	\ref{Sec_Num_simulations}.
\end{remark}

	We now study the case when hypothesis (\ref{H2})\ is false.
	
	\begin{proposition}
		\label{Proposition_2}Assume that hypothesis (\ref{H1}) is true, and that
		hypothesis (\ref{H2})\ is false. Then, for $t>0$, $%
		{\textstyle\sum\nolimits_{r\in\mathbb{J}}}
		\pi_{r,v}\left(  t\right)  =1$ if and only if $\left\{  e_{j}\right\}  _{j\in
			\mathbb{J}}$ is a complete orthonormal subset (a basis) of $L^{2}(\mathbb{Z}_{p})$.
		Furthermore, if $\mathbb{J}$ is a finite set, then for $t>0$, $%
		{\textstyle\sum\nolimits_{r\in\mathbb{J}}}
		\pi_{r,v}\left(  t\right)  <1$. (At $t=0$ this bound fails to be strict:
		$\Psi_{v}\left(  x,0\right)  =e_{v}(x)\in\mathcal{H}_{\mathbb{J}}$ trivially,
		so $z_{v}(x,0)=0$ and $%
		{\textstyle\sum\nolimits_{r\in\mathbb{J}}}
		\pi_{r,v}\left(  0\right)  =1$, regardless of whether $\left\{  e_{j}\right\}
		_{j\in\mathbb{J}}$ is complete; this degenerate initial-time case is
		excluded from the statement.)
	\end{proposition}
	
	\begin{proof}
		By reasoning as in the proof of the previous proposition, we have%
		\[%
		{\textstyle\sum\nolimits_{r\in\mathbb{J}}}
		\pi_{r,v}\left(  t\right)  =%
		{\textstyle\sum\nolimits_{r\in\mathbb{J}}}
		\left\vert c_{r}\left(  t\right)  \right\vert ^{2}=1\text{ if and only if
		}z_{v}(x,t)=0,
		\]
		for any $t\geq0$. For $t>0$, since hypothesis (\ref{H2}) is false, $\Psi
		_{v}\left(  x,t\right)$ need not remain in $\mathcal{H}_{\mathbb{J}}$, and
		$z_{v}(x,t)=0$ is equivalent to the set $\left\{
		e_{j}\right\}  _{j\in\mathbb{J}}$ being complete in $L^{2}(\mathbb{Z}_{p})$. The
		last assertion follows from the fact that $L^{2}(\mathbb{Z}_{p})$ is an
		infinite-dimensional Hilbert space, so a finite set $\left\{  e_{j}\right\}
		_{j\in\mathbb{J}}$ can never be complete. At $t=0$, by contrast,
		$\Psi_{v}(x,0)=e_{v}(x)\in\mathcal{H}_{\mathbb{J}}$ always, so $z_{v}(x,0)=0$
		trivially and the stated equivalence is vacuous.
	\end{proof}

	With the hypotheses of the above proposition, we assume that $\mathbb{J}$ is a
	finite set. We attach to each $e_{r}$, $r\in\mathbb{J}$, a state, denote as
	$r$, and define a transition probability $\pi_{r,v}\left(  t\right)  $ for
	$r$, $v\in\mathbb{J}$, as in (\ref{TranProb}). We introduce an extra state
	denoted as $\infty$, and define%
	\[
	\pi_{\infty,v}\left(  t\right)  =1-%
	{\textstyle\sum\nolimits_{r\in\mathbb{J}}}
	\pi_{r,v}\left(  t\right)  \text{, for }v\in\mathbb{J},
	\]
	and
	\[
	\pi_{s,\infty}\left(  t\right)  =\left\{
	\begin{array}
		[c]{ccc}%
		0 & \text{if} & s\in\mathbb{J}\\
		&  & \\
		1 & \text{if} & s=\infty,
	\end{array}
	\right.
	\]
	for any $t\geq0$.
	
	\begin{theorem}
		\label{Theorem_0}Assume that hypothesis (\ref{H1}) is true, and that
		$\mathbb{J}$ is a finite set (which means that hypothesis (\ref{H2})\ is
		false). Then, $\left[  \pi_{s,r}\left(  t\right)  \right]  _{s,r\in
			\mathbb{J}\cup\left\{  \infty\right\}  }$ is the transition matrix of a
		continuous-time quantum Markov chain with state space $\mathbb{J}\cup\left\{
		\infty\right\}  $.
	\end{theorem}
	
	The presented construction is entirely different from the one given in
	\cite{Farhi-Gutman}-\cite{Childs et al}, since here we use an
	infinite-dimensional Hilbert space $L^{2}(\mathbb{Z}_{p})$, and we do not use
	graphs. In contrast, in the references mentioned above, the constructions use
	finite-dimensional Hilbert spaces and adjacency matrices. On the other hand,
	we propose to interpret \ the state $\infty$ in the previous theorem as an
	interaction with an external system. We do not present numerical simulations
	for this type of network here. The main difficulty is the computational
	implementation of the condition $\mathrm{e}^{-\mathrm{i}\boldsymbol{H}t}\mathcal{H}_{\mathbb{J}%
	}\nsubseteqq\mathcal{H}_{\mathbb{J}}$.
	
	\subsection{\label{Section-CTQMC-II}Construction of CTQMCs II}
	
	The second construction is based on the Born interpretation of the
	wavefunctions, which is also valid in the $p$-adic framework, \cite{Zuniga-AP}%
	-\cite{Zuniga-PhA}. Let $\mathcal{K}$ be an open compact subset of
	$\mathbb{Z}_{p}$ as before. The function $\left\vert \Psi_{v}\left(
	x,t\right)  \right\vert ^{2}$ is the probability density for the position of a
	particle/walker at the time $t$ given that at $t=0$ the particle/walker was in
	$\mathcal{K}_{v}$. By the Born interpretation
	\[
	\widetilde{\pi}_{r,v}\left(  t\right)  :=%
	{\displaystyle\int\limits_{\mathcal{K}_{r}}}
	\left\vert \Psi_{v}\left(  x,t\right)  \right\vert ^{2}dx
	\]
	is the probability of finding the particle/walker in $\mathcal{K}_{r}$; which
	we interpret as a transition probability from a $\mathcal{K}_{v}$ to
	$\mathcal{K}_{r}$. Now, using partition (\ref{partition_K}), we have
	\begin{align*}
		1  &  =%
		{\displaystyle\int\limits_{\mathcal{K}}}
		\left\vert \Psi_{v}\left(  x,t\right)  \right\vert ^{2}dx=%
		{\displaystyle\int\limits_{\mathcal{K}\smallsetminus\mathcal{N}}}
		\left\vert \Psi_{v}\left(  x,t\right)  \right\vert ^{2}dx\\
		&  =%
		{\displaystyle\sum\limits_{j\in\mathbb{J}}}
		\text{ }%
		{\displaystyle\int\limits_{\mathcal{K}_{j}}}
		\left\vert \Psi_{v}\left(  x,t\right)  \right\vert ^{2}dx=%
		{\displaystyle\sum\limits_{j\in\mathbb{J}}}
		\widetilde{\pi}_{j,v}\left(  t\right)  .
	\end{align*}
	So, the matrix $\left[  \widetilde{\pi}_{r,v}\left(  t\right)  \right]
	_{r,v\in\mathbb{J}}$ defines a quantum Markov chain with space state
	$\mathbb{J}$, where a walker jumps between the sets $\mathcal{K}_{v}$ to
	$\mathcal{K}_{r}$ with a transition probability $\widetilde{\pi}_{r,v}\left(
	t\right)  $.
	
	The computation of the transition probabilities $\widetilde{\pi}_{r,v}\left(
	t\right)  $ is involved, even in dimension one and for simple partitions of
	the unit ball, \cite{Zuniga-Mayes}.
	
	\section{$p$-Adic heat equations and CTMCs}
	
	In this section, we provide a brief discussion of the results in Appendix B,
	assuming a minimal mathematical background without including proofs. In the
	mentioned appendix, a rigorous treatment of the results of this section is given.
	
	\subsection{A class of $p$-adic heat equations}
	
	In \cite{Zuniga-networks}, and the references therein, the first author
	established that CTMCs can be obtained as discretizations of p-adic heat
	equations. Here, based on this reference, we study a class of $p$-adic heat
	equations and the associated CTMCs. The techniques used here are different
	from the ones used in \cite{Zuniga-networks}. We fix a real-valued radial
	function $J(x)=J(\left\vert x\right\vert _{p})$, where $J:\left[  0,1\right]
	\rightarrow\left[  0,\infty\right)  $, satisfying
	\[
	\left\Vert J(\left\vert x\right\vert _{p})\right\Vert _{1}=%
	{\displaystyle\int\limits_{\mathbb{Z}_{p}}}
	J(\left\vert x\right\vert _{p})dx=1.
	\]
	We extend $J(\left\vert x\right\vert _{p})$ as zero outside of the unit ball,
	then $J(\left\vert x\right\vert _{p})\in L^{1}\left(  \mathbb{Q}_{p}\right)
	$, and consequently, the Fourier transform $\widehat{J}(\xi):\mathbb{Q}%
	_{p}\rightarrow\mathbb{C}$ is well-defined. The Fourier transform does not
	preserve the support of $J(\left\vert x\right\vert _{p})$, i.e., $\widehat
	{J}(\xi)\neq0$ for points $\xi$ outside of the unit ball.
	
	We now construct an operator%
	\[%
	\begin{array}
		[c]{cccc}%
		\boldsymbol{J}: & L^{\rho}\left(  \mathbb{Z}_{p}\right)   & \rightarrow &
		L^{\rho}\left(  \mathbb{Z}_{p}\right)  \\
		&  &  & \\
		& \varphi\left(  x\right)   & \rightarrow & J\left(  \left\vert x\right\vert
		_{p}\right)  \ast\varphi\left(  x\right)  -\varphi\left(  x\right)
	\end{array}
	\]
	which is a \ linear, bounded operator, for $\rho\in\left[  0,\infty\right]  $.
	Here, we require only the cases $\rho=1,2$. Furthermore, $\boldsymbol{J}$ is a
	pseudo-differential operator,%
	\[
	\boldsymbol{J}\varphi\left(  x\right)  =\mathcal{F}_{\xi\rightarrow x}%
	^{-1}\left(  \left(  \widehat{J}(\left\vert \xi\right\vert _{p})-1\right)
	\mathcal{F}_{x\rightarrow\xi}\varphi\right)  =%
	{\displaystyle\int\limits_{\mathbb{Q}_{p}}}
	\left(  \widehat{J}(\left\vert \xi\right\vert _{p})-1\right)  \widehat
	{\varphi}\left(  \xi\right)  \chi_{p}\left(  -\xi x\right)  dx,
	\]
	for $\varphi\in\mathcal{D}\left(  \mathbb{Z}_{p}\right)  $, where
	$\mathcal{F}$ denotes the Fourier transform, see Appendix A.
	
	We fix a time horizon $T\in\left(  0,\infty\right]  $, and consider the following initial value problem:%
	\begin{equation}
		\left\{
		\begin{array}
			[c]{l}%
			\frac{\partial}{\partial t}u\left(  x,t\right)  =J\left(  \left\vert
			x\right\vert _{p}\right)  \ast u\left(  x,t\right)  -u\left(  x,t\right)
			\text{, }x\in\mathbb{Z}_{p},t\in\left[  0,T\right]  \\
			\\
			u\left(  x,0\right)  =u_{0}\left(  x\right)  \in\mathcal{D}(\mathbb{Z}_{p}).
		\end{array}
		\right.  \label{System_1}%
	\end{equation}
	By passing to the Fourier transform in (\ref{System_1}),%
	\[
	\left\{
	\begin{array}
		[c]{l}%
		\frac{\partial}{\partial t}\widehat{u}\left(  \xi,t\right)  =\left(
		\widehat{J}(\left\vert \xi\right\vert _{p})-1\right)  \widehat{u}\left(
		\xi,t\right)  \text{, }\xi\in\mathbb{Q}_{p},t\in\left[  0,T\right]  \\
		\\
		\widehat{u}\left(  \xi,0\right)  =\widehat{u}_{0}\left(  \xi\right)
		\in\mathcal{D}(\mathbb{Q}_{p}).
	\end{array}
	\right.
	\]
	Then,$\widehat{\text{ }u}\left(  \xi,t\right)  =\mathrm{e}^{t\left(  \widehat
		{J}(\left\vert \xi\right\vert _{p})-1\right)  }\widehat{u}_{0}\left(
	\xi\right)  $, and setting \ $Z_{0}(x,t)=\mathcal{F}_{\xi\rightarrow x}%
	^{-1}\left(  \mathrm{e}^{t\left(  \widehat{J}(\left\vert \xi\right\vert _{p})-1\right)
	}\right)  $, and taking formally the inverse Fourier transform,
	\begin{align*}
		u\left(  x,t\right)   &  =\mathcal{F}_{\xi\rightarrow x}^{-1}\left(
		\mathrm{e}^{t\left(  \widehat{J}(\left\vert \xi\right\vert _{p})-1\right)  }\widehat
		{u}_{0}\left(  \xi\right)  \right)  =\mathcal{F}_{\xi\rightarrow x}%
		^{-1}\left(  \mathrm{e}^{t\left(  \widehat{J}(\left\vert \xi\right\vert _{p})-1\right)
		}\right)  \ast u_{0}\left(  x\right)  \\
		&  =Z_{0}(x,t)\ast u_{0}\left(  x\right)  .
	\end{align*}
	In Appendix B, using \cite[Theorem 3.1]{Zuniga-networks}, we show that
	$u\left(  x,t\right)  =Z_{0}\left(  x,t\right)  \ast u_{0}\left(  x\right)  $
	is the unique solution of (\ref{System_1}), in a suitable function space, and
	that
	\[
	p(t,x,B)=%
	\begin{cases}
		\int_{B}Z_{0}(x-y,t)dy=1_{B}\left(  x\right)  \ast Z_{0}(x,t) & \text{for
		}t>0,\quad x\in\mathbb{Z}_{p},\\
		1_{B}(x) & \text{for }t=0,
	\end{cases}
	\]
	is the transition function of a Markov process. Furthermore,
	\[%
	{\displaystyle\int\limits_{\mathbb{Z}_{p}}}
	Z_{0}(x,t)dx=1\text{, for }t>0,
	\]
	see Theorem \ref{Theorem_1}, in Appendix B.
	
	\subsection{CTMCs}
	
	In this section, we discuss the results given in Appendix C. We take
	$\mathcal{K}_{r}$, $r\in\mathbb{J}$ as in (\ref{partition_K}), \ where
	$\mathbb{J}$ is a countable set, possibly infinite. The value%
	\[
	p_{r,v}(t):=%
	{\displaystyle\int\limits_{\mathcal{K}_{v}}}
	p(t,x,\mathcal{K}_{r})\frac{dx}{\mu_{\text{Haar}}\left(  \mathcal{K}%
		_{v}\right)  },
	\]
	with $\mu_{\text{Haar}}\left(  \mathcal{K}_{v}\right)  =\int_{\mathcal{K}_{v}%
	}dx$, gives the transition probability that a
	particle/walker starting at some position in $\mathcal{K}_{v}$ will be found
	in the set $\mathcal{K}_{r}$ at the time $t$. \ Furthermore, $\sum
	_{r\in\mathbb{J}}p_{r,v}(t)=1$, for $t>0$, cf. Lemma \ref{Lemma_2} in Appendix C.
	
	Now, $\left[  p_{r,v}(t)\right]  _{r,v\in\mathbb{J}}$ is the transition
	matrix of a CTMC. Furthermore, under a suitable hypothesis on the Fourier
	transform of kernel $J(\left\vert x\right\vert _{p})$, the CTMC admits a
	stationary distribution $\left[  p_{r}^{\text{sta}}\right]  _{r\in\mathbb{J}%
	}=\left[  \int_{\mathcal{K}_{r}}dx\right]  _{r\in\mathbb{J}}$; see Appendix C,
	Theorem \ref{Theorem_2}. In Lemma \ref{Lemma_6A}, in Appendix E, we compute
	the rate matrix for Markov chains with finite number of states. The existence
	of a stationary distribution for this type of  CTMC follows from well-known
	results. Here, we work with CTMCs with countable states (finite or infinite)
	which are obtained as discretizations of $p$-adic heat equations; see
	\cite{Zuniga-networks}, and the references therein.

	\section{ $p$-adic Schr\"{o}dinger equations and CTQMCs}
	
	\subsection{A class of $p$-adic Schr\"{o}dinger equations}
	
	By performing a Wick rotation $t\rightarrow \mathrm{i}t$ in (\ref{System_1}), and
	taking $\Psi\left(  x,t\right)  =u\left(  x,it\right)  $, we obtain the
	following Schr\"{o}dinger equations in the unit ball:%
	\begin{equation}
		\left\{
		\begin{array}
			[c]{l}%
			\mathrm{i}\frac{\partial}{\partial t}\Psi\left(  x,t\right)  =-J\left(  \left\vert
			x\right\vert _{p}\right)  \ast\Psi\left(  x,t\right)  +\Psi\left(  x,t\right)
			\text{, }x\in\mathbb{Z}_{p},t\geq0\\
			\\
			\Psi\left(  x,0\right)  =\psi_{0}\left(  x\right)  \in\mathcal{D}%
			(\mathbb{Z}_{p}).
		\end{array}
		\right.  \label{System_2}%
	\end{equation}
	By using the Fourier transform, we obtain that the solution of (\ref{System_2}%
	) is given by%
	\begin{equation}
		\Psi\left(  x,t\right)  =Z_{0}(x,\mathrm{i}t)\ast\psi_{0}\left(  x\right)  \text{ in
		}\mathcal{D}^{\prime}(\mathbb{Q}_{p})\text{,}\label{Solution}%
	\end{equation}
	cf. \ Appendix D. This result is useful for calculations; however, it is
	necessary to show that if $\psi_{0}(x)\in L^{2}\left(  \mathbb{Z}_{p}\right)
	$, with $\left\Vert \psi_{0}\right\Vert _{2}=1$, then $\left\Vert \Psi
	(\cdot,t)\right\Vert _{2}=1$ for any $t>0$. This fact is established in
	Theorem \ref{Theorem_3}, in Appendix D.
	
	\subsection{CTQMCs}
	
	We use Theorem \ref{Theorem_0A} to construct a family of CTQMCs: for $l\geq1$,
	set $\mathbb{J}=G_{l}$ and
	\begin{equation}
		\pi_{r,v}\left(  t\right)  =\left\vert \left\langle p^{\frac{l}{2}}%
		\Omega\left(  p^{l}\left\vert x-r\right\vert _{p}\right)  ,Z_{0}(x,it)\ast
		p^{\frac{l}{2}}\Omega\left(  p^{l}\left\vert x-v\right\vert _{p}\right)
		\right\rangle \right\vert ^{2}, \label{CTQMC}%
	\end{equation}
	for $r,v\in G_{l}$. Then, $\left[  \pi_{r,v}\left(  t\right)  \right]
	_{r,v\in G_{l}}$ is the transition matrix of a CTQMC. Furthermore, $\pi
	_{r,v}\left(  t\right)  $ is an oscillatory integral:%
	\[
	\pi_{r,v}\left(  t\right)  =p^{-2l}\left\vert \text{ }%
	{\displaystyle\int\limits_{\left\vert z\right\vert _{p}\leq p^{l}}}
	\chi_{p}\left(  z\left(  v-r\right)  \right)  e^{-it\widehat{J}\left(
		\left\vert z\right\vert _{p}\right)  }dz\right\vert ^{2};
	\]

	see Proposition \ref{Poposition_4}, Theorem \ref{Theorem_5}, in Appendix D.

	\section{Discretizations of $p$-adic evolutionary
		differential equations}
	
	The numerical simulations required here use spatial discretizations of the
	$p$-adic heat and Schr\"{o}dinger equations. An in-depth discussion is given in
	Appendix D.
	
	For $l\geq1$, the $p$-adic integers of the form%
	\begin{equation}
		I=I_{0}+I_{1}p+\ldots+I_{l-1}p^{l-1}, \label{reprsentatives}%
	\end{equation}
	with $I_{j}\in\left\{  0,\ldots,p-1\right\}  $, form a set of representatives
	of the elements of the quotient \ group $G_{l}=\mathbb{Z}_{p}/p^{l}%
	\mathbb{Z}_{p}$. From a geometric perspective, $G_{l}$ corresponds to a finite
	rooted tree, with $l+1$ levels; the vertices at the top level correspond with
	the $p$-adic integers of the form (\ref{reprsentatives}). We denote by
	$\Omega\left(  p^{l}\left\vert x-I\right\vert _{p}\right)  $ the
	characteristic function of the ball $I+p^{l}\mathbb{Z}_{p}$. We denote by
	$\mathcal{D}_{l}(\mathbb{Z}_{p})$, the $\mathbb{C}$-vector space spanned by
	the linearly independent set $\left\{  \Omega\left(  p^{l}\left\vert
	x-I\right\vert _{p}\right)  \right\}  _{I\in G_{l}}$. The space of test
	functions $\mathcal{D}(\mathbb{Z}_{p})$\ satisfies
	\[
	\mathcal{D}(\mathbb{Z}_{p})=%
	{\displaystyle\bigcup\limits_{l=0}^{\infty}}
	\mathcal{D}_{l}(\mathbb{Z}_{p})\text{, with }\mathcal{D}_{0}(\mathbb{Z}%
	_{p})=\mathbb{C}\text{, and }\mathcal{D}_{l}(\mathbb{Z}_{p})\subset
	\mathcal{D}_{l+1}(\mathbb{Z}_{p})\text{.}%
	\]
	Furthermore, since the Haar measure of $\mathbb{Z}_{p}$ is one, we have%
	\[
	L^{1}(\mathbb{Z}_{p})\supset L^{\rho}(\mathbb{Z}_{p})\supset\mathcal{C}%
	(\mathbb{Z}_{p})\supset\mathcal{D}(\mathbb{Z}_{p}),\text{ for  }  \rho\in\left(
	1,\infty\right]  \text{,}%
	\]
	and $\mathcal{D}(\mathbb{Z}_{p})$ is dense in $L^{1}(\mathbb{Z}_{p})$, which
	implies that a function $f$ $\in$ $L^{1}(\mathbb{Z}_{p})$ can be approximated
	by a function $\varphi\in$ $\mathcal{D}_{l}(\mathbb{Z}_{p})$, with $l=l(f)$.
	This fact implies the existence of a linear bounded operator
	\[
	\boldsymbol{P}_{l}:L^{1}(\mathbb{Z}_{p})\rightarrow\mathcal{D}_{l}%
	(\mathbb{Z}_{p})
	\]
	such that $\boldsymbol{P}_{l}f\in\mathcal{D}_{l}(\mathbb{Z}_{p})$ is an
	approximation of $f$ $\in$ $L^{1}(\mathbb{Z}_{p})$, see Appendix E for details.
	
	If $\boldsymbol{G}:L^{1}(\mathbb{Z}_{p})\rightarrow L^{1}(\mathbb{Z}_{p})$ is
	a linear operator, a discretization $\boldsymbol{G}^{\left(  l\right)
	}:\mathcal{D}_{l}(\mathbb{Z}_{p})\rightarrow\mathcal{D}_{l}(\mathbb{Z}_{p})$
	is constructed as $\boldsymbol{G}^{\left(  l\right)  }=\boldsymbol{P}%
	_{l}\boldsymbol{G}$. Since $\mathcal{D}_{l}(\mathbb{Z}_{p})$ is a
	finite-dimensional vector space, $\boldsymbol{G}^{\left(  l\right)  }$\ is
	just a matrix.
	
	On the other hand, the space $\mathcal{D}_{l}(\mathbb{Z}_{p})$ is invariant
	under the operator $\boldsymbol{J}f=J\ast f- f$, cf. Lemma
	\ref{Lemma_6A} in Appendix E; in this way, we construct spatial approximations
	of the $p$-adic heat and Schr\"{o}dinger equations (\ref{System_1}%
	)-(\ref{System_2}); see equations (\ref{Discrte-Heat-Eq}%
	)-(\ref{Discrete-Schr-Eq}) in Appendix E. We compute explicitly the matrix
	$\boldsymbol{J}^{\left(  l\right)  }$, see Lemma \ref{Lemma_6A} in Appendix E.
	
	We identify $\mathcal{D}_{l}(\mathbb{Z}_{p})$ with the Hilbert space $%
\mathbb{C}^{p^{l}}$. This identification maps 
\[
\left\{ p^{\frac{l}{2}}\Omega \left( p^{l}\left\vert x-r\right\vert
_{p}\right) \right\} _{r\in G_{l}}
\]
into the standard orthonormal basis $\ \left\{ e_{r}\right\} $ $_{r\in G_{l}}
$ of $\mathbb{C}^{p^{l}}$. We show that the transition probabilities of the
CTQMC, $\left[ \pi _{r,v}\left( t\right) \right] _{r,v\in G_{l}}$, given in (%
\ref{CTQMC}) satisfy $\pi _{I,J}\left( t\right) =\left\vert \left\langle
e_{I}\right\vert \mathrm{e}^{it\boldsymbol{J}^{\left( l\right) }}\left\vert
e_{J}\right\rangle _{l}\right\vert ^{2}$, for $J,I\in G_{l}$. We use this
formula in the numerical simulations given in the next section.

	\label{Sec_Num_simulations}\section{Numerical simulations}
	We want to compare (via numerical simulation) the CTQMCs against a classical
	CTMCs. We take $r\in\mathbb{J}=G_{l}$, and $\mathcal{K}_{r}=B_{-l}%
	(r)=r+p^{l}\mathbb{Z}_{p}$, and identify $\left\{  p^{\frac{l}{2}}%
	\Omega\left(  p^{l}\left\vert x-r\right\vert _{p}\right)  \right\}  _{r\in
		G_{l}}$ with the standard orthonormal basis $\left\{  e_{r}\right\}  _{r\in
		G_{l}}$ of $\mathbb{R}^{p^{l}}$. Then 
\begin{eqnarray*}
p_{r,v}(t) &=&p^{l}\int\limits_{B_{-l}(v)}Z_{0}\left( x,t\right) \ast
\Omega \left( p^{l}\left\vert x-r\right\vert _{p}\right) dx \\
&=&\left\langle \Omega \left( p^{l}\left\vert x-v\right\vert _{p}\right)
,Z_{0}\left( x,t\right) \ast \Omega \left( p^{l}\left\vert x-r\right\vert
_{p}\right) \right\rangle ,
\end{eqnarray*}
for $r,v\in G_{l}$. Now, $Z_{0}\left(  x,t\right)  \ast\Omega\left(
	p^{l}\left\vert x-r\right\vert _{p}\right)  $ is the solution of the Cauchy
	problem for the heat equation with initial datum $\Omega\left(  p^{l}%
	\left\vert x-r\right\vert _{p}\right)  $, then by Remark \ref{Nota_final},
	\[
	Z_{0}\left(  x,t\right)  \ast\Omega\left(  p^{l}\left\vert x-r\right\vert
	_{p}\right)  =\mathrm{e}^{-t\boldsymbol{J}}\Omega\left(  p^{l}\left\vert x-r\right\vert
	_{p}\right)  =\mathrm{e}^{t\boldsymbol{J}^{\left(  l\right)  }}e_{r},
	\]
	and%
	\begin{equation}
		p_{r,v}(t)=\left\langle e_{v},\mathrm{e}^{t\boldsymbol{J}^{\left(  l\right)  }}%
		e_{r}\right\rangle _{l}.\label{Prob_Transition_II}%
	\end{equation}

	\subsection{Heat and Schr\"{o}dinger equation attached to Bessel potentials}
	
	Set $\Gamma\left(  \alpha\right)  =\frac{1-p^{\alpha-1}}{1-p^{-\alpha}}$,
	$\alpha\in\mathbb{R}\smallsetminus\left\{  0\right\}  $. We set for
	$\alpha\neq1$,%
	\[
	J_{\alpha}\left(  x\right)  =J_{\alpha}\left(  \left\vert
	x\right\vert _{p}\right)  =\frac{1}{\Gamma\left(  \alpha\right)  }\left\{
	\left\vert x\right\vert _{p}^{\alpha-1}-p^{\alpha-1}\right\}  \Omega\left(
	\left\vert x\right\vert _{p}\right)  ,
	\]
	and
	\[
	J_{1}\left(  x\right)  =J_{1}\left(  \left\vert
	x\right\vert _{p}\right)  =\left(  1-p^{-1}\right)  \log_{p}\left(  \frac
	{p}{\left\vert x\right\vert _{p}}\right)  \Omega\left(  \left\vert
	x\right\vert _{p}\right).
	\]
	Notice that for $\alpha\in\left(  0,1\right]  $, $J_{\alpha}\left(  \left\vert
	x\right\vert _{p}\right)  $ is not defined at the origin.
	
	Then $\left\Vert J_{\alpha}\right\Vert _{1}=1$, $\ J_{\alpha}\left(  |\xi
	|_{p}\right)  \geq0$, for $x\neq0$, and 
$\widehat{J}_{\alpha }\left( |\xi |_{p}\right) =\left( \max \left\{
1,\left\vert \xi \right\vert _{p}\right\} \right) ^{-\alpha }$, 
see \cite[Lemma 5.2]{Taibleson}.

	For $\alpha\in\left(  0,\infty\right)  $, we define the operator
	\[
	\boldsymbol{J}_{\alpha}\varphi\left(  x\right)  =J_{\alpha}\left(  \left\vert
	x\right\vert _{p}\right)  \ast\varphi\left(  x\right)  -\varphi\left(
	x\right)  \text{, for }\varphi\in\mathcal{C}(\mathbb{Z}_{p}),
	\]
	the heat equation
	\[
	\frac{\partial}{\partial t}u\left(  x,t\right)  =J_{\alpha}\left(  \left\vert
	x\right\vert _{p}\right)  \ast u\left(  x,t\right)  -u\left(  x,t\right)
	\text{, }x\in\mathbb{Z}_{p},t\geq0,
	\]
	and the Schr\"{o}dinger equation%
	\[
	\mathrm{i}\frac{\partial}{\partial t}\Psi\left(  x,t\right)  =-\left\{  J_{\alpha
	}\left(  \left\vert x\right\vert _{p}\right)  \ast\Psi\left(  x,t\right)
	-\Psi\left(  x,t\right)  \right\}  \text{, }x\in\mathbb{Z}_{p},t\geq0.
	\]

	\subsection{Numerical simulation 1}
	
	In this section, for $\alpha\in\left(  0,\infty\right)  \smallsetminus\left\{
	1\right\}  $, we provide numerical simulations for the CTMCs and CTQMCs
	attached to the heat and Schr\"{o}dinger equations attached to Bessel potentials.

	\textbf{Matrix of operator} $\boldsymbol{J}_{\alpha}$%
	
	\begin{equation}
		\boldsymbol{J}^{\left(  l\right)  }\left(  \alpha\right)  =\left[
		J_{I,K}^{\left(  l\right)  }\left(  \alpha\right)  \right]  _{I,K\in G_{l}%
		},\label{Matrix_J_1}%
	\end{equation}%
	\begin{equation}
		J_{I,K}^{\left(  l\right)  }\left(  \alpha\right)  =\left\{
		\begin{array}
			[c]{lll}%
			p^{-l}\frac{1-p^{-\alpha}}{1-p^{\alpha-1}}\left\{  \left\vert I-K\right\vert
			_{p}^{\alpha-1}-p^{\alpha-1}\right\}   & \text{if} & I\neq K\\
			&  & \\
			-%
			{\displaystyle\sum\limits_{\substack{I\in G_{l}\\I\neq0}}}
			p^{-l}\frac{1-p^{-\alpha}}{1-p^{\alpha-1}}\left\{  \left\vert I\right\vert
			_{p}^{\alpha-1}-p^{\alpha-1}\right\}   & \text{if} & I=K
		\end{array}
		\right.  \label{Matrix_J_2}%
	\end{equation}
	
	\begin{figure}
		\centering
		\includegraphics[width=0.75\linewidth]{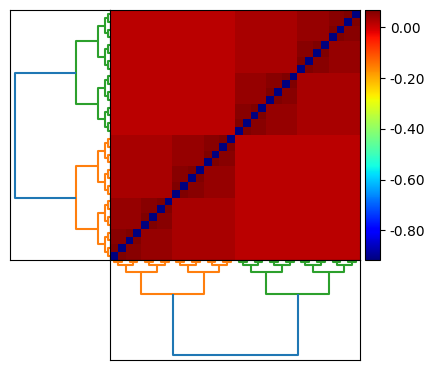}
		\caption{Numerical Simulation 1. Matrix $\boldsymbol{J}^{\left(  l\right)  }\left(  \alpha\right)$. The parameters are $p=2$, $l=5$, $\alpha=1.2$. The  CTMC has $32$ states organized in a finite tree $G_{5}$. The figure illustrates the ultrametric nature of matrix $\boldsymbol{J}^{\left(  l\right)  }\left(  \alpha\right)$.}
		\label{Figure 1}
	\end{figure}
	
	\textbf{CTMC}%
	
	\[
	p_{I,J}(t)=\left\langle e_{J},\mathrm{e}^{t\boldsymbol{J}^{\left(  l\right)
		}\left(  \alpha\right)  }e_{I}\right\rangle _{l}%
	\]

	\begin{figure}
		\centering
		\includegraphics[width=1\linewidth]{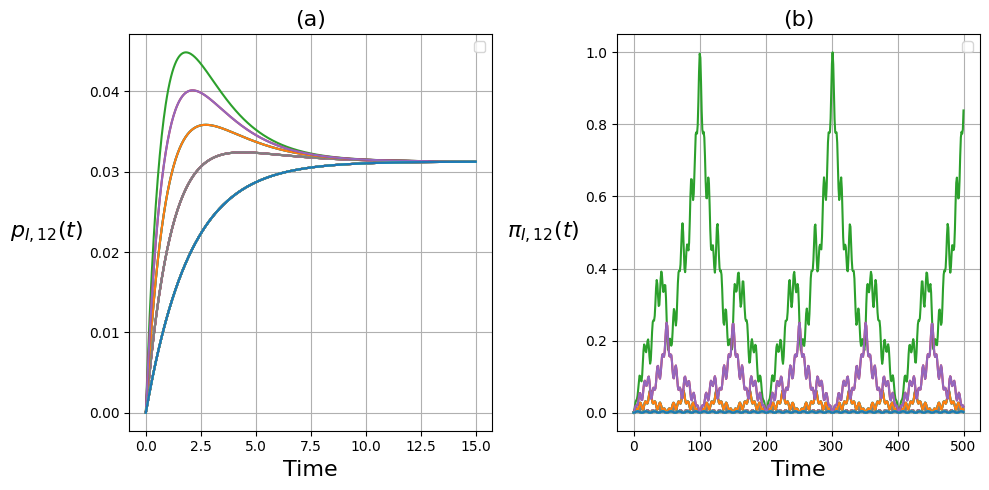}
		\caption{Numerical Simulation 1.  $p_{I,12}(t)$ versus $\pi_{I,12}\left(  t\right)  $. The state $I$ runs through $5$ states. The parameters are $p=2$, $l=5$, $\alpha=1.2$. Figure (a)
			shows that $\lim_{t\rightarrow\infty}p_{I,12}(t)\approx2^{-5}$, while figure (b) illustrates that
			$\lim_{t\rightarrow\infty}\pi_{I,12}\left(  t\right)  $\ does not exist.
		}
		\label{Figure 2}
	\end{figure}

	\textbf{CTQMC}%
	\[
	\pi_{I,J}\left(  t\right)  =\left\vert \left\langle e_{I}\right\vert
	\mathrm{e}^{\mathrm{i}t\boldsymbol{J}^{\left(  l\right)  }\left(  \alpha\right)
	}\left\vert e_{J}\right\rangle _{l}\right\vert ^{2}\text{, for }J,I\in G_{l},
	\]
	for $\alpha\in\left(  0,\infty\right)  \smallsetminus\left\{  1\right\}  $.
	Here, it is relevant to mention, that $\boldsymbol{J}^{\left(  l\right)
	}\left(  \alpha\right)  $ is a symmetric matrix, and thus $\pi_{I,J}\left(
	t\right)  $ is well-defined for $\alpha\in\left(  0,\infty\right)
	\smallsetminus\left\{  1\right\}  $.
	
	\begin{figure}
		\centering
		\includegraphics[width=1\linewidth]{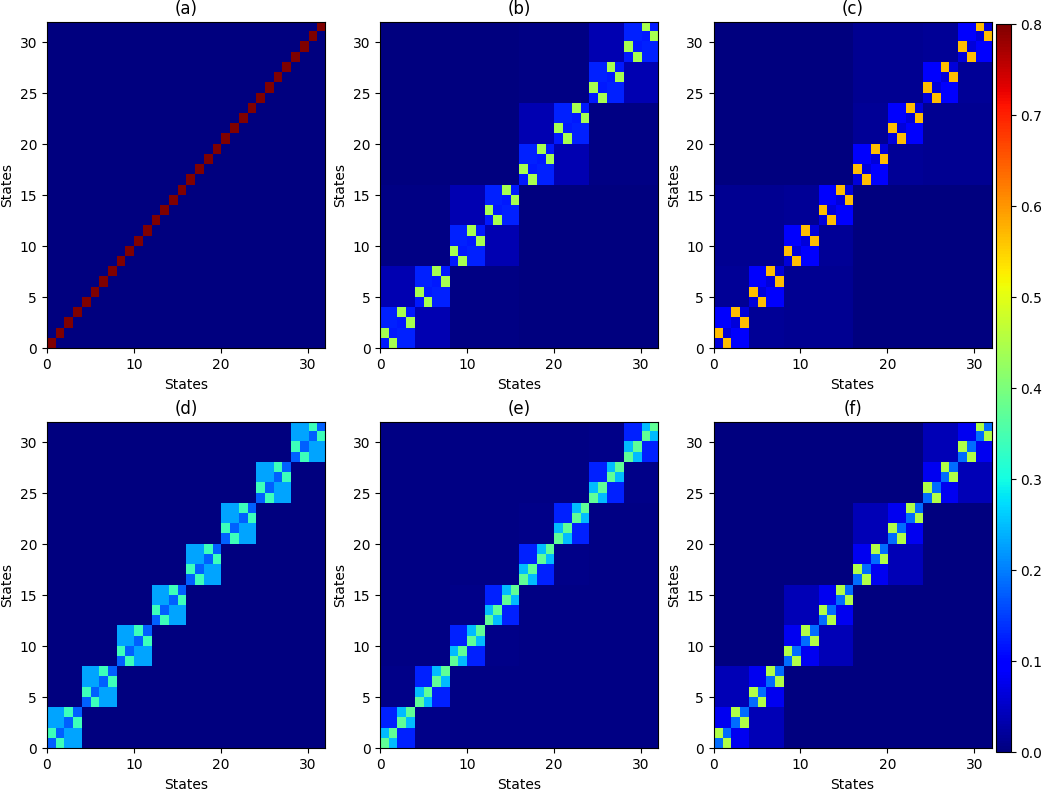}
		\caption{Numerical Simulation 1. $\pi_{I,J}\left(  t\right)  $ for six different times: $t=0.1$, $200$,
			$500$, $1000$, $4000$, $10000$.  The states $I$, $J$ run through $32$ states.
			The parameters are $p=2$, $l=5$, $\alpha=1.2$. 
		}
		\label{Figure 3}
	\end{figure}

	\begin{figure}
		\centering
		\includegraphics[width=1\linewidth]{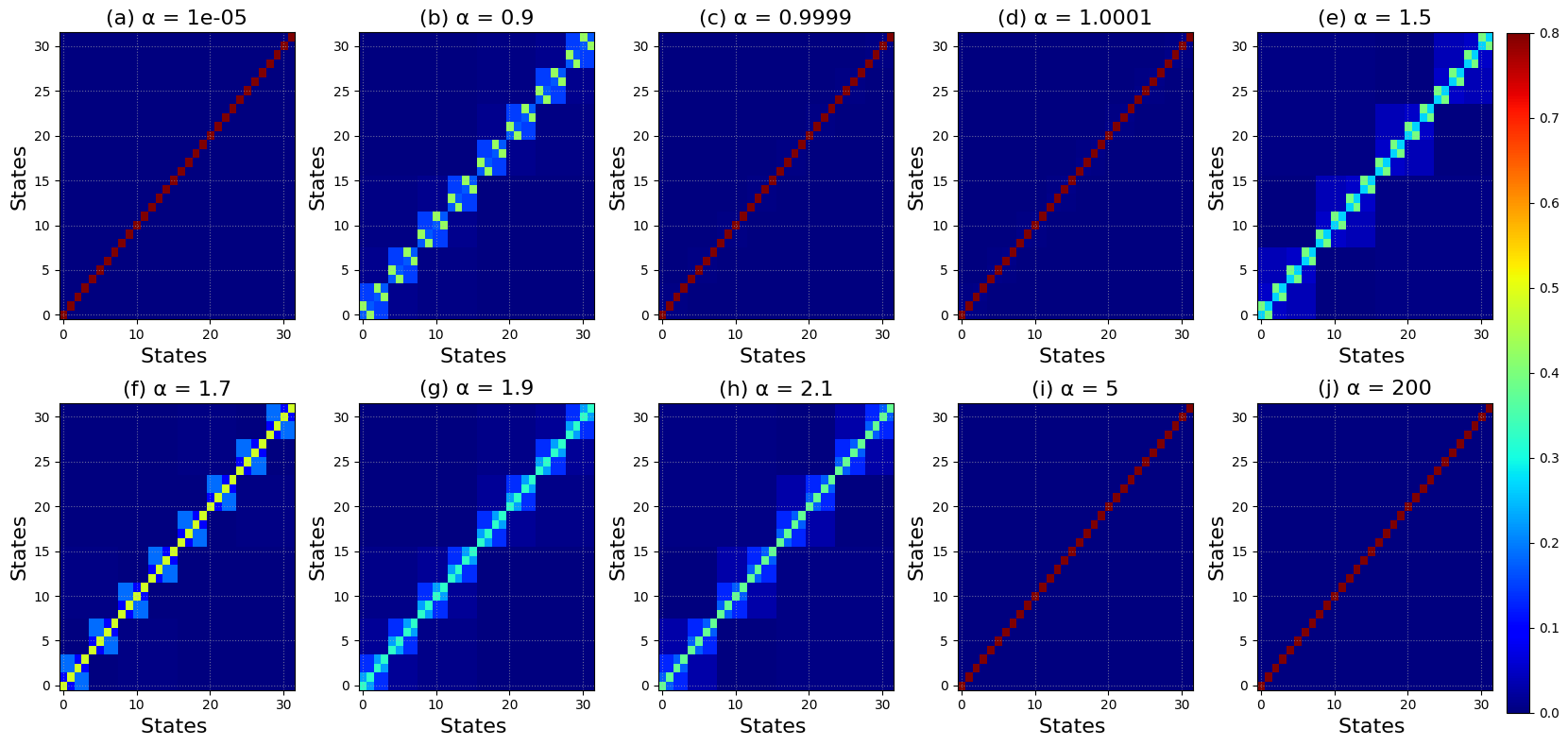}
		\caption{Numerical Simulation 1. $\pi_{I,J}\left(  200\right)  $ for six different values of the parameter
			\ $\alpha$.  The other parameters are $p=2$, $l=5$. The states $I$, $J$ run
			through $32$ states. The values $\alpha=0$, $1$ correspond to a pole,
			respectively a zero, of $\Gamma\left(  \alpha\right)  $. When the value of
			$\alpha$ is near to $0$, $1$, or $\alpha\geq5$, the transitions of the CTQMC
			occur around the diagonal; the walker only performs very short walks around each
			state. While for other values of $\alpha$, the walker performs long walks around any state.}
		\label{Figure 4}
	\end{figure}
	
	\begin{figure}
		\centering
		\includegraphics[width=0.75\linewidth]{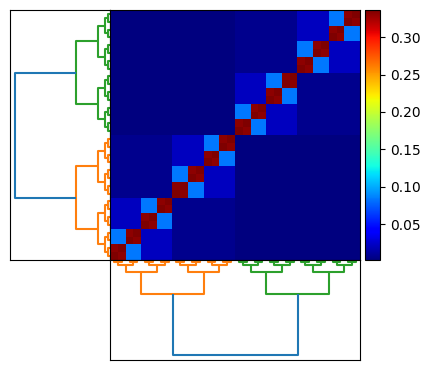}
		\caption{Numerical Simulation 1. The limiting  distribution (or longtime average) is defined as
			\[
			\chi_{I,J}=\lim_{T\rightarrow\infty}\frac{1}{T}%
			{\displaystyle\int\limits_{0}^{T}}
			\pi_{I,J}\left(  t\right)  dt.
			\]
			A numerical approximation of this  limiting distribution, $p=2$, $l=5$,
			$\alpha=1.2$ is shown in the figure. The states $I$, $J$ run through $32$
			states. In the literature, this limiting distribution is compared with
			\[
			p_{I}^{\text{sta}}=\lim_{t\rightarrow\infty}p_{I,J}\left(  t\right)  =p^{-l},
			\]
			cf. Theorem \ref{Theorem_2}, in Appendix C. The simulation shows that $\chi_{I,J}>p_{I}^{\text{sta}}=2^{-5}$. This fact is interpreted as the computational power of the CTQMC is greater than that of the corresponding CTMC. The computation of the average used the time interval $\left[  0,10000\right]$. The inequality $\chi_{I,J}>p_{I}^{\text{sta}}$ is only valid for certain maximal $ I$s, since the sum over $I$ produces one on both sides. }
		\label{Figure 5}
	\end{figure}
	
	\subsection{Numerical simulation 2}
	In this section, for $\alpha=1$, we provide numerical simulations for the CTMCs and CTQMCs attached to the heat and Schr\"{o}dinger equations attached to Bessel potentials.

	\textbf{Matrix of operator} ${\boldsymbol{J}}_{1}$%
	
	\[
	{\boldsymbol{J}}^{\left(  l\right)  }\left(  1\right)  =\left[
	J_{I,K}^{\left(  l\right)  }\left(  1\right)  \right]  _{I,K\in G_{l}},\text{
	}%
	\]

	\[
	J_{I,K}^{\left(  l\right)  }\left(  1\right)  =\left\{
	\begin{array}
		[c]{lll}%
		p^{-l}\left(  1-p^{-1}\right)  \log_{p}\left(  \frac{p}{\left\vert
			I-K\right\vert _{p}}\right)   & \text{if} & I\neq K\\
		&  & \\
		-%
		{\displaystyle\sum\limits_{\substack{I\in G_{l}\\I\neq0}}}
		p^{-l}\left(  1-p^{-1}\right)  \log_{p}\left(  \frac{p}{\left\vert
			I\right\vert _{p}}\right)   & \text{if} & I=K.
	\end{array}
	\right.
	\]

	\textbf{CTMC}%
	
	\[
	p_{r,v}(t)=\left\langle e_{v},\mathrm{e}^{t\boldsymbol{J}^{\left(  l\right)
		}\left(  1\right)  }e_{r}\right\rangle _{l}.
	\]

	\textbf{CTQMC}%
	\[
	\pi_{I,J}\left(  t\right)  =\left\vert \left\langle e_{I}\right\vert
	\mathrm{e}^{\mathrm{i}t\boldsymbol{J}^{\left(  l\right)  }\left(  1\right)  }\left\vert
	e_{J}\right\rangle _{l}\right\vert ^{2}\text{, for }J,I\in G_{l}.
	\]
	
	\begin{figure}
		\centering
		\includegraphics[width=1\linewidth]{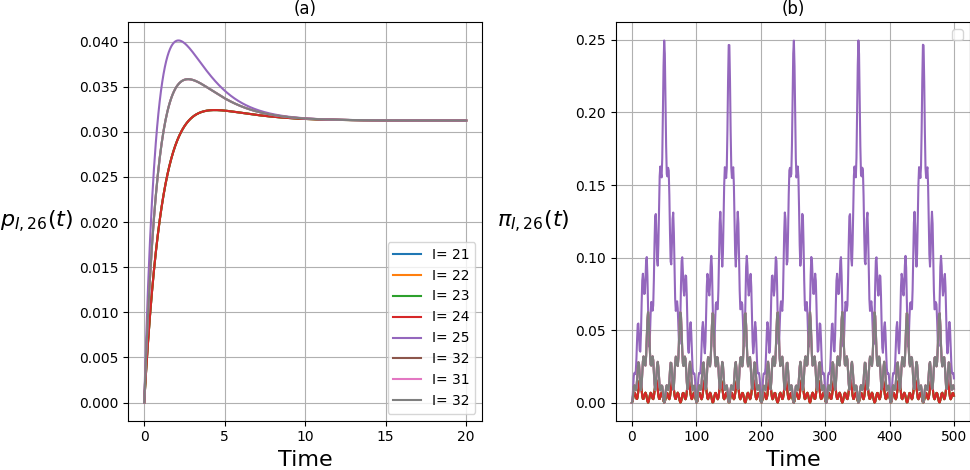}
		\caption{Numerical Simulation 2. $p_{I,26}(t)$ \ versus $\pi_{I,26}(t)$ for $8$ different states. The parameters are $p=2$, $l=5$, $\alpha=1$. The figure (a) clearly illustrates that $\lim_{t\rightarrow\infty}p_{I,26}(t)\approx2^{-5}$.
		}
		\label{Figure 6}
	\end{figure}
	
	\begin{figure}
		\centering
		\includegraphics[width=1\linewidth]{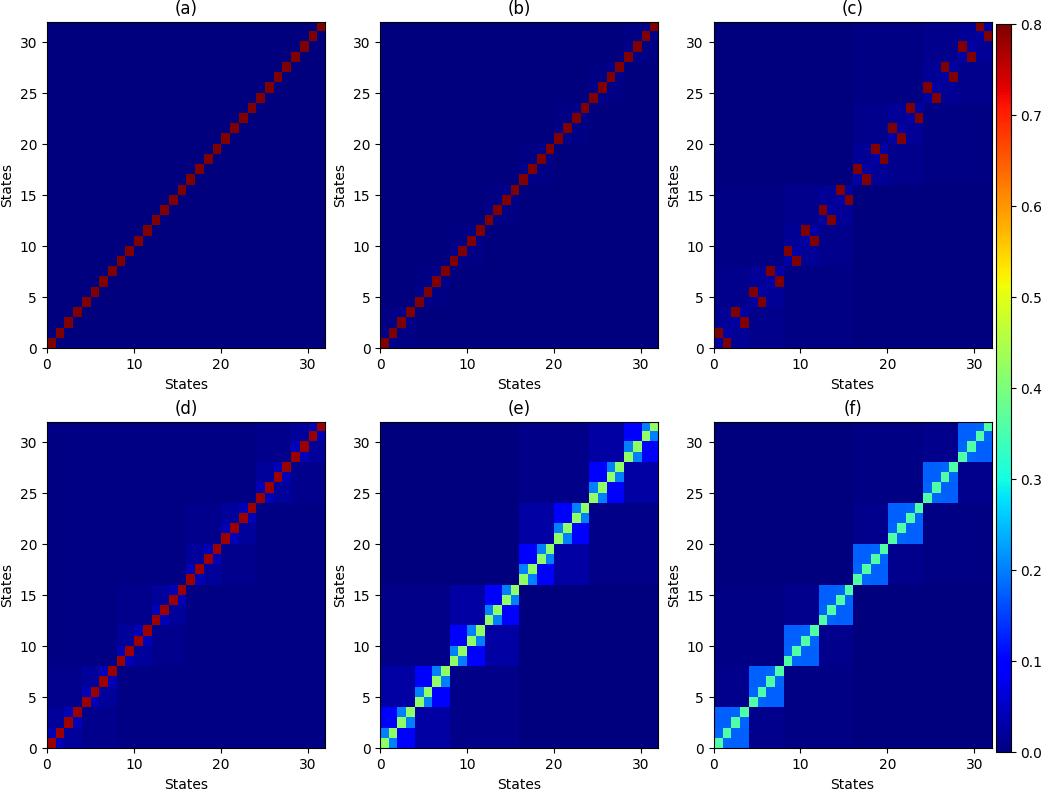}
		\caption{Numerical Simulation 2. $\pi_{I,J}\left(  t\right)  $ for six different times: $t=0.1$, $200$,
			$500$, $1000$, $4000$, $10000$.  The states $I$, $J$ run through $32$ states.
			The parameters are $p=2$, $l=5$, $\alpha=1$. Notice that the figures 12,15, 21 in \cite{Mulkne-Blumen} also show an ultrametric nature. }
		\label{Figure 7}
	\end{figure}

\begin{figure}
		\centering
		\includegraphics[width=0.75\linewidth]{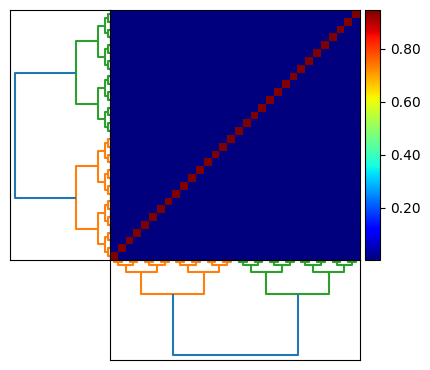}
		\caption{Numerical Simulation 2. A numerical approximation of the limiting distribution $\chi_{I, J}$. The computation of the average used the time interval $\left[  0,10000\right]$. The parameters are $p=2$, $l=5$, $\alpha=1$.}
		\label{Figure 8}
\end{figure}
	
\begin{figure}
		\centering
		\includegraphics[width=1\linewidth]{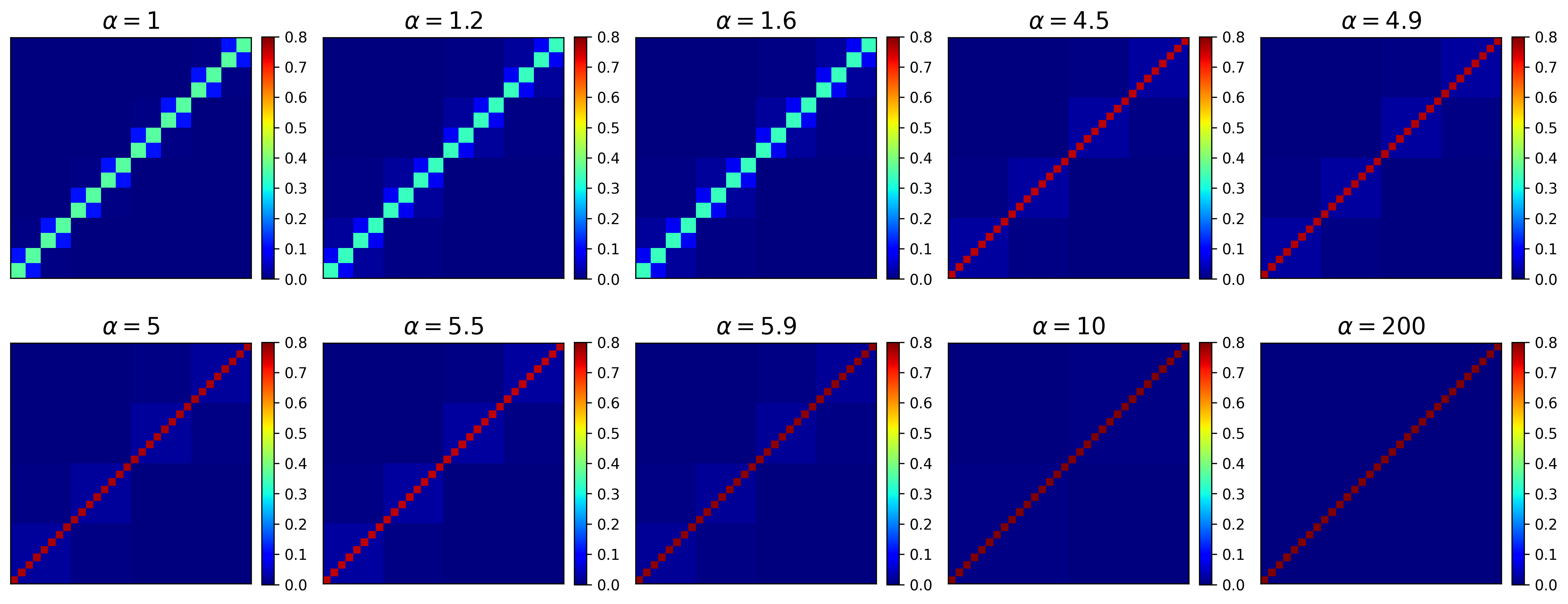}
		\caption{Limiting distribution $\chi=$ $\left[  \chi_{I,J}\right]  $ for several values of $\alpha$. The parameters are $p=2$, $l=5$. The computation of the average
			used the time interval $\left[  0,10000\right]  $.  This figure illustrates the strong influence of $\alpha$ on the limiting distribution. For $ \alpha\geq 4.5$, the distribution $\chi$ is concentrated around the diagonal, which means that only very short walks around each state are allowed.}
		\label{Figure 9}
\end{figure}
	
\begin{figure}
		\centering
		\includegraphics[width=1\linewidth]{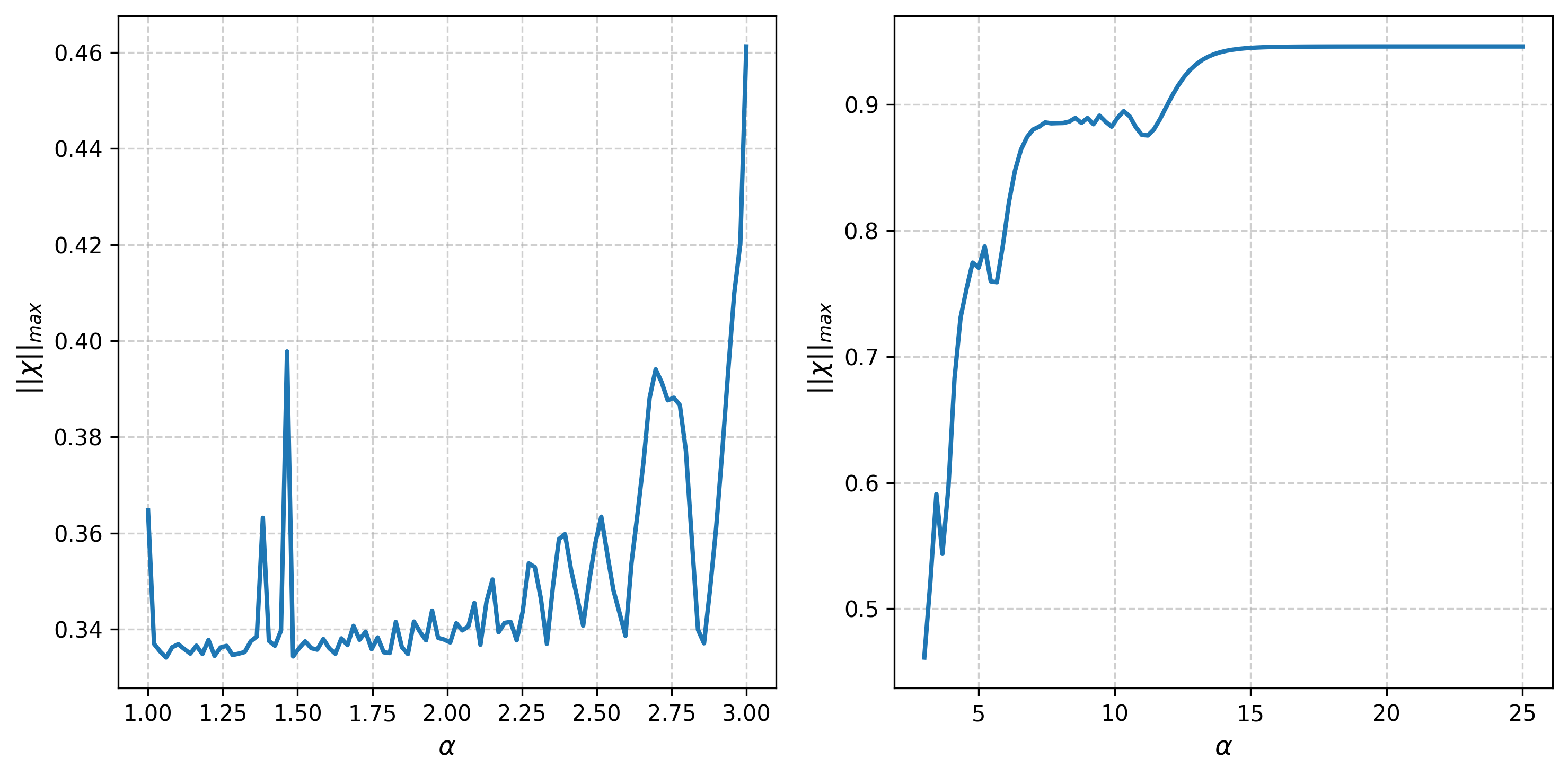}
		\caption{For  $\chi=$ $\left[  \chi_{I,J}\right]  $, we set $\left\Vert \chi\right\Vert
			_{\text{max}}=\max_{I,J}\left\{  \chi_{I,J}\right\}  $. The figure shows
			$\left\Vert \chi\right\Vert _{\text{max}}$ for several values of $\alpha$. The other
			parameters are $p=2$, $l=5$. The computation of the average used the time
			interval $\left[  0,10000\right]  $. This figure shows an alternative way to visualize the information displayed in Figure \ref{Figure 9}.
		}
		\label{Figure 10}
\end{figure}

\FloatBarrier

\section{Appendix A: Basic facts on $p$-adic analysis}
	
	In this section, we fix the notation and collect some basic results on
	$p$-adic analysis that we will use throughout the article. For a detailed
	exposition on $p$-adic analysis, the reader may consult \cite{V-V-Z},
	\cite{Zuniga-Textbook}, \cite{Alberio et al}, \cite{Taibleson}.
	
	\subsection{The field of $p$-adic numbers}
	
	The field of $p$-adic numbers $\mathbb{Q}_{p}$ is defined as the completion of
	the field of rational numbers $\mathbb{Q}$ with respect to the $p$-adic norm
	$|\cdot|_{p}$, which is defined as
	\[
	|x|_{p}=%
	\begin{cases}
		0 & \text{if }x=0\\
		p^{-\gamma} & \text{if }x=p^{\gamma}\dfrac{a}{b},
	\end{cases}
	\]
	where $a$ and $b$ are integers coprime with $p$. The integer $\gamma
	=ord_{p}(x):=ord(x)$, with $ord(0):=+\infty$, is called the\textit{\ }$p$-adic
	order of $x$.
	
	The metric space $\left(  \mathbb{Q}_{p},|\cdot|_{p}\right)  $ is a complete
	ultrametric space. As a topological space $\mathbb{Q}_{p}$\ is homeomorphic to
	a Cantor-like subset of the real line, see, e.g., \cite{V-V-Z}, \cite{Alberio et al}.
	
	Any $p$-adic number $x\neq0$ has a unique expansion of the form
	\[
	x=p^{ord(x)}\sum_{j=0}^{\infty}x_{j}p^{j},
	\]
	where $x_{j}\in\{0,1,2,\dots,p-1\}$ and $x_{0}\neq0$. By using this expansion,
	we define \textit{the fractional part }$\{x\}_{p}$\textit{ of }$x\in
	\mathbb{Q}_{p}$ as the rational number
	\[
	\{x\}_{p}=%
	\begin{cases}
		0 & \text{if }x=0\text{ or }ord(x)\geq0\\
		p^{ord(x)}\sum_{j=0}^{-ord(x)-1}x_{j}p^{j} & \text{if }ord(x)<0.
	\end{cases}
	\]
	In addition, any $x\in\mathbb{Q}_{p}\smallsetminus\left\{  0\right\}  $ can be
	represented uniquely as $x=p^{ord(x)}v$, where $\left\vert v\right\vert
	_{p}=1$.
	
	\subsection{Topology of $\mathbb{Q}_{p}$}
	
	For $r\in\mathbb{Z}$, denote by $B_{r}(a)=\{x\in\mathbb{Q}_{p};\left\vert
	x-a\right\vert _{p}\leq p^{r}\}$ the ball of radius $p^{r}$ with center at
	$a\in\mathbb{Q}_{p}$, and take $B_{r}(0):=B_{r}$. Notice that $B_{0}%
	=\mathbb{Z}_{p}$ is the ring of\textit{ }$p$-adic integers. We also denote by
	$S_{r}(a)=\{x\in\mathbb{Q}_{p};\left\vert x-a\right\vert _{p}=p^{r}\}$ the
	sphere of radius\textit{ }$p^{r}$ with center at $a\in\mathbb{Q}_{p}$, and
	take $S_{r}(0):=S_{r}$. We notice that $S_{0}=\mathbb{Z}_{p}^{\times}$ (the
	group of units of $\mathbb{Z}_{p}$). The balls and spheres are both open and
	closed subsets in $\mathbb{Q}_{p}$. In addition, two balls in $\mathbb{Q}_{p}$
	are either disjoint or one is contained in the other.
	
	As a topological space $\left(  \mathbb{Q}_{p},\left\vert \cdot\right\vert
	_{p}\right)  $ is totally disconnected, i.e., the only connected \ subsets of
	$\mathbb{Q}_{p}$ are the empty set and the points. A subset of $\mathbb{Q}%
	_{p}$ is compact if and only if it is closed and bounded in $\mathbb{Q}_{p}$,
	see, e.g., \cite[Section 1.3]{V-V-Z}, or \cite[Section 1.8]{Alberio et al}. The balls
	and spheres are compact subsets. Thus $\left(  \mathbb{Q}_{p},\left\vert
	\cdot\right\vert _{p}\right)  $ is a locally compact topological space.
	
	\subsection{The Haar measure}
	
	Since $(\mathbb{Q}_{p},+)$ is a locally compact topological group, there
	exists a Haar measure $dx$, which is invariant under translations, i.e.,
	$d(x+a)=dx$, \cite{Halmos}. If we normalize this measure by the condition
	$\int_{\mathbb{Z}_{p}}dx=1$, then $dx$ is unique.
	
	\begin{notation}
		We will use $\Omega\left(  p^{-r}\left\vert x-a\right\vert _{p}\right)  $ to
		denote the characteristic function of the ball $B_{r}(a)=a+p^{-r}%
		\mathbb{Z}_{p}$, where
		\[
		\mathbb{Z}_{p}=\left\{  x\in\mathbb{Q}_{p};\left\vert x\right\vert _{p}%
		\leq1\right\}
		\]
		is the $1$-dimensional unit ball. For more general sets, we will use the
		notation $1_{A}$ for the characteristic function of set $A$.
	\end{notation}
	
	\subsection{The Bruhat-Schwartz space}
	
	A complex-valued function $\varphi$ defined on $\mathbb{Q}_{p}$ is
	\textit{called locally constant} if for any $x\in\mathbb{Q}_{p}$ there exist
	an integer $l(x)\in\mathbb{Z}$ such that%
	\begin{equation}
		\varphi(x+x^{\prime})=\varphi(x)\text{ for any }x^{\prime}\in B_{l(x)}.
		\label{local_constancy}%
	\end{equation}
	A function $\varphi:\mathbb{Q}_{p}\rightarrow\mathbb{C}$ is called a
	Bruhat-Schwartz function\textit{ }(or a test function) if it is locally
	constant with compact support. Any test function can be represented as a
	linear combination, with complex coefficients, of characteristic functions of
	balls. The $\mathbb{C}$-vector space of Bruhat-Schwartz functions is denoted
	by $\mathcal{D}(\mathbb{Q}_{p})$. For $\varphi\in\mathcal{D}(\mathbb{Q}_{p})$,
	the largest number $l=l(\varphi)$ satisfying (\ref{local_constancy}) is called the exponent of local constancy (or the parameter of constancy) of $\varphi$.
	
	\subsection{$L^{\rho}$ spaces}
	
	Given $\rho\in\lbrack1,\infty)$, we denote by $L^{\rho}\left(\mathbb{Q}_{p}\right)  $ the $\mathbb{C}$-vector space of all the complex valued
	functions $g$ satisfying
	\[
	\left\Vert g\right\Vert _{\rho}=\left(  \text{ }%
	{\displaystyle\int\limits_{\mathbb{Q}_{p}}}
	\left\vert g\left(  x\right)  \right\vert ^{\rho}dx\right)  ^{\frac{1}{\rho}%
	}<\infty,
	\]
	where $dx$ is the normalized Haar measure on $\left(  \mathbb{Q}_{p},+\right)
	$.
	
	If $U$ is an open subset of $\mathbb{Q}_{p}$, $\mathcal{D}(U)$ denotes the
	$\mathbb{C}$-vector space of test functions with supports contained in $U$,
	then $\mathcal{D}(U)$ is dense in
	\[
	L^{\rho}\left(  U\right)  =\left\{  \varphi:U\rightarrow\mathbb{C};\left\Vert
	\varphi\right\Vert _{\rho}=\left\{
	{\displaystyle\int\limits_{U}}
	\left\vert \varphi\left(  x\right)  \right\vert ^{\rho}dx\right\}
	^{\frac{1}{\rho}}<\infty\right\}  ,
	\]
	for $1\leq\rho<\infty$, see, e.g., \cite[Section 4.3]{Alberio et al}.
	
	\subsection{The Fourier transform}
	
	Set $\chi_{p}(y)=\exp(2\pi i\{y\}_{p})$ for $y\in\mathbb{Q}_{p}$. The map
	$\chi_{p}(\cdot)$ is an additive character on $\mathbb{Q}_{p}$, i.e., a
	continuous map from $\left(  \mathbb{Q}_{p},+\right)  $ into $S$ (the unit
	circle considered as a multiplicative group) satisfying $\chi_{p}(x_{0}%
	+x_{1})=\chi_{p}(x_{0})\chi_{p}(x_{1})$, $x_{0},x_{1}\in\mathbb{Q}_{p}$.\ The
	additive characters of $\mathbb{Q}_{p}$ form an Abelian group which is
	isomorphic to $\left(  \mathbb{Q}_{p},+\right)  $. The isomorphism is given by
	$\kappa\rightarrow\chi_{p}(\kappa x)$, see, e.g., \cite[Section 2.3]{Alberio et al}.
	
	The Fourier transform of $\varphi\in\mathcal{D}(\mathbb{Q}_{p})$ is defined
	as
	\[
	\mathcal{F}\varphi(\xi)=%
	{\displaystyle\int\limits_{\mathbb{Q}_{p}}}
	\chi_{p}(\xi x)\varphi(x)dx\quad\text{for }\xi\in\mathbb{Q}_{p}.
	\]
	The Fourier transform is a linear isomorphism from $\mathcal{D}(\mathbb{Q}%
	_{p})$ onto itself satisfying
	\begin{equation}
		(\mathcal{F}(\mathcal{F}\varphi))(\xi)=\varphi(-\xi), \label{Eq_FFT}%
	\end{equation}
	see, e.g., \cite[Section 4.8]{Alberio et al}. We will also use the notation
	$\mathcal{F}_{x\rightarrow\kappa}\varphi$ and $\widehat{\varphi}$\ for the
	Fourier transform of $\varphi$.
	
	The Fourier transform extends to $L^{2}$. If $f\in L^{2}\left(  \mathbb{Q}%
	_{p}\right)  $, its Fourier transform is defined as
	\[
	(\mathcal{F}f)(\xi)=\lim_{k\rightarrow\infty}%
	{\displaystyle\int\limits_{\left\vert x\right\vert _{p}\leq p^{k}}}
	\chi_{p}(\xi x)f(x)dx,\quad\text{for }\xi\in%
	\mathbb{Q}
	_{p},
	\]
	where the limit is taken in $L^{2}\left(  \mathbb{Q}_{p}\right)  $. We recall
	that the Fourier transform is unitary on $L^{2}\left(  \mathbb{Q}_{p}\right)
	,$ i.e. $||f||_{2}=||\mathcal{F}f||_{2}$ for $f\in L^{2}$ and that
	(\ref{Eq_FFT}) is also valid in $L^{2}$, see, e.g., \cite[Chapter III, Section
	2]{Taibleson}.
	
	\subsection{Distributions}
	
	The $\mathbb{C}$-vector space $\mathcal{D}^{\prime}\left(  \mathbb{Q}%
	_{p}\right)  $ of all continuous linear functionals on $\mathcal{D}%
	(\mathbb{Q}_{p})$ is called the Bruhat-Schwartz space of distributions. Every
	linear functional on $\mathcal{D}(\mathbb{Q}_{p})$ is continuous, i.e.,
	$\mathcal{D}^{\prime}\left(  \mathbb{Q}_{p}\right)  $\ agrees with the
	algebraic dual of $\mathcal{D}(\mathbb{Q}_{p})$, see, e.g., \cite[Chapter 1,
	VI.3, Lemma]{V-V-Z}.
	
	We endow $\mathcal{D}^{\prime}\left(  \mathbb{Q}_{p}\right)  $ with the weak
	topology, i.e. a sequence $\left\{  T_{j}\right\}  _{j\in\mathbb{N}}$ in
	$\mathcal{D}^{\prime}\left(  \mathbb{Q}_{p}\right)  $ converges to $T$ if
	$\lim_{j\rightarrow\infty}T_{j}\left(  \varphi\right)  =T\left(
	\varphi\right)  $ for any $\varphi\in\mathcal{D}(\mathbb{Q}_{p})$. The map
	\[%
	\begin{array}
		[c]{lll}%
		\mathcal{D}^{\prime}\left(  \mathbb{Q}_{p}\right)  \times\mathcal{D}%
		(\mathbb{Q}_{p}) & \rightarrow & \mathbb{C}\\
		\left(  T,\varphi\right)  & \rightarrow & T\left(  \varphi\right)
	\end{array}
	\]
	is a bilinear form which is continuous in $T$ and $\varphi$ separately. We
	call this map the pairing between $\mathcal{D}^{\prime}\left(  \mathbb{Q}%
	_{p}\right)  $ and $\mathcal{D}(\mathbb{Q}_{p})$. From now on we will use
	$\left(  T,\varphi\right)  $ instead of $T\left(  \varphi\right)  $.
	
	\subsection{The Fourier transform of a distribution}
	
	The Fourier transform $\mathcal{F}\left[  T\right]  $ of a distribution
	$T\in\mathcal{D}^{\prime}\left(  \mathbb{Q}_{p}\right)  $ is defined by%
	\[
	\left(  \mathcal{F}\left[  T\right]  ,\varphi\right)  =\left(  T,\mathcal{F}%
	\left[  \varphi\right]  \right)  \text{ for all }\varphi\in\mathcal{D}\left(
	\mathbb{Q}_{p}\right)  \text{.}%
	\]
	The Fourier transform $T\rightarrow\mathcal{F}\left[  T\right]  $ is a linear
	and continuous isomorphism from $\mathcal{D}^{\prime}\left(  \mathbb{Q}%
	_{p}\right)  $\ onto $\mathcal{D}^{\prime}\left(  \mathbb{Q}_{p}\right)  $.
	Furthermore, $T=\mathcal{F}\left[  \mathcal{F}\left[  T\right]  \left(
	-\xi\right)  \right]  $.
	
	Let $T\in\mathcal{D}^{\prime}\left(  \mathbb{Q}_{p}\right)  $ be a
	distribution. Then supp $T\subset B_{L}$ if and only if $\mathcal{F}\left[
	T\right]  $ is a locally constant function, and the exponent of local
	constancy of $\mathcal{F}\left[  T\right]  $ is $\geq-L$. In addition%
	
	\[
	\mathcal{F}\left[  T\right]  \left(  \xi\right)  =\left(  T\left(  y\right)
	,\Omega\left(  p^{-L}\left\vert y\right\vert _{p}\right)  \chi_{p}\left(  \xi
	y\right)  \right)  ,
	\]
	see, e.g., \cite[Section 4.9]{Alberio et al}.
	
	\subsection{The direct product of distributions}
	
	Given $F,G\in\mathcal{D}^{\prime}\left(  \mathbb{Q}_{p}\right)  $, their
	\textit{direct product }$F\times G$ is defined by the formula%
	\[
	\left(  F\left(  x\right)  \times G\left(  y\right)  ,\varphi\left(
	x,y\right)  \right)  =\left(  F\left(  x\right)  ,\left(  G\left(  y\right)
	,\varphi\left(  x,y\right)  \right)  \right)  \text{ for }\varphi\left(
	x,y\right)  \in\mathcal{D}\left(  \mathbb{Q}_{p}^{2}\right)  .
	\]
	The direct product is commutative: $F\times G=G\times F$. In addition the
	direct product is continuous with respect to the joint factors.
	
	\subsection{The convolution of distributions}
	
	Given $F,G\in\mathcal{D}^{\prime}\left(  \mathbb{Q}_{p}\right)  $, their
	convolution $F\ast G$ is defined by%
	\[
	\left(  F\ast G,\varphi\right)  =\lim_{k\rightarrow\infty}\left(  F\left(
	y\right)  \times G\left(  x\right)  ,\Omega\left(  p^{-k}\left\vert
	y\right\vert _{p}\right)  \varphi\left(  x+y\right)  \right)
	\]
	if the limit exists for all $\varphi\in\mathcal{D}\left(  \mathbb{Q}%
	_{p}\right)  $. We recall that if $F\ast G$ exists, then $G\ast F$ exists and
	$F\ast G=G\ast F$, see, e.g., \cite[Section 7.1]{V-V-Z}. If $F,G\in
	\mathcal{D}^{\prime}\left(  \mathbb{Q}_{p}\right)  $ and supp$G\subset B_{L}$,
	then the convolution $F\ast G$ exists, and it is given by the formula%
	\[
	\left(  F\ast G,\varphi\right)  =\left(  F\left(  y\right)  \times G\left(
	x\right)  ,\Omega\left(  p^{-L}\left\vert y\right\vert _{p}\right)
	\varphi\left(  x+y\right)  \right)  \text{ for }\varphi\in\mathcal{D}\left(
	\mathbb{Q}_{p}\right)  .
	\]
	In the case in which $G=\psi\in\mathcal{D}\left(  \mathbb{Q}_{p}\right)  $,
	$F\ast\psi$ is a locally constant function given by
	\[
	\left(  F\ast\psi\right)  \left(  y\right)  =\left(  F\left(  x\right)
	,\psi\left(  y-x\right)  \right)  ,
	\]
	see, e.g., \cite[Section 7.1]{V-V-Z}.
	
	\subsection{The multiplication of distributions}
	
	Set $\delta_{k}\left(  x\right)  :=p^{k}\Omega\left(  p^{k}\left\vert
	x\right\vert _{p}\right)  $ for $k\in\mathbb{N}$. Given $F,G\in\mathcal{D}%
	^{\prime}\left(  \mathbb{Q}_{p}\right)  $, their product $F\cdot G$ is defined
	by%
	\[
	\left(  F\cdot G,\varphi\right)  =\lim_{k\rightarrow\infty}\left(  G,\left(
	F\ast\delta_{k}\right)  \varphi\right)
	\]
	if the limit exists for all $\varphi\in\mathcal{D}\left(  \mathbb{Q}%
	_{p}\right)  $. If the product $F\cdot G$ exists then the product $G\cdot F$
	exists and they are equal.
	
	We recall that \ the existence of the product $F\cdot G$ is equivalent \ to
	the existence of $\mathcal{F}\left[  F\right]  \ast\mathcal{F}\left[
	G\right]  $. In addition, $\mathcal{F}\left[  F\cdot G\right]  =\mathcal{F}%
	\left[  F\right]  \ast\mathcal{F}\left[  G\right]  $ and $\mathcal{F}\left[
	F\ast G\right]  =\mathcal{F}\left[  F\right]  \cdot\mathcal{F}\left[
	G\right]  $, see, e.g., \cite[Section 7.5]{V-V-Z}.

	\section{Appendix B: $p$-adic heat equations on the unit ball and Markov
		processes}
	\subsection{Preliminary results}
	
	For fix a real-valued function $J:\left[  0,1\right]  \rightarrow\left[
	0,\infty\right)  $, we define%
	\[%
	\begin{array}
		[c]{cccc}%
		J: & \mathbb{Z}_{p} & \rightarrow & \left[  0,\infty\right)  \\
		& x & \rightarrow & J(\left\vert x\right\vert _{p}).
	\end{array}
	\]
	We assume that
	\[
	\left\Vert J\right\Vert _{L^{1}\left(  \mathbb{Z}_{p}\right)  }:=\left\Vert
	J(\left\vert x\right\vert _{p})\right\Vert _{1}=%
	{\displaystyle\int\limits_{\mathbb{Z}_{p}}}
	J(\left\vert x\right\vert _{p})dx=1.
	\]

	\begin{remark}
		(i) By extending $J(\left\vert x\right\vert _{p})$ as zero outside of the unit
		ball, we have $J(\left\vert x\right\vert _{p})\in L^{1}\left(  \mathbb{Q}%
		_{p}\right)  $, and then its Fourier transform $\widehat{J}(\xi):\mathbb{Q}%
		_{p}\rightarrow\mathbb{C}$ is well-defined. The Fourier transform does not
		preserve the support of $J(\left\vert x\right\vert _{p})$, i.e. $\widehat
		{J}(\xi)\neq0$ for points $\xi$ outside of the unit ball.
		
		(ii) Notice that $\widehat{J}(0)=\left\Vert J(\left\vert x\right\vert
		_{p})\right\Vert _{1}=1$.
	\end{remark}
	
	\begin{lemma}
		\label{Lemma_1A}With the above notation, the following assertions hold:
		
		\noindent(i) the Fourier transform $\widehat{J}(\xi)$ of \ $J(\left\vert
		x\right\vert _{p})$ is a real-valued function given by
		\[
		\widehat{J}(\xi)=%
		{\displaystyle\sum\limits_{j=0}^{\infty}}
		J(p^{-j})\left\{  p^{-j}\Omega\left(  p^{-j}\left\vert \xi\right\vert
		_{p}\right)  -p^{-j-1}\Omega\left(  p^{-j-1}\left\vert \xi\right\vert
		_{p}\right)  \right\}  \text{, }\xi\in\mathbb{Q}_{p};
		\]

		\noindent(ii) for $\xi\in\mathbb{Q}_{p}\smallsetminus\left\{  0\right\}  $,
		$\widehat{J}(\xi)=\widehat{J}(u\xi)$, for any $u\in\mathbb{Z}_{p}^{\times}$,
		this means that $\widehat{J}(\xi)$ is a radial function $\widehat{J}%
		(\xi)=\widehat{J}(\left\vert \xi\right\vert _{p})$ for $\xi\in\mathbb{Q}%
		_{p}\smallsetminus\left\{  0\right\}  $;
		
		\noindent(iii) for $\xi\in\mathbb{Q}_{p}$, $\widehat{J}(\xi)=\widehat{J}%
		(\xi+\xi_{0})$ for any $\xi\in\mathbb{Z}_{p}$. In particular $\widehat
		{J}(0)=\widehat{J}(\xi_{0})=\left\Vert J(\left\vert x\right\vert
		_{p})\right\Vert _{1}$, for $\xi_{0}\in\mathbb{Z}_{p}$. Then $\widehat{J}%
		(\xi)=\widehat{J}(\left\vert \xi\right\vert _{p})$ for $\xi\in\mathbb{Q}_{p}$.
	\end{lemma}
	
	\begin{proof}
		(i) By using the partition,
		\[
		\mathbb{Z}_{p}\smallsetminus\left\{  0\right\}  =%
		{\displaystyle\bigsqcup\limits_{j=0}^{\infty}}
		p^{j}\mathbb{Z}_{p}^{\times},
		\]
		and the fact that $\mu_{\text{Haar}}\left(  \left\{  0\right\}  \right)  =0$,
		\begin{align}
			\widehat{J}(\xi)  &  =%
			{\displaystyle\int\limits_{\mathbb{Z}_{p}}}
			J(\left\vert x\right\vert _{p})\chi_{p}\left(  \xi x\right)  dx=%
			{\displaystyle\int\limits_{\mathbb{Z}_{p}\smallsetminus\left\{  0\right\}  }}
			J(\left\vert x\right\vert _{p})\chi_{p}\left(  \xi x\right)  dx=%
			{\displaystyle\sum\limits_{j=0}^{\infty}}
			\text{ }%
			{\displaystyle\int\limits_{p^{j}\mathbb{Z}_{p}^{\times}}}
			J(\left\vert x\right\vert _{p})\chi_{p}\left(  \xi x\right)  dx\nonumber\\
			&  =%
			{\displaystyle\sum\limits_{j=0}^{\infty}}
			J(p^{-j})%
			{\displaystyle\int\limits_{p^{j}\mathbb{Z}_{p}^{\times}}}
			\chi_{p}\left(  \xi x\right)  dx=%
			{\displaystyle\sum\limits_{j=0}^{\infty}}
			J(p^{-j})\left\{  \text{ }%
			{\displaystyle\int\limits_{p^{j}\mathbb{Z}_{p}}}
			\chi_{p}\left(  \xi x\right)  dx-%
			{\displaystyle\int\limits_{p^{j+1}\mathbb{Z}_{p}}}
			\chi_{p}\left(  \xi x\right)  dx\right\}  . \label{Eq_J_FT}%
		\end{align}
		Now, the announced formula follows from (\ref{Eq_J_FT}), by the following
		formula:
		\[%
		{\displaystyle\int\limits_{p^{j}\mathbb{Z}_{p}}}
		\chi_{p}\left(  \xi x\right)  dx=p^{-j}%
		{\displaystyle\int\limits_{\mathbb{Z}_{p}}}
		\chi_{p}\left(  p^{j}\xi y\right)  dy=p^{-j}\Omega\left(  p^{-j}\left\vert
		\xi\right\vert _{p}\right)  \text{, where }x=p^{j}y,dx=p^{-j}dy.
		\]
		(ii) By changing variables as $y=ux$, $dy=dx$, and using that $\left\vert
		u^{-1}y\right\vert _{p}$ $=\left\vert y\right\vert _{p}$, for any
		$u\in\mathbb{Z}_{p}^{\times}$,
		\[
		\widehat{J}(u\xi)=%
		{\displaystyle\int\limits_{\mathbb{Z}_{p}}}
		J(\left\vert x\right\vert _{p})\chi_{p}\left(  \xi ux\right)  dx=%
		{\displaystyle\int\limits_{\mathbb{Z}_{p}}}
		J(\left\vert y\right\vert _{p})\chi_{p}\left(  \xi y\right)  dy=\widehat
		{J}(\xi).
		\]
		(iii) If $\xi_{0}\in\mathbb{Z}_{p}$, then $\chi_{p}\left(  \xi_{0}x\right)
		=1$ for any $x\in\mathbb{Z}_{p}$, consequently, $\widehat{J}(\xi)=\widehat
		{J}(\xi+\xi_{0})$ for any $\xi\in\mathbb{Z}_{p}$. Finally, for $\xi_{0}%
		\in\mathbb{Z}_{p}$, \
		\[
		\widehat{J}(\xi_{0})=%
		{\displaystyle\int\limits_{\mathbb{Z}_{p}}}
		J(\left\vert x\right\vert _{p})dx=\left\Vert J(\left\vert x\right\vert
		_{p})\right\Vert _{1}=\widehat{J}(0).
		\]
		
	\end{proof}
	
	\begin{remark}
		As a corollary, the function $F(\xi)=\exp\left(  -t\left(  1-\widehat
		{J}(\left\vert \xi\right\vert _{p})\right)  \right)  $, for $t>0$, is a
		locally constant function in $\mathbb{Q}_{p}$, satisfying $F(\xi)=F(\xi
		+\xi_{0})$, for $\xi\in\mathbb{Q}_{p}$, $\xi_{0}\in\mathbb{Z}_{p}$. Then,
		$F(\xi)$ defines a distribution from $\mathcal{D}^{\prime}\left(
		\mathbb{Q}_{p}\right)  $, whose Fourier transform is a distribution supported
		in the unit ball,%
		\[
		Z_{0}\left(  x,t\right)  :=\left\{
		\begin{array}
			[c]{ll}%
			\mathcal{F}_{\xi\rightarrow x}^{-1}\left(  \mathrm{e}^{-t\left(  1-\widehat
				{J}(\left\vert \xi\right\vert _{p})\right)  }\right)  \in\mathcal{D}^{\prime
			}\left(  \mathbb{Z}_{p}\right)  , & \text{if }t>0.\\
			& \\
			\delta\left(  x\right)   & \text{if }t=0,
		\end{array}
		\right.
		\]
		where $\delta\left(  x\right)  $ is the Dirac distribution, see, e.g.,
		\cite[Section VII.3]{V-V-Z}, \cite[Theorem 4.9.3]{Alberio et al},
		\cite[Theorem 1.161]{Zuniga-Textbook}. Notice that $\left\vert \widehat
		{J}(\left\vert \xi\right\vert _{p})\right\vert \leq1$, for any $\xi
		\in\mathbb{Q}_{p}$.
	\end{remark}
	
	\begin{lemma}
		\label{Lemma_2A}With the above notation, the following formula holds:%
		\begin{align}
			Z_{0}\left(  x,t\right)   &  =\Omega\left(  \left\vert x\right\vert
			_{p}\right)  +%
			{\displaystyle\sum\limits_{j=1}^{\infty}}
			\mathrm{e}^{-t\left(  1-\widehat{J}\left(  p^{j}\right)  \right)  }\left\{  p^{j}%
			\Omega\left(  p^{j}\left\vert x\right\vert _{p}\right)  -p^{\left(
				j-1\right)  }\Omega\left(  p^{\left(  j-1\right)  }\left\vert x\right\vert
			_{p}\right)  \right\}  \label{Z_0(x,t)}\\
			&  =\Omega\left(  \left\vert x\right\vert _{p}\right)  -\mathrm{e}^{-t\left(
				1-\widehat{J}\left(  p\right)  \right)  }1_{\mathbb{Z}_{p}^{\times}}\left(
			x\right)  +\Omega\left(  p\left\vert x\right\vert _{p}\right)
			{\displaystyle\sum\limits_{j=1}^{\infty}}
			p^{j}\mathrm{e}^{-t\left(  1-\widehat{J}\left(  p^{j}\right)  \right)  },\nonumber
		\end{align}
		for $x\in\mathbb{Z}_{p}\smallsetminus\left\{  0\right\}  $, and $t>0$.
	\end{lemma}
	
	\begin{remark}
		We take $x\neq0$, due to the fact, that for $\ x=0$, the convergence of the
		series%
		\[
		\left(  1-p^{-1}\right)
		{\displaystyle\sum\limits_{j=1}^{\infty}}
		p^{j}\mathrm{e}^{-t\left(  1-\widehat{J}\left(  p^{j}\right)  \right)  }%
		\]
		is unknown,\ i.e., we require an extra hypothesis on $\widehat{J}$ to
		determine it. Also, for $x\neq0$, the series in (\ref{Z_0(x,t)}) is a finite
		sum due to the fact that $p^{j}\Omega\left(  p^{j}\left\vert x\right\vert
		_{p}\right)  -p^{\left(  j-1\right)  }\Omega\left(  p^{\left(  j-1\right)
		}\left\vert x\right\vert _{p}\right)  =0$ for $j\geq2+ord(x)$.
	\end{remark}
	
	\begin{proof}
		By using that $\Omega\left(  p^{-j}\left\vert x\right\vert _{p}\right)
		-\Omega\left(  p^{-\left(  j-1\right)  }\left\vert x\right\vert _{p}\right)  $
		is the characteristic function of the set $S_{j}=p^{-j}\mathbb{Z}%
		_{p}\smallsetminus p^{-\left(  j-1\right)  }\mathbb{Z}_{p}$, we have
		\begin{align*}
			\mathrm{e}^{-t\left(  1-\widehat{J}(\left\vert \xi\right\vert _{p})\right)  } &
			=\Omega\left(  \left\vert \xi\right\vert _{p}\right)  +%
			{\displaystyle\sum\limits_{j=1}^{\infty}}
			\mathrm{e}^{-t\left(  1-\widehat{J}(p^{j})\right)  }1_{S_{j}}\left(  \xi\right)  \\
			&  =\Omega\left(  \left\vert \xi\right\vert _{p}\right)  +%
			{\displaystyle\sum\limits_{j=1}^{\infty}}
			\mathrm{e}^{-t\left(  1-\widehat{J}(p^{j})\right)  }\left\{  \Omega\left(
			p^{-j}\left\vert x\right\vert _{p}\right)  -\Omega\left(  p^{-\left(
				j-1\right)  }\left\vert x\right\vert _{p}\right)  \right\}  .
		\end{align*}

		Now,
		\[
		\Omega\left(  \left\vert \xi\right\vert _{p}\right)  +%
		{\displaystyle\sum\limits_{j=1}^{m}}
		\mathrm{e}^{-t\left(  1-\widehat{J}(p^{j})\right)  }\left\{  \Omega\left(
		p^{-j}\left\vert x\right\vert _{p}\right)  -\Omega\left(  p^{-\left(
			j-1\right)  }\left\vert x\right\vert _{p}\right)  \right\}  \rightarrow
		\mathrm{e}^{-t\left(  1-\widehat{J}(\left\vert \xi\right\vert _{p})\right)  }%
		\]
		as $m\rightarrow\infty$ in $\mathcal{D}^{\prime}\left(  \mathbb{Q}_{p}\right)
		$, and since the Fourier transform is a topological and algebraic isomorphism
		on $\mathcal{D}^{\prime}\left(  \mathbb{Q}_{p}\right)  $,%
		\begin{multline*}
			\mathcal{F}_{\xi\rightarrow x}^{-1}\left(  \Omega\left(  \left\vert
			\xi\right\vert _{p}\right)  \right)  +%
			{\displaystyle\sum\limits_{j=1}^{m}}
			\mathrm{e}^{-t\left(  1-\widehat{J}(p^{j})\right)  }\mathcal{F}_{\xi\rightarrow x}%
			^{-1}\left\{  \Omega\left(  p^{-j}\left\vert x\right\vert _{p}\right)
			-\Omega\left(  p^{-\left(  j-1\right)  }\left\vert x\right\vert _{p}\right)
			\right\}  \\
			\rightarrow\mathcal{F}_{\xi\rightarrow x}^{-1}\left(  \mathrm{e}^{-t\left(
				1-\widehat{J}(\left\vert \xi\right\vert _{p})\right)  }\right)
		\end{multline*}
		as $m\rightarrow\infty$ \ in $\mathcal{D}^{\prime}\left(  \mathbb{Q}%
		_{p}\right)  $. Now, the announced formula follows from
		\begin{align*}
			&  \mathcal{F}_{\xi\rightarrow x}^{-1}\left(  \Omega\left(  \left\vert
			\xi\right\vert _{p}\right)  \right)  +%
			{\displaystyle\sum\limits_{j=1}^{m}}
			\mathrm{e}^{-t\left(  1-\widehat{J}(p^{j})\right)  }\mathcal{F}_{\xi\rightarrow x}%
			^{-1}\left\{  \Omega\left(  p^{-j}\left\vert x\right\vert _{p}\right)
			-\Omega\left(  p^{-\left(  j-1\right)  }\left\vert x\right\vert _{p}\right)
			\right\}  \\
			&  =\Omega\left(  \left\vert
			x\right\vert _{p}\right)  +%
			{\displaystyle\sum\limits_{j=1}^{m}}
			\mathrm{e}^{-t\left(  1-\widehat{J}\left(  p^{j}\right)  \right)  }\left\{  p^{j}%
			\Omega\left(  p^{j}\left\vert x\right\vert _{p}\right)  -p^{\left(
				j-1\right)  }\Omega\left(  p^{\left(  j-1\right)  }\left\vert x\right\vert
			_{p}\right)  \right\}  ,
		\end{align*}
		by using that%
		\[
		\mathcal{F}_{\xi\rightarrow x}^{-1}\left\{  \Omega\left(  p^{k}\left\vert
		x\right\vert _{p}\right)  =p^{-k}\Omega\left(  p^{-k}\left\vert x\right\vert
		_{p}\right)  \right\}  .
		\]
		
	\end{proof}
	
	The following result is established by using the above technique.
	
	\begin{lemma}
		\label{Lemma_3A}Set%
		\[
		\boldsymbol{J}\left(  Z_{0}\left(  x,t\right)  \right)  :=\mathcal{F}%
		_{\xi\rightarrow x}^{-1}\left(  \left\{  \widehat{J}(\left\vert \xi\right\vert
		_{p})-1\right\}  \mathrm{e}^{-t\left(  1-\widehat{J}(\left\vert \xi\right\vert
			_{p})\right)  }\right)  \in\mathcal{D}^{\prime}\left(  \mathbb{Z}_{p}\right)
		,
		\]
		for $t>0$. Then%
		\[
		\boldsymbol{J}\left(  Z_{0}\left(  x,t\right)  \right)  =%
		{\displaystyle\sum\limits_{j=1}^{\infty}}
		\mathrm{e}^{-t\left(  1-\widehat{J}\left(  p^{j}\right)  \right)  }\left(  \widehat
		{J}\left(  p^{j}\right)  -1\right)  \left\{  p^{j}\Omega\left(  p^{j}%
		\left\vert x\right\vert _{p}\right)  -p^{\left(  j-1\right)  }\Omega\left(
		p^{\left(  j-1\right)  }\left\vert x\right\vert _{p}\right)  \right\}  ,
		\]
		for $x\in\mathbb{Z}_{p}\smallsetminus\left\{  0\right\}  $, and $t>0$.
	\end{lemma}
	
	\begin{lemma}
		\label{Lemma_4A}%
		\[%
		{\displaystyle\int\limits_{\mathbb{Z}_{p}}}
		Z_{0}(x,t)dx=1\text{, for }t>0\text{.}%
		\]
		
	\end{lemma}
	
	\begin{proof}
		By Lemma (\ref{Lemma_1A}), for $t>0$,
		\begin{align*}
			\left(  Z_{0}(x,t),\Omega\left(  \left\vert x\right\vert _{p}\right)  \right)
			&  =\left(  \mathcal{F}_{\xi\rightarrow x}^{-1}\left( \mathrm{e}^{-t\left(
				1-\widehat{J}(\left\vert \xi\right\vert _{p})\right)  }\right)  ,\Omega\left(
			\left\vert x\right\vert _{p}\right)  \right)  =\left(  \mathrm{e}^{-t\left(
				1-\widehat{J}(\left\vert \xi\right\vert _{p})\right)  },\Omega\left(
			\left\vert \xi\right\vert _{p}\right)  \right)  \\
			&  =%
			{\displaystyle\int\limits_{\mathbb{Z}_{p}}}
			\mathrm{e}^{-t\left(  1-\widehat{J}(\left\vert \xi\right\vert _{p})\right)  }%
			d\xi=\mathrm{e}^{-t\left(  1-\widehat{J}(0)\right)  }%
			{\displaystyle\int\limits_{\mathbb{Z}_{p}}}
			d\xi=1.
		\end{align*}
		
	\end{proof}
	
	\begin{lemma}
		\label{Lemma_5A}(i) For $x\in I+p^{l}\mathbb{Z}_{p}$,%
		\[
		Z_{0}\left(  x,t\right)  \ast\Omega\left(  p^{l}\left\vert x-I\right\vert
		_{p}\right)  =1-\left(  1-p^{-1}\right)  \sum_{j=0}^{l-1}p^{-j}Z_{0}\left(  p^{-j},t\right)  ,
		\]
		for $t>0$ (for $l=1$ this reduces to $1-Z_{0}\left(  1,t\right)  \left(  1-p^{-1}\right)  $).

		\noindent(ii) For $x\in K+p^{l}\mathbb{Z}_{p}$, $K\neq I$ as elements from
		$G_{l}$,%
		\[
		Z_{0}\left(  x,t\right)  \ast\Omega\left(  p^{l}\left\vert x-I\right\vert
		_{p}\right)  =p^{-l}Z_{0}\left(  \left\vert K-I\right\vert _{p},t\right)  ,
		\]
		for $t>0$.
		
		(iii) For $x\in\mathbb{Z}_{p}$, and $t>0$,
		\begin{align*}
			Z_{0}\left(  x,t\right)  \ast\Omega\left(  p^{l}\left\vert x-I\right\vert
			_{p}\right)   &  =\left[  1-\left(  1-p^{-1}\right)  \sum_{j=0}^{l-1}p^{-j}Z_{0}\left(  p^{-j},t\right)  \right]
			\Omega\left(  p^{l}\left\vert x-I\right\vert _{p}\right)  +\\
			&
			{\displaystyle\sum\limits_{\substack{K\in G_{l}\\K\neq I}}}
			p^{-l}Z_{0}\left(  \left\vert K-I\right\vert _{p},t\right)  \Omega\left(
			p^{l}\left\vert x-K\right\vert _{p}\right)  .
		\end{align*}
		\noindent(iv) The space $\mathcal{D}_{l}(\mathbb{Z}_{p})$ is invariant under the operator
		$\varphi\left(  x\right)  \rightarrow Z_{0}(x,t)\ast\varphi\left(  x\right)  $.
		
	\end{lemma}
	
	\begin{proof}
		Set%
		\[
		Z_{0}^{\left(  m\right)  }\left(  x,t\right)  =\Omega\left(  \left\vert
		x\right\vert _{p}\right)  +%
		{\displaystyle\sum\limits_{j=1}^{m}}
		\mathrm{e}^{-t\left(  1-\widehat{J}\left(  p^{j}\right)  \right)  }\left\{  p^{j}%
		\Omega\left(  p^{j}\left\vert x\right\vert _{p}\right)  -p^{\left(
			j-1\right)  }\Omega\left(  p^{\left(  j-1\right)  }\left\vert x\right\vert
		_{p}\right)  \right\}  .
		\]
		Then, $Z_{0}^{\left(  m\right)  }\left(  x,t\right)  \rightarrow Z_{0}\left(
		x,t\right)  $ as $m\rightarrow\infty$ in $\mathcal{D}^{\prime}\left(
		\mathbb{Q}_{p}\right)  $. We now recall that the convolution of a distribution
		with a test function is a locally constant function (see, e.g. \cite[Theorem
		1.154]{Zuniga-Textbook}) given by%
		\begin{gather*}
			Z_{0}\left(  x,t\right)  \ast\Omega\left(  p^{l}\left\vert x-I\right\vert
			_{p}\right)  =\left(  Z_{0}\left(  y,t\right)  ,\Omega\left(  p^{l}\left\vert
			x-I-y\right\vert _{p}\right)  \right)  \\
			=\lim_{m\rightarrow\infty}\left(  Z_{0}^{\left(  m\right)  }\left(
			y,t\right)  ,\Omega\left(  p^{l}\left\vert x-I-y\right\vert _{p}\right)
			\right)  =\lim_{m\rightarrow\infty}\text{ }%
			{\displaystyle\int\limits_{x-I+p^{l}\mathbb{Z}p}}
			\text{ }Z_{0}^{\left(  m\right)  }\left(  y,t\right)  dy.
		\end{gather*}
		We first consider the case, $x\in I+p^{l}\mathbb{Z}_{p}$, i.e. $x-I+p^{l}%
		\mathbb{Z}_{p}=p^{l}\mathbb{Z}_{p}$,%
		\[
		Z_{0}\left(  x,t\right)  \ast\Omega\left(  p^{l}\left\vert x-I\right\vert
		_{p}\right)  =\lim_{m\rightarrow\infty}\text{ }%
		{\displaystyle\int\limits_{p^{l}\mathbb{Z}p}}
		\text{ }Z_{0}^{\left(  m\right)  }\left(  y,t\right)  dy=%
		{\displaystyle\int\limits_{p^{l}\mathbb{Z}p}}
		\text{ }Z_{0}\left(  y,t\right)  dy.
		\]
		Now, we use \ that $Z_{0}\left(  y,t\right)  =$ $Z_{0}\left(  \left\vert
		y\right\vert _{p},t\right)  $, and Lemma \ref{Lemma_4A}. For $l=1$, $%
		\mathbb{Z}_{p}=\mathbb{Z}_{p}^{\times}\sqcup p\mathbb{Z}_{p}$, but for
		general $l\geq1$ one must use the finer partition $\mathbb{Z}_{p}%
		=p^{l}\mathbb{Z}_{p}\sqcup\bigsqcup_{j=0}^{l-1}p^{j}\mathbb{Z}_{p}^{\times}$
		(the case $l=1$ being the partition with only the $j=0$ term $\mathbb{Z}%
		_{p}^{\times}$ besides $p\mathbb{Z}_{p}$); thus
		\[
		1=%
		{\displaystyle\int\limits_{\mathbb{Z}_{p}}}
		Z_{0}\left(  \left\vert y\right\vert _{p},t\right)  dy=%
		{\displaystyle\int\limits_{p^{l}\mathbb{Z}p}}
		\text{ }Z_{0}\left(  \left\vert y\right\vert _{p},t\right)  dy+%
		\sum_{j=0}^{l-1}\text{ }%
		{\displaystyle\int\limits_{p^{j}\mathbb{Z}p^{\times}}}
		\text{ }Z_{0}\left(  \left\vert y\right\vert _{p},t\right)  dy,
		\]
		for $t>0$. Since $Z_{0}(\cdot,t)$ is constant equal to $Z_{0}(p^{-j},t)$ on
		each sphere $p^{j}\mathbb{Z}_{p}^{\times}$, of Haar measure $p^{-j}\left(
		1-p^{-1}\right)  $, we get
		\[
		1=%
		{\displaystyle\int\limits_{p^{l}\mathbb{Z}p}}
		\text{ }Z_{0}\left(  \left\vert y\right\vert _{p},t\right)  dy+\left(
		1-p^{-1}\right)  \sum_{j=0}^{l-1}p^{-j}Z_{0}\left(  p^{-j},t\right)  ,
		\]
		for $t>0$, therefore%
		\[
		Z_{0}\left(  x,t\right)  \ast\Omega\left(  p^{l}\left\vert x-I\right\vert
		_{p}\right)  =%
		{\displaystyle\int\limits_{p^{l}\mathbb{Z}p}}
		\text{ }Z_{0}\left(  \left\vert y\right\vert _{p},t\right)  dy=1-\left(
		1-p^{-1}\right)  \sum_{j=0}^{l-1}p^{-j}Z_{0}\left(  p^{-j},t\right)  ,
		\]
		for $t>0$. (For $l=1$ this reduces to $1-\left(  1-p^{-1}\right)
		Z_{0}\left(  1,t\right)  $.)
		
		We now \ consider the case, $x\in K+p^{l}\mathbb{Z}_{p}$, i.e., $K\neq I$ in
		$G_{l}=p^{l}\mathbb{Z}_{p}$, \ by using that $Z_{0}^{\left(  m\right)
		}\left(  y,t\right)  =Z_{0}^{\left(  m\right)  }\left(  \left\vert
		y\right\vert _{p},t\right)  $,%
		\[
		Z_{0}\left(  x,t\right)  \ast\Omega\left(  p^{l}\left\vert x-I\right\vert
		_{p}\right)  =\lim_{m\rightarrow\infty}\text{ }%
		{\displaystyle\int\limits_{p^{l}\mathbb{Z}p}}
		\text{ }Z_{0}^{\left(  m\right)  }\left(  y,t\right)  dy=p^{-l}Z_{0}\left(
		\left\vert K-I\right\vert _{p},t\right)  .
		\]
		
	\end{proof}
	
	\subsection{$p$-Adic heat equations and Markov processes in the unit ball}
	
	With $J$ as before, the mapping%
	\[%
	\begin{array}
		[c]{cccc}%
		\boldsymbol{J}: & L^{\rho}\left(  \mathbb{Z}_{p}\right)   & \rightarrow &
		L^{\rho}\left(  \mathbb{Z}_{p}\right)  \\
		&  &  & \\
		& \varphi\left(  x\right)   & \rightarrow & J\left(  \left\vert x\right\vert
		_{p}\right)  \ast\varphi\left(  x\right)  -\varphi\left(  x\right)
	\end{array}
	\]
	gives rise to a well-defined \ linear, bounded operator, for $\rho\in\left[
	0,\infty\right]  $. Indeed,%
	\[
	\left\Vert \boldsymbol{J}\varphi\left(  x\right)  \right\Vert _{\rho
	}=\left\Vert J\ast\varphi-\varphi\right\Vert _{\rho}\leq\left\Vert
	J\ast\varphi\right\Vert _{\rho}+\left\Vert \varphi\right\Vert _{\rho}%
	\leq\left\Vert J\right\Vert _{1}\left\Vert \varphi\right\Vert _{\rho
	}+\left\Vert \varphi\right\Vert _{\rho}=2\left\Vert \varphi\right\Vert _{\rho
	}.
	\]

	Furthermore, $\boldsymbol{J}$ is a pseudo-differential operator,%
	\[
	\boldsymbol{J}\varphi\left(  x\right)  =\mathcal{F}_{\xi\rightarrow x}%
	^{-1}\left(  \left(  \widehat{J}(\left\vert \xi\right\vert _{p})-1\right)
	\mathcal{F}_{x\rightarrow\xi}\varphi\right)  =%
	{\displaystyle\int\limits_{\mathbb{Q}_{p}}}
	\left(  \widehat{J}(\left\vert \xi\right\vert _{p})-1\right)  \widehat
	{\varphi}\left(  \xi\right)  \chi_{p}\left(  -\xi x\right)  d\xi,
	\]
	for $\varphi\in\mathcal{D}\left(  \mathbb{Z}_{p}\right)  $. Notice that
	$\boldsymbol{J}\varphi\in\mathcal{D}\left(  \mathbb{Z}_{p}\right)  $, indeed,
	$\widehat{J}(\left\vert \xi\right\vert _{p})-1$ is a locally constant
function, then its product with $\widehat{\varphi}\left(  \xi\right)  $ is a
test function as well, since the product of a locally constant function with a test function is again a test function. Now, we use the fact that the space of test functions is invariant under the Fourier transform.
	
	\begin{proposition}
		\label{Prop_1A}The function
		\[
		u\left(  x,t\right)  =Z_{0}\left(  x,t\right)  \ast u_{0}\left(  x\right)
		\text{, }x\in\mathbb{Z}_{p},t\geq0\text{,}%
		\]
		is a solution of the Cauchy problem
		\begin{equation}
			\left\{
			\begin{array}
				[c]{l}%
				u\left(  \cdot,t\right)  \in\mathcal{D}(\mathbb{Z}_{p})\text{ for }%
				t\geq0\text{, }u\left(  x,\cdot\right)  \in C^{1}\left[  0,\infty\right) \\
				\\
				\frac{\partial}{\partial t}u\left(  x,t\right)  =\boldsymbol{J}u\left(
				x,t\right)  \text{, }x\in\mathbb{Z}_{p},t\geq0\\
				\\
				u\left(  x,0\right)  =u_{0}\left(  x\right)  \in\mathcal{D}(\mathbb{Z}_{p}).
			\end{array}
			\right.  \label{Cauchy_Problem_4}%
		\end{equation}
		
	\end{proposition}
	
	\begin{proof}
		Since $u_{0}\left(  x\right)  \in\mathcal{D}(\mathbb{Z}_{p})=\cup
		_{l}\mathcal{D}_{l}(\mathbb{Z}_{p})$, we may assume that $u_{0}\left(
		x\right)  \in\mathcal{D}_{l}(\mathbb{Z}_{p})$ for some positive integer $l$. On the other hand, $\frac{\partial}{\partial t}$, $\boldsymbol{J}$ are linear
		operators, and $\left\{  \Omega\left(  p^{l}\left\vert x-I\right\vert
		_{p}\right)  \right\}  _{I\in G_{l}}$ is a basis for $\mathcal{D}%
		_{l}(\mathbb{Z}_{p})$, we may assume without loss of generality that
		$u_{0}\left(  x\right)  =\Omega\left(  p^{l}\left\vert x-I\right\vert
		_{p}\right)  $.
		
		We now show that $u\left(  x,t\right)  =Z_{0}\left(  x,t\right)  \ast
		\Omega\left(  p^{l}\left\vert x-I\right\vert _{p}\right)  $, $x\in
		\mathbb{Z}_{p},t>0$, is a solution of the differential equation in
		(\ref{Cauchy_Problem_4}). Indeed, by Lemmas \ref{Lemma_3A}, \ref{Lemma_5A},%
		\[
		\frac{\partial}{\partial t}u\left(  x,t\right)  =\frac{\partial}{\partial
			t}\left\{  Z_{0}\left(  x,t\right)  \ast\Omega\left(  p^{l}\left\vert
		x-I\right\vert _{p}\right)  \right\}  =\left(  \frac{\partial}{\partial
			t}Z_{0}\left(  x,t\right)  \right)  \ast\Omega\left(  p^{l}\left\vert
		x-I\right\vert _{p}\right)  ,
		\]

		in $\mathcal{D}^{\prime}\left(  \mathbb{Z}_{p}\right)  $. Now, by Lemmas
		\ref{Lemma_2A} and \ref{Lemma_3A},%
		\[
		\frac{\partial}{\partial t}Z_{0}\left(  x,t\right)  =\boldsymbol{J}\left(
		Z_{0}\left(  x,t\right)  \right)  ,
		\]
		in $\mathcal{D}^{\prime}\left(  \mathbb{Z}_{p}\right)  $, and by convoluting
		both sides with $\Omega\left(  p^{l}\left\vert x-I\right\vert _{p}\right)  $,
		we get
		\[
		\frac{\partial}{\partial t}u\left(  x,t\right)  =\boldsymbol{J}\left(
		Z_{0}\left(  x,t\right)  \right)  \ast\Omega\left(  p^{l}\left\vert
		x-I\right\vert _{p}\right)  .
		\]
		On the other hand,
		\begin{multline*}
			\boldsymbol{J}\left(  Z_{0}\left(  x,t\right)  \ast\Omega\left(
			p^{l}\left\vert x-I\right\vert _{p}\right)  \right)  =\\
			\mathcal{F}_{\xi\rightarrow x}^{-1}\left(  \left\{  \widehat{J}(\left\vert
			\xi\right\vert _{p})-1\right\}  e^{-t\left(  1-\widehat{J}(\left\vert
				\xi\right\vert _{p})\right)  }p^{-l}\chi_{p}\left(  \xi I\right)
			\Omega\left(  p^{-l}\left\vert x\right\vert _{p}\right)  \right)  \\
			=\mathcal{F}_{\xi\rightarrow x}^{-1}\left(  \left\{  \widehat{J}(\left\vert
			\xi\right\vert _{p})-1\right\}  e^{-t\left(  1-\widehat{J}(\left\vert
				\xi\right\vert _{p})\right)  }\right)  \ast\Omega\left(  p^{l}\left\vert
			x-I\right\vert _{p}\right)  \\
			=\boldsymbol{J}\left(  Z_{0}\left(  x,t\right)  \right)  \ast\Omega\left(
			p^{l}\left\vert x-I\right\vert _{p}\right)  ,
		\end{multline*}
		in $\mathcal{D}^{\prime}\left(  \mathbb{Z}_{p}\right)  $, i.e., for $t>0$,%
		\[
		\frac{\partial}{\partial t}u\left(  x,t\right)  =\boldsymbol{J}u\left(
		x,t\right)  \text{ in }\mathcal{D}^{\prime}\left(  \mathbb{Z}_{p}\right)  .
		\]
		The initial condition in (\ref{Cauchy_Problem_4}) is satisfied because
		$Z_{0}\left(  x,0\right)  =\delta\left(  x\right)  $. Finally, $Z_{0}\left(
		x,t\right)  \ast u_{0}\left(  x\right)  $ is a test function for every
		$t\geq0$, cf. Lemma \ref{Lemma_5A}-(iii).
	\end{proof}
	
	We now review a result from (\cite[Theorem 3.1]{Zuniga-networks}). We denote
	by $\mathcal{C}(\mathbb{Z}_{p})$ the $\mathbb{R}$-vector space of continuous
	functions on the unit ball. We recall that $\mathcal{D}\left(  \mathbb{Z}%
	_{p}\right)  $ is dense in $\mathcal{C}(\mathbb{Z}_{p})$. We fix a time
	horizon $T\in\left(  0,\infty\right]  $, and consider the following initial
	value problem:%
	\begin{equation}
		\left\{
		\begin{array}
			[c]{l}%
			u\left(  \cdot,t\right)  \in\mathcal{C}(\mathbb{Z}_{p})\text{ for }%
			t\geq0\text{, }u\left(  x,\cdot\right)  \in\mathcal{C}^{1}\left(  \left[
			0,T\right]  \right)  \text{ for }x\in\mathbb{Z}_{p}\\
			\\
			\frac{\partial}{\partial t}u\left(  x,t\right)  =J\left(  \left\vert
			x\right\vert _{p}\right)  \ast u\left(  x,t\right)  -u\left(  x,t\right)
			\text{, }x\in\mathbb{Z}_{p},t\in\left[  0,T\right]  \\
			\\
			u\left(  x,0\right)  =u_{0}\left(  x\right)  \in\mathcal{C}(\mathbb{Z}_{p}).
		\end{array}
		\right.  \label{Cauchy_Problem_5}%
	\end{equation}

	The operator $\boldsymbol{J}\varphi\left(  x\right)  =J\left(  \left\vert
	x\right\vert _{p}\right)  \ast\varphi\left(  x\right)  -\varphi\left(
	x\right)  $, for  $\varphi\left(  x\right)  \in\mathcal{C}(\mathbb{Z}_{p})$,
	satisfies the maximum principle, i.e., \ if $\varphi\left(  x_{0}\right)
	=\sup_{x\in\mathbb{Z}_{p}}\varphi\left(  x\right)  $, then%
	\[
	\boldsymbol{J}\varphi\left(  x_{0}\right)  ={\displaystyle\int\limits_{\mathbb{Z}_{p}}}
	J\left(  \left\vert x_{0}-y\right\vert _{p}\right)  \left\{  \varphi\left(
	y\right)  -\varphi\left(  x_{0}\right)  \right\}  dy\leq0\text{.}%
	\]

	The condition $J\left(  \left\vert x\right\vert _{p}\right)  \leq1$,
	$x\in\mathbb{Z}_{p}$, was used in (\cite[Theorem 3.1]{Zuniga-networks}) to
	assure that operator $\boldsymbol{J}$ satisfies the maximum principle.
	Consequently, here we do not need this condition. By (\cite[Theorem
	3.1]{Zuniga-networks}), there exists a probability measure $p_{t}\left(
	x,\cdot\right)  $, $t\in\left[  0,T\right]  $, with $T=T(u_{0})$,
	$x\in\mathbb{Z}_{p}$, on the Borel $\sigma$-algebra $\mathcal{B}$ of
	$\mathbb{Z}_{p}$, such that the Cauchy problem (\ref{Cauchy_Problem_5}) has a
	unique solution of the form%
	\[
	u(x,t)=\int\limits_{\mathbb{Z}_{p}}u_{0}(y)p_{t}\left(  x,dy\right)  .
	\]
	In addition, $p_{t}\left(  x,\cdot\right)  $ is the transition function of a
	Markov process $\mathfrak{X}$ whose paths are right continuous and have no
	discontinuities other than jumps. We now determine the measure $p_{t}\left(
	x,\cdot\right)  $.
	
	By comparing the solutions of Cauchy problems (\ref{Cauchy_Problem_4}) and
	(\ref{Cauchy_Problem_5}), we conclude that
	\[
	\int\limits_{\mathbb{Z}_{p}}1_{B}(y)p_{t}\left(  x,dy\right)  =\int
	\limits_{\mathbb{Z}_{p}}1_{B}(y)Z_{0}\left(  x-y,t\right)  dy\text{, }%
	\]
	for $t\in\left[  0,T\right]  $ and any open compact subset $B$ of
	$\mathbb{Z}_{p}$. Then $p_{t}\left(  x,dy\right)  $ exists for any
	$t\in\left[  0,\infty\right]  $, and since the open compact subset $B$ of
	$\mathbb{Z}_{p}$ generate the Borel $\sigma$-algebra $\mathcal{B}$ of
	$\mathbb{Z}_{p}$, then $p_{t}\left(  x,dy\right)  =Z_{0}\left(  x-y,t\right)
	dy$ for $t\geq0$. \ So we have the following result:
	
	\begin{theorem}
		\label{Theorem_1}The Cauchy problem%
		\begin{equation}
			\left\{
			\begin{array}
				[c]{l}%
				u\left(  \cdot,t\right)  \in\mathcal{C}(\mathbb{Z}_{p})\text{ for }%
				t\geq0\text{, }u\left(  x,\cdot\right)  \in C^{1}\left[  0,\infty\right)  \\
				\\
				\frac{\partial}{\partial t}u\left(  x,t\right)  =J\left(  \left\vert
				x\right\vert _{p}\right)  \ast u\left(  x,t\right)  -u\left(  x,t\right)
				\text{, }x\in\mathbb{Z}_{p},t\geq0\\
				\\
				u\left(  x,0\right)  =u_{0}\left(  x\right)  \in\mathcal{C}(\mathbb{Z}_{p}).
			\end{array}
			\right.  \label{Eq_Cauchy_problem_2}%
		\end{equation}
		has a unique solution of the form $u\left(  x,t\right)  =Z_{0}\left(
		x,t\right)  \ast u_{0}\left(  x\right)  $. Furthermore,
		\begin{equation}
			p(t,x,B)=%
			\begin{cases}
				\int_{B}Z_{0}(x-y,t)dy=1_{B}\left(  x\right)  \ast Z_{0}(x,t) & \text{for
				}t>0,\quad x\in\mathbb{Z}_{p},\quad B\in\mathcal{B}\\
				1_{B}(x) & \text{for }t=0,
			\end{cases}
			\label{Probability}%
		\end{equation}
		is the transition function of a Markov process $\mathfrak{X}$ whose paths are
		right-continuous and have no discontinuities other than jumps. Furthermore,
		\[%
		{\displaystyle\int\limits_{\mathbb{Z}_{p}}}
		Z_{0}(x,t)dx=1\text{, for }t>0.
		\]
		
	\end{theorem}
	
	\begin{remark}
		The value $p(t,x,B)$ gives the transition probability that
		a particle/walker starting at the point $x\in\mathbb{Z}_{p}$ will be found in
		the set $B$ at the time $t$.
	\end{remark}

	\section{Appendix C: Continuous-time Markov chains}
	
	We take $\mathcal{K}_{r}$, $r\in\mathbb{J}$ as in (\ref{partition_K}), \ where
	$\mathbb{J}$ is a countable set, possibly infinite. The value%
	\[
	p_{r,v}(t):=%
	{\displaystyle\int\limits_{\mathcal{K}_{v}}}
	p(t,x,\mathcal{K}_{r})\frac{dx}{\mu_{\text{Haar}}\left(  \mathcal{K}%
		_{v}\right)  }=%
	{\displaystyle\int\limits_{\mathcal{K}_{v}}}
	Z_{0}\left(  x,t\right)  \ast1_{\mathcal{K}_{r}}\left(  x\right)  \frac
	{dx}{\mu_{\text{Haar}}\left(  \mathcal{K}_{v}\right)  },
	\]
	see (\ref{Probability}), gives the transition probability
	that a particle/walker starting at some position in $\mathcal{K}_{v}$ will be
	found in the set $\mathcal{K}_{r}$ at the time $t$. Notice that
	\[
	p_{r,v}(0)=\left\{
	\begin{array}
		[c]{ccc}%
		1 & \text{if} & r=v\\
		&  & \\
		0 & \text{if} & r\neq v.
	\end{array}
	\right.
	\]

	\begin{lemma}
		\label{Lemma_2}With the above notation, for any $v\in\mathbb{J}$ , for $t>0$,
		\[%
		{\displaystyle\sum\limits_{r\in\mathbb{J}}}
		p_{r,v}(t)=1.
		\]
		
	\end{lemma}
	
	\begin{proof}
		We fix \ a family $\left\{  \mathbb{J}_{n}\right\}  _{n\in\mathbb{N}}$ of
		subsets of $\mathbb{J}$ such that each $\mathbb{J}_{n}$ is a finite set,
		$\mathbb{J}_{n}\subset\mathbb{J}_{n+1}$ for every $n\in\mathbb{N}$, and
		$\cup_{n\in\mathbb{N}}\mathbb{J}_{n}=\mathbb{J}$. Notice that
		\begin{align*}
			Z_{0}\left(  x,t\right)  \ast%
			{\displaystyle\sum\limits_{r\in\mathbb{J}_{n}}}
			1_{\mathcal{K}_{r}}\left(  x\right)   &  =%
			{\displaystyle\int\limits_{\mathbb{Z}_{p}}}
			Z_{0}\left(  y,t\right)  \left(
			{\displaystyle\sum\limits_{r\in\mathbb{J}_{n}}}
			1_{\mathcal{K}_{r}}\left(  x-y\right)  \right)  dy\leq\Omega\left(  \left\vert
			x\right\vert _{p}\right)
			{\displaystyle\int\limits_{\mathbb{Z}_{p}}}
			Z_{0}\left(  y,t\right)  dy\\
			&  \leq \mathrm{e}^{-t\left(  1-\widehat{J}(0)\right)  }\Omega\left(  \left\vert
			x\right\vert _{p}\right)  =\Omega\left(  \left\vert x\right\vert _{p}\right)
			\text{, }%
		\end{align*}
		\ since $\widehat{J}(0)=1$, for any $t>0$, cf. Lemma \ref{Lemma_4A}. By the
		dominated convergence theorem,%

\begin{multline*}
\sum\limits_{r\in \mathbb{J}}p_{r,v}(t)=\frac{1}{\mu _{\text{Haar}}\left( 
\mathcal{K}_{v}\right) }\lim_{n\rightarrow \infty }\sum\limits_{r\in 
\mathbb{J}_{n}}\text{ }\int\limits_{\mathcal{K}_{v}}Z_{0}\left( x,t\right)
\ast 1_{\mathcal{K}_{r}}\left( x\right) dx \\
=\frac{1}{\mu _{\text{Haar}}\left( \mathcal{K}_{v}\right) }%
\lim_{n\rightarrow \infty }\int\limits_{\mathcal{K}_{v}}Z_{0}\left(
x,t\right) \ast \sum\limits_{r\in \mathbb{J}_{n}}1_{\mathcal{K}_{r}}\left(
x\right) dx \\
=\frac{1}{\mu _{\text{Haar}}\left( \mathcal{K}_{v}\right) }\int\limits_{%
\mathcal{K}_{v}}Z_{0}\left( x,t\right) \ast \lim_{n\rightarrow \infty
}\sum\limits_{r\in \mathbb{J}_{n}}1_{\mathcal{K}_{r}}\left( x\right) dx=%
\frac{1}{\mu _{\text{Haar}}\left( \mathcal{K}_{v}\right) }\int\limits_{%
\mathcal{K}_{v}}Z_{0}\left( x,t\right) \ast \Omega \left( \left\vert
x\right\vert _{p}\right) dx \\
=\frac{1}{\mu _{\text{Haar}}\left( \mathcal{K}_{v}\right) }\int\limits_{%
\mathcal{K}_{v}}\left\{ \int\limits_{\mathbb{Z}_{p}}Z_{0}\left(
x-y,t\right) dy\right\} dx=\frac{1}{\mu _{\text{Haar}}\left( \mathcal{K}%
_{v}\right) }\int\limits_{\mathcal{K}_{v}}\left\{ \int\limits_{\mathbb{Z}%
_{p}}Z_{0}\left( u,t\right) du\right\} dx \\
=\frac{\mathrm{e}^{-t\left( 1-\widehat{J}(0)\right) }\int_{\mathcal{K}_{v}}dx}{\mu _{%
\text{Haar}}\left( \mathcal{K}_{v}\right) }=1.
\end{multline*}
	\end{proof}
	
	We now show that $\left[  p_{r,v}(t)\right]  _{r,v\in\mathbb{J}}$ satisfies
	the classical Chapman-Kolmogorov equation, which justifies calling it the
	transition matrix of a CTMC in Theorem \ref{Theorem_2} below. We first
	establish the semigroup property of the heat kernel $Z_{0}(x,t)$.
	
	\begin{lemma}
		\label{Lemma_CK_semigroup}For $t,s>0$,
		\[
		Z_{0}(\cdot,t)\ast Z_{0}(\cdot,s)=Z_{0}(\cdot,t+s)\text{, i.e., }%
		{\displaystyle\int\limits_{\mathbb{Z}_{p}}}
		Z_{0}(x-y,t)\,Z_{0}(y,s)\,dy=Z_{0}(x,t+s)\text{, }x\in\mathbb{Z}_{p}.
		\]
	\end{lemma}
	
\begin{proof}
Since $Z_{0}(x,t)=\mathcal{F}_{\xi \rightarrow x}^{-1}\left( \mathrm{e}%
^{-t\left( 1-\widehat{J}(\left\vert \xi \right\vert _{p})\right) }\right) $,
and the Fourier transform takes convolutions to pointwise products, 
\begin{multline*}
\mathcal{F}\left[ Z_{0}(\cdot ,t)\ast Z_{0}(\cdot ,s)\right] (\xi )=\mathrm{e%
}^{-t\left( 1-\widehat{J}(\left\vert \xi \right\vert _{p})\right) }\mathrm{e}%
^{-s\left( 1-\widehat{J}(\left\vert \xi \right\vert _{p})\right) }=\mathrm{e}%
^{-\left( t+s\right) \left( 1-\widehat{J}(\left\vert \xi \right\vert
_{p})\right) } \\
=\mathcal{F}\left[ Z_{0}(\cdot ,t+s)\right] (\xi ).
\end{multline*}%
Since $\mathcal{F}$ is a topological isomorphism of $\mathcal{D}^{\prime
}\left( \mathbb{Q}_{p}\right) $, the claim follows. 
\end{proof}
	
	\begin{proposition}
		\label{Prop_CK_continuum}For every $x\in\mathbb{Z}_{p}$, every Borel set
		$B\subset\mathbb{Z}_{p}$, and all $t,s\geq0$,
		\[
		p(t+s,x,B)=%
		{\displaystyle\int\limits_{\mathbb{Z}_{p}}}
		p(t,y,B)\,p(s,x,dy).
		\]
	\end{proposition}
	
	\begin{proof}
		For $t,s>0$, using Lemma \ref{Lemma_CK_semigroup} (with $u=x-w$) and
		integrating over $w\in B$,%
		\[
		p(t+s,x,B)=%
		{\displaystyle\int\limits_{B}}
		Z_{0}(x-w,t+s)\,dw=%
		{\displaystyle\int\limits_{B}}
		{\displaystyle\int\limits_{\mathbb{Z}_{p}}}
		Z_{0}(x-w-z,t)\,Z_{0}(z,s)\,dz\,dw.
		\]
		Since $Z_{0}\geq0$, Tonelli's theorem allows us to interchange the order of
		integration:%
		\[
		=%
		{\displaystyle\int\limits_{\mathbb{Z}_{p}}}
		Z_{0}(z,s)\left[
		{\displaystyle\int\limits_{B}}
		Z_{0}(x-w-z,t)\,dw\right]  dz=%
		{\displaystyle\int\limits_{\mathbb{Z}_{p}}}
		Z_{0}(z,s)\,p(t,x-z,B)\,dz.
		\]
		Substituting $y=x-z$, a Haar-measure-preserving involution of
		$\mathbb{Z}_{p}$, and using that $Z_{0}$ is a radial function (so
		$Z_{0}(x-y,s)=Z_{0}(y-x,s)$), we get%
		\[
		=%
		{\displaystyle\int\limits_{\mathbb{Z}_{p}}}
		p(t,y,B)\,Z_{0}(x-y,s)\,dy=%
		{\displaystyle\int\limits_{\mathbb{Z}_{p}}}
		p(t,y,B)\,p(s,x,dy).
		\]
		For $t=0$ or $s=0$, both sides reduce to $p(s,x,B)$ (respectively
		$p(t,x,B)$), using the convention $p(0,x,B)=1_{B}(x)$.
	\end{proof}
	
	\begin{theorem}
		\label{Theorem_CK_discrete}Assume, in addition to the standing hypotheses,
		that for every $r,v\in\mathbb{J}$ and $t>0$, the function $x\rightarrow
		p(t,x,\mathcal{K}_{r})$ is constant on $\mathcal{K}_{v}$, equal to
		$p_{r,v}(t)$. (This holds, in particular, for the congruence-class
		partitions $\mathcal{K}_{j}=j+p^{l}\mathbb{Z}_{p}$, $j\in G_{l}$, used in
		Section \ref{Sec_Num_simulations}, by Lemma \ref{Lemma_5A} and Lemma
		\ref{Lemma_7}.) Then $\left[  p_{r,v}(t)\right]  _{r,v\in\mathbb{J}}$
		satisfies the classical Chapman-Kolmogorov equation: for all $t,s\geq0$,
		\[
		p_{r,v}(t+s)=%
		{\displaystyle\sum\limits_{u\in\mathbb{J}}}
		p_{r,u}(t)\,p_{u,v}(s).
		\]
	\end{theorem}
	
	\begin{proof}
		By Proposition \ref{Prop_CK_continuum},%
		\[
		p_{r,v}(t+s)=\frac{1}{\mu_{\text{Haar}}\left(  \mathcal{K}_{v}\right)  }%
		{\displaystyle\int\limits_{\mathcal{K}_{v}}}
		p(t+s,x,\mathcal{K}_{r})\,dx=\frac{1}{\mu_{\text{Haar}}\left(
			\mathcal{K}_{v}\right)  }%
		{\displaystyle\int\limits_{\mathcal{K}_{v}}}
		{\displaystyle\int\limits_{\mathbb{Z}_{p}}}
		p(t,y,\mathcal{K}_{r})\,p(s,x,dy)\,dx.
		\]
		Splitting the inner integral over the partition $\left\{  \mathcal{K}%
		_{u}\right\}  _{u\in\mathbb{J}}$ and using the hypothesis (so that
		$p(t,y,\mathcal{K}_{r})=p_{r,u}(t)$ is constant for $y\in\mathcal{K}_{u}$,
		and can be pulled out of the integral),%
		\[
		=\frac{1}{\mu_{\text{Haar}}\left(  \mathcal{K}_{v}\right)  }%
		{\displaystyle\sum\limits_{u\in\mathbb{J}}}
		p_{r,u}(t)%
		{\displaystyle\int\limits_{\mathcal{K}_{v}}}
		p(s,x,\mathcal{K}_{u})\,dx.
		\]
		By the hypothesis again (with the roles of $t,r$ played by $s,u$),
		$x\rightarrow p(s,x,\mathcal{K}_{u})$ is constant, equal to $p_{u,v}(s)$,
		for $x\in\mathcal{K}_{v}$, so the last integral equals $\mu_{\text{Haar}%
		}\left(  \mathcal{K}_{v}\right)  p_{u,v}(s)$, and%
		\[
		p_{r,v}(t+s)=%
		{\displaystyle\sum\limits_{u\in\mathbb{J}}}
		p_{r,u}(t)\,p_{u,v}(s).
		\]
	\end{proof}
	
	\begin{remark}
		Without the local-constancy hypothesis of Theorem \ref{Theorem_CK_discrete}%
		, the same computation only gives the coarser identity%
		\[
		p_{r,v}(t+s)=\frac{1}{\mu_{\text{Haar}}\left(  \mathcal{K}_{v}\right)  }%
		{\displaystyle\sum\limits_{u\in\mathbb{J}}}
		{  \displaystyle \text{  }  \int\limits_{\mathcal{K}_{u}}}
		p(t,w,\mathcal{K}_{r})\,p(s,w,\mathcal{K}_{v})\,dw,
		\]
		which need not factor into $\sum_{u}p_{r,u}(t)p_{u,v}(s)$ for an arbitrary
		partition (\ref{partition_K})\ (this is the classical \textquotedblleft
		lumpability\textquotedblright\ question for coarse-grained Markov chains).
		Since every partition used in this paper is of the compatible
		congruence-class type, this subtlety does not affect any result stated
		below.
	\end{remark}
	
	\begin{theorem}
		\label{Theorem_2} $\left[  p_{r,v}(t)\right]  _{r,v\in\mathbb{J}}$ is the
		transition matrix of a CTMC. If $\widehat{J}(\left\vert u\right\vert _{p})<1$,
		for $u\in\mathbb{Q}_{p}\smallsetminus\mathbb{Z}_{p}$, then
		\[
		\lim_{t\rightarrow\infty}\left\vert p_{r,v}(t)-\mu_{\text{Haar}}%
		(\mathcal{K}_{r})\right\vert =0\text{,}%
		\]
		for $t>0$. Consequently, the CTMC admits a stationary distribution $\left[
		p_{r}^{\text{sta}}\right]  _{r\in\mathbb{J}}$, where%
		\[
		p_{r}^{\text{sta}}=\lim_{t\rightarrow\infty}p_{r,v}(t)=\mu_{\text{Haar}%
		}(\mathcal{K}_{r}).
		\]
		
	\end{theorem}
	
	\begin{proof}
		(i) By using that $\mathbb{Z}_{p}$ is an additive topological group, for a
		fixed $x\in\mathbb{Z}_{p}$, the mapping $y\rightarrow z=x-y$ is a topological
		and algebraic isomorphism preserving the Haar measure, then
		\begin{gather}
			\frac{1}{\mu_{\text{Haar}}\left(  \mathcal{K}_{v}\right)  }%
			{\displaystyle\int\limits_{\mathcal{K}_{v}}}
			\Omega\left(  \left\vert x\right\vert _{p}\right)  \ast1_{\mathcal{K}_{r}%
			}\left(  x\right)  dx=\frac{1}{\mu_{\text{Haar}}\left(  \mathcal{K}%
				_{v}\right)  }%
			{\displaystyle\int\limits_{\mathcal{K}_{v}}}
			{\displaystyle\int\limits_{\mathbb{Z}_{p}}}
			1_{\mathcal{K}_{r}}\left(  x-y\right)  dy\nonumber\\
			=\frac{1}{\mu_{\text{Haar}}\left(  \mathcal{K}_{v}\right)  }%
			{\displaystyle\int\limits_{\mathcal{K}_{v}}}
			{\displaystyle\int\limits_{\mathbb{Z}_{p}}}
			1_{\mathcal{K}_{r}}\left(  z\right)  dz=\frac{\mu_{\text{Haar}}\left(
				\mathcal{K}_{v}\right)  \mu_{\text{Haar}}\left(  \mathcal{K}_{r}\right)  }%
			{\mu_{\text{Haar}}\left(  \mathcal{K}_{v}\right)  }=\mu_{\text{Haar}%
			}(\mathcal{K}_{r}).\label{Identity}%
		\end{gather}
		Now, for $t>0$, we set%
		\[
		F(x;t):=1_{\mathcal{K}_{v}}\left(  x\right)  \left\{  \left\{  Z_{0}\left(
		x,t\right)  -\Omega\left(  \left\vert x\right\vert _{p}\right)  \right\}
		\ast1_{\mathcal{K}_{r}}\left(  x\right)  \right\}  .
		\]
		Since the convolution of a distribution ($Z_{0}\left(  x,t\right)
		-\Omega\left(  \left\vert x\right\vert _{p}\right)  $) and a test function
		($1_{\mathcal{K}_{r}}\left(  x\right)  $) is a locally function; and that the
		product of this function times the test function $1_{\mathcal{K}_{v}}\left(
		x\right)  $ is a test function, then $F(x)$ is a test function in $x$, for
		$t>0$. Therefore,%
		\[%
		{\displaystyle\int\limits_{\mathbb{Q}_{p}}}
		F(x,t)dx=\widehat{F}(0,t).
		\]
		On the other hand, by using (\ref{Identity}),
		\begin{align*}
			\left\vert p_{r,v}(t)-\mu_{\text{Haar}}(\mathcal{K}_{r})\right\vert  &
			=\frac{1}{\mu_{\text{Haar}}\left(  \mathcal{K}_{v}\right)  }\left\vert \text{
			}%
			{\displaystyle\int\limits_{\mathcal{K}_{v}}}
			\left\{  Z_{0}\left(  x,t\right)  -\Omega\left(  \left\vert x\right\vert
			_{p}\right)  \right\}  \ast1_{\mathcal{K}_{r}}\left(  x\right)  dx\right\vert
			\\
			&  =\frac{1}{\mu_{\text{Haar}}\left(  \mathcal{K}_{v}\right)  }\left\vert
			\text{ }%
			{\displaystyle\int\limits_{\mathbb{Q}_{p}}}
			F(x,t)dx\right\vert =\left\vert \widehat{F}(0,t)\right\vert .
		\end{align*}
		Then, we have to show that $\lim_{t\rightarrow\infty}\widehat{F}(0,t)=0.$
		
		By using $1=\widehat{J}(0)$ and $\widehat{J}(\left\vert \xi\right\vert
		_{p})=\widehat{J}(0)$ for $\left\vert \xi\right\vert _{p}\leq1$,%
		\begin{align*}
			\widehat{F}(\xi,t) &  =\widehat{1}_{\mathcal{K}_{v}}\left(  \xi\right)
			\ast\left\{  \mathrm{e}^{-t\left(  1-\widehat{J}(\left\vert \xi\right\vert
				_{p})\right)  }-\Omega\left(  \left\vert \xi\right\vert _{p}\right)  \right\}
			\widehat{1}_{\mathcal{K}_{r}}\left(  \xi\right)  \\
			&  =\widehat{1}_{\mathcal{K}_{v}}\left(  \xi\right)  \ast\left\{
			\begin{array}
				[c]{lll}%
				0 & \text{if } & \left\vert \xi\right\vert _{p}\leq1\\
				&  & \\
				\widehat{1}_{\mathcal{K}_{r}}\left(  \xi\right) \mathrm{e} ^{-t\left(  1-\widehat
					{J}(\left\vert \xi\right\vert _{p})\right)  } & \text{if } & \left\vert
				\xi\right\vert _{p}\geq p
			\end{array}
			\right.  \\
			&  =%
			{\displaystyle\int\limits_{\mathbb{Q}_{p}\smallsetminus\mathbb{Z}_{p}}}
			\widehat{1}_{\mathcal{K}_{v}}\left(  \xi-u\right)  \mathrm{e}^{-t\left(  1-\widehat
				{J}(\left\vert u\right\vert _{p})\right)  }\widehat{1}_{\mathcal{K}_{r}%
			}\left(  u\right)  du.
		\end{align*}

		Now, by the hypothesis $1-\widehat{J}(\left\vert u\right\vert _{p})>0$ for
		$u\in\mathbb{Q}_{p}\smallsetminus\mathbb{Z}_{p}$, $\mathrm{e}^{-t\left(  1-\widehat
			{J}(\left\vert u\right\vert _{p})\right)  }\leq1$ for $t>0$, and the fact that
		$\widehat{1}_{\mathcal{K}_{v}}\left(  -u\right)  \widehat{1}_{\mathcal{K}_{r}%
		}\left(  u\right)  $ is a test function, we can apply the dominated
		convergence theorem,
		\begin{align*}
			\lim_{t\rightarrow\infty}\widehat{F}(0,t)  & =\lim_{t\rightarrow\infty}%
			{\displaystyle\int\limits_{\mathbb{Q}_{p}\smallsetminus\mathbb{Z}_{p}}}
			\widehat{1}_{\mathcal{K}_{v}}\left(  -u\right)  \mathrm{e}^{-t\left(  1-\widehat
				{J}(\left\vert u\right\vert _{p})\right)  }\widehat{1}_{\mathcal{K}_{r}%
			}\left(  u\right)  du\\
			& =%
			{\displaystyle\int\limits_{\mathbb{Q}_{p}\smallsetminus\mathbb{Z}_{p}}}
			\widehat{1}_{\mathcal{K}_{v}}\left(  -u\right)  \left\{  \lim_{t\rightarrow
				\infty}\mathrm{e}^{-t\left(  1-\widehat{J}(\left\vert u\right\vert _{p})\right)
			}\right\}  \widehat{1}_{\mathcal{K}_{r}}\left(  u\right)  du=0.
		\end{align*}
		
	\end{proof}

	\section{Appendix D: $p$-adic Schr\"{o}dinger equations on the unit ball and
		quantum Markov chains}
	By performing a Wick rotation $t\rightarrow it$ in (\ref{Eq_Cauchy_problem_2}%
	), and taking $\Psi\left(  x,t\right)  =u\left(  x,it\right)  $, we obtain the
	Schr\"{o}dinger equations in the unit ball,%
	
	\begin{equation}
		\left\{
		\begin{array}
			[c]{l}%
			\mathrm{i}\frac{\partial}{\partial t}\Psi\left(  x,t\right)  =-J\left(  \left\vert
			x\right\vert _{p}\right)  \ast\Psi\left(  x,t\right)  +\Psi\left(  x,t\right)
			\text{, }x\in\mathbb{Z}_{p},t\geq0\\
			\\
			\Psi\left(  x,0\right)  =\psi_{0}\left(  x\right)  \in\mathcal{D}%
			(\mathbb{Z}_{p}).
		\end{array}
		\right.  \label{Schrodinger-Equation}%
	\end{equation}
	By using Fourier transforms, we obtain that the solution of
	(\ref{Schrodinger-Equation}) is given by%
	\begin{equation}
		\Psi\left(  x,t\right)  =\mathcal{F}_{\xi\rightarrow x}^{-1}(\mathrm{e}^{-\mathrm{i}\left(
			1-\widehat{J}\left(  \left\vert \xi\right\vert _{p}\right)  \right)  t}%
		)\ast\psi_{0}\left(  x\right)  \text{ in }\mathcal{D}^{\prime}(\mathbb{Q}%
		_{p})\text{,}\label{Useful-Formula}%
	\end{equation}
	where $\mathcal{F}_{\xi\rightarrow x}$ denotes the Fourier transform in
	$\mathcal{D}^{\prime}(\mathbb{Q}_{p})$. Notice that $\Psi\left(  x,t\right)  $
	is a locally constant function in $x$, for any $t\geq0$.
	
	Since
	\[
	\mathrm{e}^{-i\left(  1-\widehat{J}\left(  \left\vert \xi+\xi_{0}\right\vert
		_{p}\right)  \right)  t}=\mathrm{e}^{-i\left(  1-\widehat{J}\left(  \left\vert
		\xi\right\vert _{p}\right)  \right)  t}%
	\]
	for any $\xi_{0}\in\mathbb{Z}_{p}$, $\xi\in\mathbb{Q}_{p}$, $\mathcal{F}%
	_{\xi\rightarrow x}^{-1}(\mathrm{e}^{-\mathrm{i}\left(  1-\widehat{J}\left(  \left\vert
		\xi\right\vert _{p}\right)  \right)  t})$ is a distribution with support in
	$\mathbb{Z}_{p}$, see, e.g., \cite[Section VII.3]{V-V-Z}, \cite[Theorem
	4.9.3]{Alberio et al}, \cite[Theorem 1.161]{Zuniga-Textbook}. The explicit
	calculation of $\mathcal{F}_{\xi\rightarrow x}^{-1}(\mathrm{e}^{-\mathrm{i}\left(  1-\widehat
		{J}\left(  \left\vert \xi\right\vert _{p}\right)  \right)  t})$ can be carried
	out as in the case of the heat kernel, and\ one obtains that%
	\[
	\mathcal{F}_{\xi\rightarrow x}^{-1}(\mathrm{e} ^{-\mathrm{i}\left(  1-\widehat{J}\left(
		\left\vert \xi\right\vert _{p}\right)  \right)  t}))=Z_{0}(x,\mathrm{i}t)
	\]
	for $x\in\mathbb{Z}_{p}\smallsetminus\left\{  0\right\}  $, $t>0$.
	
	The formula (\ref{Useful-Formula}) is useful; however, we need to show that
	$\left\Vert \Psi\left(  \cdot,t\right)  \right\Vert _{2}=1$, for any
	$\left\Vert \psi_{0}\right\Vert _{2}=1$, and any $t\geq0$. To achieve this
	result, we need a different approach.
	
	\begin{lemma}
		\label{Lemma_3}(i) For $k\in\left\{  1,\ldots,p-1\right\}  $, $r\leq0$, and
		$b\in\mathbb{Q}_{p}/\mathbb{Z}_{p}$, $bp^{-r}\in\mathbb{Z}_{p}$,
		\[
		\boldsymbol{J}\Psi_{rbk}\left(  x\right)  =(\widehat{J}\left(  p^{1-r}\right)
		-1)\Psi_{rbk}\left(  x\right)  .
		\]
		\noindent(ii)%
		\[
		\boldsymbol{J}\Omega\left(  \left\vert x\right\vert _{p}\right)  =0\text{.}%
		\]
		
	\end{lemma}
	
	\begin{proof}
		(i) We recall that $\boldsymbol{J}$ is a pseudo-differential operator,%
		\[
		\boldsymbol{J}\Psi_{rbk}\left(  x\right)  =\mathcal{F}_{\xi\rightarrow x}%
		^{-1}\left(  \left(  \widehat{J}\left(  \left\vert \xi\right\vert _{p}\right)
		-1\right)  \mathcal{F}_{x\rightarrow\xi}\Psi_{rbk}\left(  x\right)  \right)  .
		\]
		Now%
		\[
		\widehat{\Psi}_{rbk}\left(  \xi\right)  =p^{\frac{r}{2}}\chi_{p}\left(
		p^{-r}b\xi\right)  \Omega\left(  \left\vert p^{-r}\xi+p^{-1}k\right\vert
		_{p}\right)  ,
		\]
		where $\Omega\left(  \left\vert p^{-r}\xi+p^{-1}k\right\vert _{p}\right)  =1$
		if and only if%
		\[
		\xi\in-p^{r-1}k+p^{r}\mathbb{Z}_{p}=p^{r-1}k+p^{r}\mathbb{Z}_{p}.
		\]
		Then, by the ultrametric property of $\left\vert \cdot\right\vert _{p}$,
		\[
		\left\vert \xi\right\vert _{p}=\left\vert p^{r-1}k+p^{r}m\right\vert
		_{p}=\left\vert p^{r-1}k\right\vert _{p}=p^{1-r},
		\]
		for any $m\in\mathbb{Z}_{p}$. Then,
		\[
		\mathcal{F}_{x\rightarrow\xi}\left\{  \boldsymbol{J}\left(  \Psi_{rbk}\left(
		x\right)  \right)  \right\}  =(\widehat{J}\left(  p^{1-r}\right)
		-1)\widehat{\Psi}_{rbk}\left(  \xi\right)  .
		\]

		(ii) The second part follows from
		\begin{gather*}
\boldsymbol{J}\Omega \left( \left\vert x\right\vert _{p}\right) =J\left(
\left\vert x\right\vert _{p}\right) \ast \Omega \left( \left\vert
x\right\vert _{p}\right) -\Omega \left( \left\vert x\right\vert _{p}\right) 
\\
=\Omega \left( \left\vert x\right\vert _{p}\right) \int\limits_{\mathbb{Z}%
_{p}}J\left( \left\vert x-y\right\vert _{p}\right) \Omega \left( \left\vert
y\right\vert _{p}\right) dy-\Omega \left( \left\vert x\right\vert
_{p}\right)  \\
=\Omega \left( \left\vert x\right\vert _{p}\right) \int\limits_{\mathbb{Z}%
_{p}}J\left( \left\vert y\right\vert _{p}\right) dy-\Omega \left( \left\vert
x\right\vert _{p}\right)  \\
=\Omega \left( \left\vert x\right\vert _{p}\right) \int\limits_{\mathbb{Z}%
_{p}}J\left( \left\vert y\right\vert _{p}\right) dy-\Omega \left( \left\vert
x\right\vert _{p}\right) =\Omega \left( \left\vert x\right\vert _{p}\right)
\left( \left\Vert J\right\Vert _{L^{1}\left( \mathbb{Z}_{p}\right)
}-1\right) .
\end{gather*}

	\end{proof}
	
	\begin{theorem}
		\label{Theorem_3}The initial value problem%
		\begin{equation}
			\left\{
			\begin{array}
				[c]{l}%
				\Psi(x,t)\in L^{2}\left(  \mathbb{Z}_{p}\right)  \text{ for }t\geq0\text{
					fixed}\\
				\\
				\Psi(x,t)\in C^{1}\left(  0,\infty\right)  \text{ for }x\in\mathbb{Z}%
				_{p}\text{ fixed}\\
				\\
				\mathrm{i}\frac{\partial\Psi(x,t)}{\partial t}=-\boldsymbol{J}\Psi(x,t)\\
				\\
				\left\Vert \Psi(\cdot,t)\right\Vert _{2}=1\text{ for any }t\geq0\\
				\\
				\Psi(x,0)=\psi_{0}(x)\in L^{2}\left(  \mathbb{Z}_{p}\right)  \text{, with
				}\left\Vert \psi_{0}\right\Vert _{2}=1,
			\end{array}
			\right.  \label{Cauchy_Problem_3}%
		\end{equation}
		has a unique solution of the form%
		\begin{equation}
			\Psi(x,t)=A_{0}\Omega\left(  \left\vert x\right\vert
			_{p}\right)  +%
			{\displaystyle\sum\limits_{k\in\left\{  1,\ldots,p-1\right\}  }}
			\text{ \ }%
			{\displaystyle\sum\limits_{\substack{b\in\mathbb{Q}_{p}/\mathbb{Z}%
						_{p}\\bp^{-r}\in\mathbb{Z}_{p}}}}
			\text{ \ }%
			{\displaystyle\sum\limits_{r\leq0}}
			\text{ }A_{_{rbk}}\mathrm{e}^{-\mathrm{i}\left(  J\left(  p^{1-r}\right)  -1\right)  t}%
			\Psi_{rbk}\left(  x\right)  {,}\label{Expansion_Psi}%
		\end{equation}
		where the $A_{0}$, $A_{rbk}\in\mathbb{C}$.
	\end{theorem}
	
	\begin{proof}
		The expansion (\ref{Expansion_Psi}) follows from the fact that $\left\{
		\Omega\left(  \left\vert x\right\vert _{p}\right)  \right\}  \cup\left\{
		\Psi_{rbk}\left(  x\right)  \right\}  $ is an orthonormal basis of
		$L^{2}\left(  \mathbb{Z}_{p}\right)  $, by using the separation of variables
		technique. Now,%
		\[
		\left\Vert \Psi(\cdot,t)\right\Vert _{2}^{2}=\left\vert A_{0}\right\vert ^{2}+%
		{\displaystyle\sum\limits_{k\in\left\{  1,\ldots,p-1\right\}  }}
		\text{ \ }%
		{\displaystyle\sum\limits_{\substack{b\in\mathbb{Q}_{p}/\mathbb{Z}%
					_{p}\\bp^{-r}\in\mathbb{Z}_{p}}}}
		\text{ \ }%
		{\displaystyle\sum\limits_{r\leq0}}
		\text{ }\left\vert A_{_{rbk}}\right\vert ^{2}=\left\Vert \psi_{0}\right\Vert
		_{2}^{2}=1.
		\]
		
	\end{proof}
	\subsection{CTQMCs}
	
	We now use Theorem \ref{Theorem_0A} to construct a family of CTQMCs. We use
	all the notation introduced in Subsections \ref{Section-CTQMC-I},
	\ref{Section-CTQMC-II},
	
	\begin{proposition}
		\label{Poposition_4}For $l\geq1$, set $\mathbb{J}=G_{l}$ and
		\[
		\pi_{r,v}\left(  t\right)  =\left\vert \left\langle p^{\frac{l}{2}}%
		\Omega\left(  p^{l}\left\vert x-r\right\vert _{p}\right)  ,Z_{0}(x,\mathrm{i}t)\ast
		p^{\frac{l}{2}}\Omega\left(  p^{l}\left\vert x-v\right\vert _{p}\right)
		\right\rangle \right\vert ^{2},
		\]
		for $r,v\in G_{l}$. Then, $\left[  \pi_{r,v}\left(  t\right)  \right]
		_{r,v\in G_{l}}$ is the transition matrix of a CTQMC.
	\end{proposition}
	
	\begin{proof}
		We use all the notation introduced in Subsection \ref{Section-CTQMC-I}. Set
        \[
		e_{v}=p^{\frac{l}{2}}\Omega\left(  p^{l}\left\vert x-v\right\vert
		_{p}\right),  \] for $v\in G_{l}$. Then $\left\{  p^{\frac{l}{2}}\Omega\left(
		p^{l}\left\vert x-v\right\vert _{p}\right)  \right\}  _{v\in G_{l}}$ is an
		orthonormal basis of the finite-dimensional Hilbert space $\mathcal{H}_{G_{l}%
		}=\mathcal{D}_{l}(\mathbb{Z}_{p})$ Furthermore,%
		\[
		\mathrm{e}^{\mathrm{i}t\boldsymbol{J}}p^{\frac{l}{2}}\Omega\left(  p^{l}\left\vert
		x-v\right\vert _{p}\right)  =Z_{0}(x,\mathrm{i}t)\ast p^{\frac{l}{2}}\Omega\left(
		p^{l}\left\vert x-v\right\vert _{p}\right)  \in\mathcal{D}_{l}(\mathbb{Z}%
		_{p}),
		\]
		cf. Lemma \ref{Lemma_5A}-(iii). This fact implies that $\mathrm{e}^{\mathrm{i}t\boldsymbol{J}%
		}\mathcal{D}_{l}(\mathbb{Z}_{p})\subset\mathcal{D}_{l}(\mathbb{Z}_{p})$. Now,
		the announced result follows from Theorem \ref{Theorem_0A}.
	\end{proof}
	
	\begin{remark}
		Using the wavefunctions $\Psi(x,t)$ and applying the techniques introduced in
		Subsection \ref{Section-CTQMC-II}, we can construct a very large class of
		CTQMCs. However, the explicit computation of the $p$-adic wavefunctions, even
		for simple models, it is very involved, see \cite{Zuniga-Mayes}. Then, for
		practical purposes, we need good discretizations (in the sense of numerical
		analysis) of the wavefunctions. Our next step is to discuss the connection
		between this work and \cite{Zuniga-QM-2}.
	\end{remark}
	
	For a positive integer $l$, we set
	\[
	\alpha_{r,v}\left(  t\right)  :=\left\langle p^{\frac{l}{2}}\Omega\left(
	p^{l}\left\vert x-r\right\vert _{p}\right)  ,Z_{0}(x,\mathrm{i}t)\ast p^{\frac{l}{2}%
	}\Omega\left(  p^{l}\left\vert x-v\right\vert _{p}\right)  \right\rangle ,
	\]
	for $r,v\in G_{l}$, and $t\geq0$. The $\alpha_{r,v}\left(  t\right)  $ are the
	(complex) transition amplitudes. Furthermore, $\pi_{r,v}\left(  t\right)
	=\left\vert \alpha_{r,v}\left(  t\right)  \right\vert ^{2}$.
	
	\begin{proposition}
		\label{Poposition_5}With the above notation,%
		\begin{equation}
			\alpha_{r,v}\left(  t\right)  =p^{-l}\left\{  \mathrm{e}^{\mathrm{i}t}%
			{\displaystyle\int\limits_{\left\vert \xi\right\vert _{p}\leq p^{l}}}
			\chi_{p}\left(  \left(  r-v\right)  \xi\right)  e^{-it\widehat{J}(\left\vert
				\xi\right\vert _{p})}d\xi\right\}  ,\label{Transition_amplitudes}%
		\end{equation}
		for $r,v\in G_{l}$, and $t\geq0$.
	\end{proposition}
	
	\begin{proof}
		By Theorem \ref{Theorem_3}, in Appendix D, 
		\[
		Z_{0}(x,\mathrm{i}t)\ast p^{\frac{l}{2}}\Omega\left(  p^{l}\left\vert x-v\right\vert
		_{p}\right)  \in L^{2}\left(  \mathbb{Z}_{p}\right)  \subset L^{2}\left(
		\mathbb{Q}_{p}\right).
		\]
		Using the fact that the Fourier transform preserves the inner product in  $L^{2}\left(
		\mathbb{Q}_{p}\right)  $, \ we have
		
		\begin{gather*}
			\left\langle p^{\frac{l}{2}}\Omega\left(  p^{l}\left\vert x-r\right\vert
			_{p}\right)  ,Z_{0}(x,\mathrm{i}t)\ast p^{\frac{l}{2}}\Omega\left(  p^{l}\left\vert
			x-v\right\vert _{p}\right)  \right\rangle \\
			=p^{-l}\left\langle \Omega\left(  p^{-l}\left\vert \xi\right\vert _{p}\right)
			\chi_{p}\left(  r\xi\right)  ,\mathrm{e}^{\mathrm{i}t\left(  \widehat{J}(\left\vert
				\xi\right\vert _{p})-1\right)  }\Omega\left(  p^{-l}\left\vert \xi\right\vert
			_{p}\right)  \chi_{p}\left(  v\xi\right)  \right\rangle \\
			=p^{-l}%
			{\displaystyle\int\limits_{\mathbb{Q}_{p}}}
			\Omega\left(  p^{-l}\left\vert \xi\right\vert _{p}\right)  \chi_{p}\left(
			r\xi\right)  \mathrm{e}^{-\mathrm{i}t\left(  \widehat{J}(\left\vert \xi\right\vert
				_{p})-1\right)  }\Omega\left(  p^{-l}\left\vert \xi\right\vert _{p}\right)
			\chi_{p}\left(  -v\xi\right)  d\xi\\
			=p^{-l}%
			{\displaystyle\int\limits_{\left\vert \xi\right\vert _{p}\leq p^{l}}}
			\chi_{p}\left(  \left(  r-v\right)  \xi\right)  \mathrm{e}^{-\mathrm{i}t\left(  \widehat
				{J}(\left\vert \xi\right\vert _{p})-1\right)  }d\xi.
		\end{gather*}

		We now summarize our results in the following theorem:
	\end{proof}
	
	\begin{theorem}
		\label{Theorem_5}With the notation introduced in Propositions
		\ref{Poposition_4}, \ref{Poposition_5}.%
		\[
		\left[  \pi_{r,v}\left(  t\right)  \right]  _{r,v\in G_{l}}=\left[
		p^{-2l}\left\vert \text{ }
		{\displaystyle\int\limits_{\left\vert \xi\right\vert _{p}\leq p^{l}}}
		\chi_{p}\left(  \left(  r-v\right)  \xi\right)  \mathrm{e}^{-\mathrm{i}t\widehat{J}(\left\vert
			\xi\right\vert _{p})}d\xi\right\vert ^{2}\right]  _{r,v\in G_{l}}%
		\]
		is the transition matrix of a CTQMC. 
	\end{theorem}
	
	In particular, the $\lim_{t\rightarrow\infty}\pi_{r,v}\left(  t\right)  $ does
	not exist, and the announced CTQMCs do not admit a stationary distribution,
	nor do they have absorbing states.

	\section{Appendix E: Discretization of $p$-adic heat and Schr\"{o}dinger
		equations}
	
	\subsection{Some technical results}
	\begin{lemma}
		\label{Lemma_7}We use $J$ and $G_{l}=\mathbb{Z}_{p}/p^{l}\mathbb{Z}_{p}$ \ as
		before. Then
		\begin{align*}
			\boldsymbol{J}\Omega\left(  p^{l}\left\vert x-K\right\vert _{p}\right)   &
			=J\left(  x\right)  \ast\Omega\left(  p^{l}\left\vert x-K\right\vert
			_{p}\right)  =p^{-l}%
			{\displaystyle\sum\limits_{I\in G_{l}\smallsetminus\left\{  K\right\}  }}
			J\left(  \left\vert I-K\right\vert _{p}\right)  \Omega\left(  p^{l}\left\vert
			x-I\right\vert _{p}\right)  +\\
			&  \left(  \text{ }%
			{\displaystyle\int\limits_{p^{l}\mathbb{Z}_{p}}}
			J\left(  \left\vert y\right\vert _{p}\right)  dy\right)  \Omega\left(
			p^{l}\left\vert x-K\right\vert _{p}\right)  ,
		\end{align*}
		for $K\in G_{l}$.
	\end{lemma}
	
	\begin{remark}
		Notice that%
		\[%
		{\displaystyle\int\limits_{\mathbb{Z}_{p}}}
		J\left(  \left\vert y\right\vert _{p}\right)  dy=%
		{\displaystyle\sum\limits_{I\in G_{l}}}
		\text{ \ }%
		{\displaystyle\int\limits_{I+p^{l}\mathbb{Z}_{p}}}
		J\left(  \left\vert y\right\vert _{p}\right)  dy=%
		{\displaystyle\sum\limits_{\substack{I\in G_{l}\\I\neq0}}}
		p^{-l}J\left(  \left\vert I\right\vert _{p}\right)  +%
		{\displaystyle\int\limits_{p^{l}\mathbb{Z}_{p}}}
		J\left(  \left\vert y\right\vert _{p}\right)  dy,
		\]
		i.e.,
		\[%
		{\displaystyle\int\limits_{p^{l}\mathbb{Z}_{p}}}
		J\left(  \left\vert y\right\vert _{p}\right)  dy=%
		{\displaystyle\int\limits_{\mathbb{Z}_{p}}}
		J\left(  \left\vert y\right\vert _{p}\right)  dy-%
		{\displaystyle\sum\limits_{\substack{I\in G_{l}\\I\neq0}}}
		p^{-l}J\left(  \left\vert I\right\vert _{p}\right)  .
		\]
		
	\end{remark}
	
	\begin{proof}
		By using that
		\[
		\Omega\left(  \left\vert x\right\vert _{p}\right)  =%
		{\displaystyle\sum\limits_{I\in G_{l}}}
		\Omega\left(  p^{l}\left\vert x-I\right\vert _{p}\right)  ,
		\]
		and $J\left(  \left\vert x\right\vert _{p}\right)  \Omega\left(
		p^{l}\left\vert x-I\right\vert _{p}\right)  =J\left(  \left\vert I\right\vert
		_{p}\right)  $, for $I\in G_{l}\smallsetminus\left\{  0\right\}  $, we have%
		\[
		J\left(  \left\vert x\right\vert _{p}\right)  =\Omega\left(  \left\vert
		x\right\vert _{p}\right)  J\left(  \left\vert x\right\vert _{p}\right)  =%
		{\displaystyle\sum\limits_{I\in G_{l}\smallsetminus\left\{  0\right\}  }}
		J\left(  \left\vert I\right\vert _{p}\right)  \Omega\left(  p^{l}\left\vert
		x-I\right\vert _{p}\right)  +J\left(  \left\vert x\right\vert _{p}\right)
		\Omega\left(  p^{l}\left\vert x\right\vert _{p}\right)  .
		\]
		From this formula, by using the identity,%
		\[
		\Omega\left(  p^{l}\left\vert x-I\right\vert _{p}\right)  \ast\Omega\left(
		p^{l}\left\vert x-K\right\vert _{p}\right)  =p^{-l}\Omega\left(
		p^{l}\left\vert x-\left(  I+K\right)  \right\vert _{p}\right)  ,
		\]
		we obtain%
		\begin{align*}
			J\left(  \left\vert x\right\vert _{p}\right)  \ast\Omega\left(  p^{l}%
			\left\vert x-K\right\vert _{p}\right)   &  =p^{-l}%
			{\displaystyle\sum\limits_{I\in G_{l}\smallsetminus\left\{  0\right\}  }}
			J\left(  \left\vert I\right\vert _{p}\right)  \Omega\left(  p^{l}\left\vert
			x-\left(  I+K\right)  \right\vert _{p}\right)  +\\
			&  \left\{  J\left(  \left\vert x\right\vert _{p}\right)  \Omega\left(
			p^{l}\left\vert x\right\vert _{p}\right)  \right\}  \ast\Omega\left(
			p^{l}\left\vert x-K\right\vert _{p}\right)  .
		\end{align*}

		Set%
		\begin{gather*}
			L\left(  x\right)  :=\left\{  J\left(  \left\vert x\right\vert _{p}\right)
			\Omega\left(  p^{l}\left\vert x\right\vert _{p}\right)  \right\}  \ast
			\Omega\left(  p^{l}\left\vert x-K\right\vert _{p}\right) \\
			=%
			{\displaystyle\int\limits_{\mathbb{Q}_{p}}}
			J\left(  \left\vert y\right\vert _{p}\right)  \Omega\left(  p^{l}\left\vert
			y\right\vert _{p}\right)  \Omega\left(  p^{l}\left\vert x-K-y\right\vert
			_{p}\right)  dy\\
			=%
			{\displaystyle\int\limits_{x-K+p^{l}\mathbb{Z}_{p}}}
			J\left(  \left\vert y\right\vert _{p}\right)  \Omega\left(  p^{l}\left\vert
			y\right\vert _{p}\right)  dy.
		\end{gather*}
		The calculation of the last integral involves two cases.
		
		\textbf{Case 1}: $x\in K+p^{l}\mathbb{Z}_{p}$
		
		In this case \ $x-K\in p^{l}\mathbb{Z}_{p}$, and
		\[
		L\left(  x\right)  =%
		{\displaystyle\int\limits_{p^{l}\mathbb{Z}_{p}}}
		J\left(  \left\vert y\right\vert _{p}\right)  \Omega\left(  p^{l}\left\vert
		y\right\vert _{p}\right)  dy=\left(  \text{ }%
		{\displaystyle\int\limits_{p^{l}\mathbb{Z}_{p}}}
		J\left(  \left\vert y\right\vert _{p}\right)  dy\right)  \Omega\left(
		p^{l}\left\vert x-K\right\vert _{p}\right)  .
		\]

		\textbf{Case 2}: $x\notin K+p^{l}\mathbb{Z}_{p}$
		
		In this case, $\left(  x-K+p^{l}\mathbb{Z}_{p}\right)  \cap p^{l}%
		\mathbb{Z}_{p}=\varnothing$, and
		\[
		L\left(  x\right)  =%
		{\displaystyle\int\limits_{x-K+p^{l}\mathbb{Z}_{p}}}
		J\left(  \left\vert y\right\vert _{p}\right)  \Omega\left(  p^{l}\left\vert
		y\right\vert _{p}\right)  dy=%
		{\displaystyle\int\limits_{\left(  x-K+p^{l}\mathbb{Z}_{p}\right)  \cap
				p^{l}\mathbb{Z}_{p}}}
		J\left(  \left\vert y\right\vert _{p}\right)  dy=0.
		\]

		In conclusion,
		\begin{align*}
			J\left(  \left\vert x\right\vert _{p}\right)  \ast\Omega\left(  p^{l}%
			\left\vert x-K\right\vert _{p}\right)   &  =p^{-l}%
			{\displaystyle\sum\limits_{I\in G_{l}\smallsetminus\left\{  0\right\}  }}
			J\left(  \left\vert I\right\vert _{p}\right)  \Omega\left(  p^{l}\left\vert
			x-\left(  I+K\right)  \right\vert _{p}\right)  +\\
			&  \left(
			{\displaystyle\int\limits_{p^{l}\mathbb{Z}_{p}}}
			J\left(  \left\vert y\right\vert _{p}\right)  dy\right)  \Omega\left(
			p^{l}\left\vert x-K\right\vert _{p}\right)  ,
		\end{align*}
		and since $G_{l}$ is an additive group, reindexing $I^{\prime}=I+K$ (so
		that $I^{\prime}$ ranges over $G_{l}\smallsetminus\left\{  K\right\}  $ as
		$I$ ranges over $G_{l}\smallsetminus\left\{  0\right\}  $, and $J\left(
		\left\vert I\right\vert _{p}\right)  =J\left(  \left\vert I^{\prime
		}-K\right\vert _{p}\right)  $),%
		\begin{align*}
			J\left(  x\right)  \ast\Omega\left(  p^{l}\left\vert x-K\right\vert
			_{p}\right)   &  =p^{-l}%
			{\displaystyle\sum\limits_{I\in G_{l}\smallsetminus\left\{  K\right\}  }}
			J\left(  \left\vert I-K\right\vert _{p}\right)  \Omega\left(  p^{l}\left\vert
			x-I\right\vert _{p}\right)  +\\
			&  \left(
			{\displaystyle\int\limits_{p^{l}\mathbb{Z}_{p}}}
			J\left(  \left\vert y\right\vert _{p}\right)  dy\right)  \Omega\left(
			p^{l}\left\vert x-K\right\vert _{p}\right)  .
		\end{align*}
		
	\end{proof}
	
	\begin{lemma}
		\label{Lemma_6A}The mapping $\boldsymbol{J}:\mathcal{D}_{l}\left(
		\mathbb{Z}_{p}\right)  \rightarrow\mathcal{D}_{l}\left(  \mathbb{Z}%
		_{p}\right)  $ is a well-defined operator. By identifying $\varphi
		\in\mathcal{D}_{l}\left(  \mathbb{Z}_{p}\right)  $ with column vector $\left[
		\varphi\left(  I\right)  \right]  _{I\in G_{l}}$, the restriction of the
		operator $\boldsymbol{J}$\ to $\mathcal{D}_{l}\left(  \mathbb{Z}_{p}\right)
		$\ is given by the matrix $\boldsymbol{J}^{\left(  l\right)  }=\left[
		J_{I,K}^{(l)}\right]  _{I,K\in G_{l}}$, where
		\begin{equation}
			J_{I,K}^{(l)}=\left\{
			\begin{array}
				[c]{lll}%
				p^{-l}J\left(  \left\vert I-K\right\vert _{p}\right)   & \text{if} & I\neq K\\
				&  & \\
				\left( \text{   }
				{\displaystyle\int\limits_{p^{l}\mathbb{Z}_{p}}}
				J\left(  \left\vert y\right\vert _{p}\right)  dy\right)  -1 & \text{if} & I=K.
			\end{array}
			\right.  \label{Matrix_J_l}%
		\end{equation}

	\end{lemma}
	
	\begin{remark}
		\label{Nota_approximation}We now consider $\boldsymbol{J}:L^{1}\left(
		\mathbb{Z}_{p}\right)  \rightarrow L^{1}\left(  \mathbb{Z}_{p}\right)  $, then
		$\boldsymbol{J}^{\left(  l\right)  }=\boldsymbol{J}\mid_{\mathcal{D}%
			_{l}\left(  \mathbb{Z}_{p}\right)  }$. For each $l\in\mathbb{N}$, there exists
		a linear bounded operator
		\[
		\boldsymbol{P}_{l}:L^{1}\left(  \mathbb{Z}_{p}\right)  \rightarrow
		\mathcal{D}_{l}\left(  \mathbb{Z}_{p}\right)  ,
		\]
		satisfying $\boldsymbol{P}_{l}=\boldsymbol{P}_{l}^{2}$. For $f\in L^{1}\left(
		\mathbb{Z}_{p}\right)  $, $\varphi_{l}=\boldsymbol{P}_{l}f$ is a good
		approximation of $f$ in $\mathcal{D}_{l}\left(  \mathbb{Z}_{p}\right)  $,
		i.e.,
		\[
		\left\Vert \varphi_{l}-f\right\Vert _{1}\rightarrow0\text{ as }l\rightarrow
		\infty.
		\]
		Then $\boldsymbol{\boldsymbol{\boldsymbol{P}_{l}}J=J}^{\left(  l\right)  }$,
		and
		\[
		\left\Vert \boldsymbol{\boldsymbol{P}_{l}J}f-\boldsymbol{J}f\right\Vert
		_{1}=\left\Vert \boldsymbol{J}\varphi_{l}-\boldsymbol{J}f\right\Vert _{1}%
		\leq\left\Vert \boldsymbol{J}\right\Vert \left\Vert \varphi_{l}-f\right\Vert
		_{1}\rightarrow0\text{ as }l\rightarrow\infty.
		\]
		Which means, that $\boldsymbol{J}^{\left(  l\right)  }\rightarrow
		\boldsymbol{J}$, as $l\rightarrow\infty$, in the operator norm. See
		\cite{Zuniga-QM-2} for an in-depth discussion.
	\end{remark}
	
	\begin{remark}
		Let $(\mathcal{G},V,E)$ be a simple graph; then, it does not contain any loops
		or multiple edges between the same pair of vertices, and all the edges are
		undirected. We take $V=G_{l}^{0}\subset G_{l}$ for some $l\in\mathbb{N}$
		fixed. This means that we identify each vertex of $\mathcal{G}$ with a
		$p$-adic integer from $G_{l}$. Let $A=\left[  A_{I,K}\right]  _{I,K\in
			G_{l}^{0}}$ be the adjacency matrix of $\mathcal{G}$. Notice that $A_{I,I}=0$
		since $\mathcal{G}$ has not loops, and $A_{I,K}=A_{K,I}$, for $I,K\in
		G_{l}^{0}$. For convenience's sake, we assume that $A_{I,K}=0$ if $I\notin
		G_{l}^{0}$ or $K\notin G_{l}^{0}$. Now, we set
		\[
		J_{I,K}^{(l)}=p^{-l}A_{I,K}\text{ if }I\neq K\text{, with }I,K\in G_{l}^{0},
		\]
		and
		\[
		J_{I,K}^{(l)}=0\text{ if }I\notin G_{l}^{0}\text{ or }K\notin G_{l}^{0}.
		\]
		We assume that $1=\int_{\mathbb{Z}_{p}}J\left(  \left\vert y\right\vert
		_{p}\right)  dy$, then%
		\begin{equation}
			1=%
			{\displaystyle\int\limits_{p^{l}\mathbb{Z}_{p}}}
			J\left(  \left\vert y\right\vert _{p}\right)  dy+%
			{\displaystyle\sum\limits_{\substack{I\in G_{l}\\I\neq0}}}
			\text{ \ }%
			{\displaystyle\int\limits_{I+p^{l}\mathbb{Z}_{p}}}
			J\left(  \left\vert y\right\vert _{p}\right)  dy=%
			{\displaystyle\int\limits_{p^{l}\mathbb{Z}_{p}}}
			J\left(  \left\vert y\right\vert _{p}\right)  dy+p^{-l}%
			{\displaystyle\sum\limits_{\substack{I\in G_{l}\\I\neq0}}}
			\text{ \ }J\left(  \left\vert I\right\vert _{p}\right)  .\label{Eq_10}%
		\end{equation}
		Now, for any $K\in G_{l}$, the mapping%
		\[%
		\begin{array}
			[c]{ccc}%
			G_{l} & \rightarrow & G_{l}\\
			I & \rightarrow & I-K
		\end{array}
		\]
		is a group isomorphism. Then%
		\begin{equation}
			p^{-l}%
			{\displaystyle\sum\limits_{\substack{I\in G_{l}\\I\neq0}}}
			\text{ \ }J\left(  \left\vert I\right\vert _{p}\right)  =p^{-l}%
			{\displaystyle\sum\limits_{\substack{I\in G_{l}\\I\neq K}}}
			\text{ \ }J\left(  \left\vert I-K\right\vert _{p}\right)  .\label{Eq_11}%
		\end{equation}

		The identification $p^{-l}J\left(  \left\vert I-K\right\vert _{p}\right)
		=p^{-l}A_{I,K}$ requires that $J\left(  \left\vert I-K\right\vert
		_{p}\right)  \in\left\{  0,1\right\}  $ for $I,K\in G_{l}^{0}$, $I\neq K$,
		i.e., that $J$ takes only the values $0$ or $1$ on $G_{l}^{0}\smallsetminus
		\left\{  0\right\}  $; equivalently, that the graph $\mathcal{G}$
		determined by $J$ is a simple (unweighted) graph rather than a weighted
		one. This holds, for instance, when $J\left(  \left\vert x\right\vert
		_{p}\right)  $ is the characteristic function of a union of spheres, but
		does not hold for a general Bessel-type kernel such as the ones used in
		Section \ref{Sec_Num_simulations}; in that more general case, $A_{I,K}%
		:=J\left(  \left\vert I-K\right\vert _{p}\right)  $ should be understood as
		the (nonnegative, symmetric) weighted adjacency matrix of a weighted graph,
		and $val(K)$ below as the corresponding weighted degree, with the rest of
		the construction unchanged. Now from (\ref{Eq_10})-(\ref{Eq_11}), taking $p^{-l}J\left(  \left\vert
		I-K\right\vert _{p}\right)  =p^{-l}A_{I,K}$, and using that
		\[
		val(K)=\sum_{\substack{I\in G_{l}^{0}\\I\neq K}}A_{I,K},
		\]
		where $val(K)$\ denotes the valence of the vertex $K$, we have
		\begin{align*}%
			{\displaystyle\int\limits_{p^{l}\mathbb{Z}_{p}}}
			J\left(  \left\vert y\right\vert _{p}\right)  dy &  =1-p^{-l}%
			{\displaystyle\sum\limits_{\substack{I\in G_{l}\\I\neq K}}}
			\text{ \ }J\left(  \left\vert I-K\right\vert _{p}\right)  =1-p^{-l}%
			{\displaystyle\sum\limits_{\substack{I\in G_{l}^{0}\\I\neq K}}}
			\text{ }A_{I,K}\\
			&  =1-p^{-l}val(K).
		\end{align*}
		In conclusion,%
		\begin{equation}
			J_{I,K}^{(l)}=\left\{
			\begin{array}
				[c]{lll}%
				p^{-l}A_{I,K} & \text{if} & I\neq K\\
				&  & \\
				-p^{-l}val(K) & \text{if} & I=K,
			\end{array}
			\right.  \label{Matrix_CTQW}%
		\end{equation}
		for $I,K\in G_{l}^{0}$. A multiple of this matrix is used in the construction
		of the classical CTQWs on graphs; see, e.g., \cite{Childs et al},
		\cite{Venegas-Andraca}. It is worth noting that in \cite{Zuniga-QM-2}, a direct construction of CTQWs on graphs was given, starting from an adjacency matrix.
	\end{remark}
	
	\subsection{Discretizations of differential equations}
	
	We now discuss the discretization of (\ref{Eq_Cauchy_problem_2}). Since
	$u\left(  \cdot,t\right)  \in\mathcal{C}(\mathbb{Z}_{p})$ for $t\geq0$, and
	$\mathcal{D}\left(  \mathbb{Z}_{p}\right)  $ is dense in $\mathcal{C}%
	(\mathbb{Z}_{p})$, then for $l$ sufficiently large, the function%
	\[
	u^{\left(  l\right)  }\left(  x,t\right)  =%
	{\displaystyle\sum\limits_{I\in G_{l}}}
	u_{I}^{\left(  l\right)  }\left(  t\right)  \Omega\left(  p^{l}\left\vert
	x-I\right\vert _{2}\right)  \text{, with }u_{I}^{\left(  l\right)  }\left(
	t\right)  \in\mathcal{C}^{1}\left[  0,\infty\right)
	\]
	is a good approximation in the norm $\left\Vert \cdot\right\Vert _{\infty}$
	for each fixed $t\geq0$. By identifying the function $u^{\left(  l\right)
	}\left(  x,t\right)  $ with the column vector $\left[  u_{I}^{\left(
		l\right)  }\right]  $, we have the following discretization of
	(\ref{Eq_Cauchy_problem_2}):%
	\begin{equation}
		\left\{
		\begin{array}
			[c]{l}%
			\left[  u_{I}^{\left(  l\right)  }\left(  t\right)  \right]  \in
			\mathbb{R}^{p^{l}}\text{, }t\geq0;\text{ }u_{I}^{\left(  l\right)  }\left(
			t\right)  \in C^{1}(\mathbb{R}_{+})\text{, }I\in G_{l}\\
			\\
			\frac{\partial}{\partial t}\left[  u_{I}^{\left(  l\right)  }\left(  t\right)
			\right]  =\boldsymbol{J}^{\left(  l\right)  }\left[  u_{I}^{\left(  l\right)
			}\left(  t\right)  \right]  \text{, }I\in G_{l},t\geq0\\
			\\
			\left[  u_{I}^{\left(  l\right)  }\left(  0\right)  \right]  =\left[
			u_{I}^{\left(  \text{init}\right)  }\right]  \in\mathbb{R}^{p^{l}}.
		\end{array}
		\right.  \label{Discrte-Heat-Eq}%
	\end{equation}
	In similar way, the discretization of (\ref{Cauchy_Problem_3}) is given by%
	\begin{equation}
		\left\{
		\begin{array}
			[c]{l}%
			\left[  \Psi_{I}^{\left(  l\right)  }\left(  t\right)  \right]  \in
			\mathbb{C}^{p^{l}}\text{, }t\geq0;\text{ }\Psi_{I}^{\left(  l\right)  }\left(
			t\right)  \in C^{1}(\mathbb{R}_{+})\text{, }I\in G_{l}\\
			\\
			\mathrm{i}\frac{\partial}{\partial t}\left[  \Psi_{I}^{\left(  l\right)  }\left(
			t\right)  \right]  =-\boldsymbol{J}^{\left(  l\right)  }\left[  \Psi
			_{I}^{\left(  l\right)  }\left(  t\right)  \right]  \text{, }I\in G_{l}%
			,t\geq0\\
			\\
			\left[  \Psi_{I}^{\left(  l\right)  }\left(  0\right)  \right]  =\left[
			\Psi_{I}^{\left(  \text{init}\right)  }\right]  \in\mathbb{C}^{p^{l}},
		\end{array}
		\right.  \label{Discrete-Schr-Eq}%
	\end{equation}
	see \cite{Zuniga-QM-2}, and the references therein, for a further discussion.
	
	The construction of CTQWs on graphs introduced in \cite{Farhi-Gutman}, see,
	also \cite{Childs et al}, \cite{Venegas-Andraca} is based on the equations
	(\ref{Discrte-Heat-Eq})-(\ref{Discrete-Schr-Eq}), when $\boldsymbol{J}%
	^{\left(  l\right)  }$ is as in (\ref{Matrix_CTQW}).
	
	\subsection{Quantum networks}
	
	From now on, $\mathcal{H}_{l}$ denotes the Hilbert space $\mathbb{C}^{p^{l}}$,
	with inner product $\left\langle \cdot,\cdot\right\rangle _{l}$, and canonical
	basis as $\left\{  e_{I}\right\}  _{I\in G_{l}}$. Since $\boldsymbol{J}%
	^{\left(  l\right)  }$ is a Hermitian matrix so $\exp(it\boldsymbol{J}%
	^{\left(  l\right)  })$ is unitary matrix. We define the transition
	probability $\pi_{I,J}\left(  t\right)  $ from $J$ to $I$ as
	\begin{equation}
		\pi_{I,J}\left(  t\right)  =\left\vert \left\langle e_{I}\right\vert
		\mathrm{e}^{\mathrm{i}t\boldsymbol{J}^{\left(  l\right)  }}\left\vert e_{J}\right\rangle
		_{l}\right\vert ^{2}\text{, for }J,I\in G_{l}. \label{Prob_Transition_I}%
	\end{equation}
	Then%
	\[%
	{\displaystyle\sum\limits_{I\in G_{l}}}
	\pi_{I,J}\left(  t\right)  =%
	{\displaystyle\sum\limits_{I\in G_{l}}}
	\left\vert \left\langle e_{I}\right\vert \mathrm{e}^{\mathrm{i}t\boldsymbol{J}^{\left(
			l\right)  }}\left\vert e_{J}\right\rangle \right\vert ^{2}=1,
	\]
	see \cite{Zuniga-QM-2}, for a further details. We call to the continuous-time
	Markov chain on $G_{l}$ determined by the transition probabilities $\left[
	\pi_{I,J}\left(  t\right)  \right]  _{I,J\in G_{l}}$, the quantum network
	associated with the discrete $p$-adic Schr\"{o}dinger equation
	(\ref{Cauchy_Problem_3}). When the matrix $\boldsymbol{J}^{\left(  l\right)
	}$ has the form (\ref{Matrix_CTQW}), we recover the CTQWs introduced in
	\cite{Farhi-Gutman}, see, also \cite{Childs et al}, \cite{Venegas-Andraca}.
	\subsection{Additional remarks}
	We now show that the CTQMC defined in (\ref{Prob_Transition_I}) agrees with
	the one given in Theorem \ref{Theorem_5}. First, we notice that $\left(
	\mathcal{D}_{l}\left(  \mathbb{Z}_{p}\right)  ,\left\langle \cdot,\cdot\right\rangle \right)  $ is isometric to $\left(  \mathcal{H}%
	_{l},\left\langle \cdot,\cdot\right\rangle _{l}\right)  $. The isometry is
	defined by%
	\[
	p^{\frac{l}{2}}\Omega\left(  p^{l}\left\vert x-r\right\vert _{p}\right)
	\rightarrow e_{r}\text{, for }r\in G_{l}.
	\]
	Using the results give in Remark \ref{Nota_approximation}, and the fact that
	$L^{2}(\mathbb{Z}_{p})\subset L^{1}(\mathbb{Z}_{p})$, \ there is a projection%
	\[
	\boldsymbol{P}_{l}:L^{2}\left(  \mathbb{Z}_{p}\right)  \rightarrow
	\mathcal{D}_{l}\left(  \mathbb{Z}_{p}\right)  ,
	\]
	satisfying $\boldsymbol{P}_{l}=\boldsymbol{P}_{l}^{2}$. For $f\in L^{2}\left(
	\mathbb{Z}_{p}\right)  $, $\varphi_{l}=\boldsymbol{P}_{l}f$ is a good
	approximation of $f$ in $\mathcal{D}_{l}\left(  \mathbb{Z}_{p}\right)  $. The
	CTQMCs constructed in Theorem \ref{Theorem_5} use the semigroup%
	\[
	\boldsymbol{P}_{l}\mathrm{e}^{\mathrm{i}t\boldsymbol{J}}:\mathcal{D}_{l}\left(  \mathbb{Z}%
	_{p}\right)  \rightarrow\mathcal{D}_{l}\left(  \mathbb{Z}_{p}\right)  ,
	\]
	while the ones introduced in this section use the semigroup%
	\[
	\mathrm{e}^{\mathrm{i}t\boldsymbol{J}^{\left(  l\right)  }}:\mathcal{H}_{l}\rightarrow
	\mathcal{H}_{l}.
	\]
	After identifying $\mathcal{H}_{l}$ and $\mathcal{D}_{l}\left(  \mathbb{Z}%
	_{p}\right)  $, and using the fact that $\boldsymbol{J}^{\left(  l\right)
	}=\boldsymbol{P}_{l}\boldsymbol{J}$ $:\mathcal{D}_{l}\left(  \mathbb{Z}%
	_{p}\right)  \rightarrow\mathcal{D}_{l}\left(  \mathbb{Z}_{p}\right)  $, we
	have%
	\[
	\boldsymbol{P}_{l}\mathrm{e}^{\mathrm{i}t\boldsymbol{J}}=\mathrm{e}^{\mathrm{i}t\boldsymbol{P}_{l}\boldsymbol{J}%
	}=\mathrm{e}^{\mathrm{i}t\boldsymbol{J}^{\left(  l\right)  }}.
	\]
	This last equality follows from $\boldsymbol{P}_{l}\boldsymbol{J}^{k}=\left(
	\boldsymbol{P}_{l}\boldsymbol{J}\right)  ^{k}$, $k\in\mathbb{N}$, which in
	turn is a consequence of $\boldsymbol{P}_{l}=\boldsymbol{P}_{l}^{2}$, by using
	that%
	\[
	\mathrm{e}^{\mathrm{i}t\boldsymbol{P}_{l}\boldsymbol{J}}=%
	{\displaystyle\sum\limits_{k=0}^{\infty}}
	\frac{\left(  it\right)  ^{k}}{k!}\left(  \boldsymbol{P}_{l}\boldsymbol{J}%
	\right)  ^{k}=%
	{\displaystyle\sum\limits_{k=0}^{\infty}}
	\frac{\left(  it\right)  ^{k}}{k!}\boldsymbol{P}_{l}\boldsymbol{J}%
	^{k}=\boldsymbol{P}_{l}%
	{\displaystyle\sum\limits_{k=0}^{\infty}}
	\frac{\left(  it\right)  ^{k}}{k!}\boldsymbol{J}^{k}=\boldsymbol{P}%
	_{l}\mathrm{e}^{\mathrm{i}t\boldsymbol{J}}.
	\]
	In the last calculation, we used the fact that all the operators involved are bounded.
	
	\begin{remark}
		\label{Nota_final}The above reasoning also shows that in the case%
		\[
		\boldsymbol{P}_{l}:L^{1}\left(  \mathbb{Z}_{p}\right)  \rightarrow
		\mathcal{D}_{l}\left(  \mathbb{Z}_{p}\right)  ,
		\]
		with $\boldsymbol{P}_{l}=\boldsymbol{P}_{l}^{2}$, we have%
		\[
		\boldsymbol{P}_{l}\mathrm{e}^{t\boldsymbol{J}}=\mathrm{e}^{t\boldsymbol{P}_{l}\boldsymbol{J}%
		}=\mathrm{e}^{t\boldsymbol{J}^{\left(  l\right)  }}:\mathcal{D}_{l}\left(
		\mathbb{Z}_{p}\right)  \rightarrow\mathcal{D}_{l}\left(  \mathbb{Z}%
		_{p}\right)  .
		\]
		This implies that in the case $\mathbb{J}=G_{l}$, and $\mathcal{K}_{r}%
		=B_{-l}(r)=r+p^{l}\mathbb{Z}_{p}$,%
		\begin{align*}
			p_{r,v}(t)  & =p^{l}%
			{\displaystyle\int\limits_{v+p^{l}\mathbb{Z}_{p}}}
			Z_{0}\left(  x,t\right)  \ast\Omega\left(  p^{l}\left\vert x-r\right\vert
			_{p}\right)  dx=p^{l}%
			{\displaystyle\int\limits_{v+p^{l}\mathbb{Z}_{p}}}
			\mathrm{e}^{t\boldsymbol{J}^{\left(  l\right)  }}\Omega\left(  p^{l}\left\vert
			x-r\right\vert _{p}\right)  dx\\
			& =%
			{\displaystyle\int\limits_{\mathbb{Z}_{p}}}
			p^{\frac{l}{2}}\Omega\left(  p^{l}\left\vert x-v\right\vert _{p}\right)
			\mathrm{e}^{t\boldsymbol{J}^{\left(  l\right)  }}p^{\frac{l}{2}}\Omega\left(
			p^{l}\left\vert x-r\right\vert _{p}\right)  dx,
		\end{align*}
		for $r,v\in G_{l}$. 
	\end{remark}

\section{Discussion}

By $p$-adic QM, we mean quantum mechanics, in the sense of Dirac-von
Neumann, with quantum states from a complex Hilbert space of the form $%
L^{2}(B)$, where $B$ is an open compact subset\ of $\mathbb{Q}_{p}^{N}$. In
this framework, the time is a real variable, but the positions are vectors
from $\mathbb{Q}_{p}^{N}$. The field of $p$-adic number $\mathbb{Q}_{p}$ is
a paramount prototype of a \textquotedblleft discrete
space:\textquotedblright\ any continuous trajectory from $\mathbb{R}$ in to $%
\mathbb{Q}_{p}$ is constant. The motion in $\mathbb{Q}_{p}$ consists of a
sequence of jumps. This assertion is also valid in any subset $B\subset 
\mathbb{Q}_{p}^{N}$; consequently, there are world lines in $B$, and so $p$%
-adic QM\ is not compatible with the theory of relativity.

$p$-Adic QM is a model of the Dirac-von Neumann formalism on a discrete
space. Here, it is relevant to mention that there are several different
types of $p$-adic QM. In some models, time is a $p$-adic variable; this type
of model is outside the scope of the Dirac-von Neumann formalism. $p$-Adic
QM with $p$-adic time was introduced in the 1980s by Vladimirov and Volovich \cite{V-V-QM3}; the literature on it is extensive. Since the 1980s, a central problem
has been to determine whether $p$-adic QM has a physical meaning. 

In \cite{Zuniga-Mayes}, the author and Mayes studied a $p$-adic version of the infinite
potential well. This model describes the confinement of a particle in a $p$-adic ball. They rigorously solved the Cauchy problem for the Schr\"{o}dinger equation and determined the stationary solutions. By dividing a $p$%
-adic ball into a finite number of sub-balls and using the wavefunctions of the infinite potential well, they constructed a continuous-time quantum walk
(CTQW) on a fully connected graph, where each vertex corresponds to a
sub-ball in the partition of the original ball. This approach shows that one
can associate a CTQW with any $p$-adic Schr\"{o}dinger equation. But this
approach requires explicit formulas for solving the Cauchy problem for the
Schr\"{o}dinger equation, which is a very involved task and probably
impossible for general equations, and, on top of that, the CTQWs produced do
not include the ones used in quantum computing, \cite{Farhi-Gutman}-\cite{Childs et al}. Here, we recall that CTQWs are a particular case of  continuous-time quantum Markov
chains (CTQMCs). In the literature, CTQWs are usually constructed using adjacency matrices, whereas CTQMCs do not require them.

In \cite{Zuniga-QM-2}, the author showed that a large class of $2$-adic Schr\"{o}dinger
equations is the scaling limit of CTQMCs. As a practical result, we construct new types CTQWs on graphs using two symmetric
matrices. This construction includes, as a particular case, the CTQWs
constructed using adjacency matrices. This paper did not include simulations
of the CTQMCs, which are essential for understanding their dynamics.

The goal of this paper is to provide a draft of the theory of CTMCs
associated with the $p$-adic Schr\"{o}dinger equations, which contains  as a
particular case of the CTQWs constructed using adjacency matrices. Here, we
study the CTMCs associated with certain $p$-adic heat equations, and the
CTQMCs related to the Schr\"{o}dinger equations obtained by Wick rotation.
We solve the initial-value problems associated with the mentioned equations
and provide numerical simulations for the corresponding CTQMCs. Behind the
results presented in this paper is the conjecture, already formulated in \cite{Zuniga-QM-2},
that general $p$-adic heat equations and the corresponding Schr\"{o}dinger
equations are scaling limits of CTMCs, respectively CTQMCs.

The connection between $p$-adic Schr\"{o}dinger equations and CTQMCs
demonstrates that $p$-adic QM has physical meaning. The $p$-adic Schr\"{o}dinger operators are non-local, so from the beginning, "spooky actions at a
distance" are allowed. Then, interpreting the experimental confirmation of
the violation of Bell's inequalities as "the universe is non-locally real,"
which means that the need to choose between realism and non-locality implies
that $p$-adic QM admits realism. The implications of this assertion are
explored by the author in \cite{Zuniga-ultimo}; see also \cite{Zuniga-AP}, \cite{Zuniga-PhA} for the discussion of other $p$-adic models in QM.

	Finally, it is worth noting that the developed framework remains valid if we replace the field of $p$-adic numbers with a more general ultrametric field, while keeping the state space as a Hilbert space of complex-valued functions. If the wavefunctions are not complex-valued, we cannot use the Dirac-von Neumann formalism. In particular, the study of CTQWs with $p$-adic transition probabilities is out of the scope of this work.

\end{document}